\def\year{2021}\relax
\pdfoutput=1
\documentclass[letterpaper]{article} 
\usepackage{aaai21}  
\usepackage{times}  
\usepackage{helvet} 
\usepackage{courier}  
\usepackage[hyphens]{url}  
\usepackage{graphicx} 
\usepackage{natbib}  
\usepackage{caption} 
\usepackage{float}
\usepackage[utf8]{inputenc} 
\usepackage[T1]{fontenc}    
\usepackage[hidelinks]{hyperref}       
\usepackage{url}            
\usepackage{booktabs}       
\usepackage{amsfonts}       
\usepackage{nicefrac}       
\usepackage{microtype}      
\usepackage{times}
\usepackage{amsmath}
\usepackage{amssymb}
\usepackage{amsthm}
\usepackage{microtype}
\usepackage{enumitem}
\usepackage{dsfont}
\usepackage[capitalise,noabbrev]{cleveref}
\usepackage{xcolor}
\usepackage{etoolbox}
\usepackage{thm-restate}
\usepackage[linesnumbered,ruled,lined,noend]{algorithm2e}
\usepackage{amsmath}
\usepackage{mathtools}
\usepackage{amsthm}
\usepackage{mleftright}
\usepackage{stmaryrd}
\usepackage{bm}
\usepackage{numprint}
\usepackage{tikz}
\usetikzlibrary{arrows.meta}
\usepackage{cleveref}
\usepackage{tcolorbox}
\usepackage{wrapfig}
\usepackage[]{minitoc}
\usepgflibrary{shapes.geometric}

\ifcsmacro{assumption}{}{
	
}
\ifcsmacro{corollary}{}{
	\newtheorem{corollary}{Corollary}[]
}
\ifcsmacro{definition}{}{
	\newtheorem{definition}{Definition}[]
}
\ifcsmacro{example}{}{
	
}
\ifcsmacro{fact}{}{
	
}
\ifcsmacro{informal_theorem}{}{
	
}
\ifcsmacro{lemma}{}{
	\newtheorem{lemma}{Lemma}[]
}
\ifcsmacro{proposition}{}{
	
}
\ifcsmacro{remark}{}{
	
}
\ifcsmacro{theorem}{}{
	\newtheorem{theorem}{Theorem}[]
}
\ifcsmacro{observation}{}{
	\newtheorem{observation}{Observation}[]
}
\mathtoolsset{centercolon}

\newcommand{\trans}{^{\!\top}}

\makeatletter
\patchcmd\algocf@Vline{\vrule}{\vrule \kern-0.4pt}{}{}
\patchcmd\algocf@Vsline{\vrule}{\vrule \kern-0.4pt}{}{}
\makeatother

\SetKwComment{Hline}{}{\vspace{-3mm}\textcolor{gray}{\hrule}\vspace{1mm}}
\definecolor{darkgrey}{gray}{0.3}
\definecolor{commentcolor}{gray}{0.5}
\SetKwComment{Comment}{\color{commentcolor}[$\triangleright$\ }{}
\SetCommentSty{}
\SetNlSty{}{\color{darkgrey}}{}
\crefname{algocf}{Algorithm}{Algorithms}
\SetKwProg{Fn}{function}{}{}
\SetKwProg{Subr}{subroutine}{}{}

\newcommand{\defeq}{\mathrel{:\mkern-0.25mu=}}

\newcommand{\cX}{\mathcal{X}}

\newcommand{\pure}{\Pi}

\newcommand{\cJ}{\mathcal{J}}
\newcommand{\cK}{\mathcal{K}}
\newcommand{\cC}{\mathcal{C}}

\newcommand{\cR}{\mathcal{R}}
\newcommand{\bbR}{\mathbb{R}}
\newcommand{\bbE}{\mathbb{E}}

\renewcommand{\next}[1]{\cC_{#1}}

\newcommand{\atOrAbove}[2]{#1\ \preceq\ #2}

\newcommand{\atOrBelow}[2]{#1\ \succeq\ #2}

\renewcommand{\vec}[1]{\bm{#1}}
\newcommand{\mat}[1]{\bm{#1}}
\DeclareMathOperator{\supp}{supp}
\DeclareMathOperator{\Span}{span}
\DeclareMathOperator{\img}{Im}
\DeclareMathOperator{\lin}{dir}

\DeclareMathOperator{\co}{co}
\DeclareMathOperator*{\argmax}{arg\,max}
\DeclareMathOperator*{\argmin}{arg\,min}

\newcommand{\pn}[2]{\|#1\|_{#2}}
\newcommand{\dn}[2]{\|#1\|_{\ast,#2}}

\newcommand{\dnlarge}[2]{\mleft\|#1\mright\|_{\ast,#2}}

\newcommand{\regu}{\varphi} 
\newcommand{\emptyseq}{\varnothing}
\newcommand{\substrategy}[1]{\xbar^{t}_{#1}} 
\newcommand{\xbar}{{\vec{x}}}

\newcommand{\Rpp}{\bbR_{>0}}
\newcommand{\Rp}{\bbR_{\ge 0}}

\newcommand{\seqf}{Q}

\newcommand{\terminalnode}{\diamond}
\renewcommand{\div}[2]{D_\regu(#1 \,\|\, #2)}

\LetLtxMacro{\baseproof}{\proof}
\LetLtxMacro{\endbaseproof}{\endproof}

\tcbuselibrary{skins,breakable}
\NewEnviron{shadedproof}[1][]{%
    \begin{tcolorbox}[enhanced jigsaw,breakable,colback=gray!15,sharp corners,frame hidden]
        \begin{baseproof}[#1]
            \BODY
        \end{baseproof}
    \end{tcolorbox}
}
\newcommand*\circled[1]{\tikz[baseline=(char.base)]{
            \node[shape=circle,draw,inner sep=1pt] (char) {\scriptsize #1};}}

\title{Bandit Linear Optimization for Sequential Decision Making\\and Extensive-Form Games}
\author{
Gabriele Farina,\textsuperscript{\rm 1}
Robin Schmucker,\textsuperscript{\rm 2}
Tuomas Sandholm\textsuperscript{\rm 1,\rm 2,\rm 3,\rm 4,\rm 5}\\
}
\affiliations {
\textsuperscript{\rm 1}Computer Science Department, Carnegie Mellon University,
\textsuperscript{\rm 2}Machine Learning Department, Carnegie Mellon University\\
\textsuperscript{\rm 3}Strategic Machine, Inc.,
\textsuperscript{\rm 4}Strategy Robot, Inc.,
\textsuperscript{\rm 5}Optimized Markets, Inc.\\
gfarina@cs.cmu.edu, rschmuck@cs.cmu.edu, sandholm@cs.cmu.edu
}
\makeatletter
\newcommand{\runinsec}{%
    \@ifstar{\runinsec@nostep}{\runinsec@step}}
\newcommand{\runinsec@step}[1]{%
    \refstepcounter{subsection}\runinsec@nostep{\thesubsection\quad#1}}
\newcommand{\runinsec@nostep}[1]{\noindent\textbf{#1}\qquad}
\makeatother

\usepackage[switch]{lineno}
\makeatletter
    \AtBeginDocument{%
      \@ifpackageloaded{amsmath}{%
        \newcommand*\patchAmsMathEnvironmentForLineno[1]{%
          \expandafter\let\csname old#1\expandafter\endcsname\csname #1\endcsname
          \expandafter\let\csname oldend#1\expandafter\endcsname\csname end#1\endcsname
          \renewenvironment{#1}%
                           {\linenomath\csname old#1\endcsname}%
                           {\csname oldend#1\endcsname\endlinenomath}%
        }%
        \newcommand*\patchBothAmsMathEnvironmentsForLineno[1]{%
          \patchAmsMathEnvironmentForLineno{#1}%
          \patchAmsMathEnvironmentForLineno{#1*}%
        }%
        \patchBothAmsMathEnvironmentsForLineno{equation}%
        \patchBothAmsMathEnvironmentsForLineno{align}%
        \patchBothAmsMathEnvironmentsForLineno{flalign}%
        \patchBothAmsMathEnvironmentsForLineno{alignat}%
        \patchBothAmsMathEnvironmentsForLineno{gather}%
        \patchBothAmsMathEnvironmentsForLineno{multline}%
      }{}
    }
\makeatother

\begin{document}
    \maketitle

    \begin{abstract}
        Tree-form sequential decision making (TFSDM) extends classical one-shot decision making by modeling tree-form interactions between an agent and a potentially adversarial environment. It captures the online decision-making problems that each player faces in an extensive-form game, as well as Markov decision processes and partially-observable Markov decision processes where the agent conditions on observed history.
        Over the past decade, there has been considerable effort into designing online optimization methods for TFSDM.
        Virtually all of that work has been in the \emph{full-feedback} setting, where the agent has access to \emph{counterfactuals}, that is, information on what \emph{would have happened} had the agent chosen a different action at any decision node. Little is known about the \emph{bandit} setting, where that assumption is reversed (no counterfactual information is available), despite this latter setting being well understood for almost 20 years in one-shot decision making.
%
%
        In this paper, we give the first algorithm for the bandit linear optimization problem for TFSDM that offers both (i) linear-time iterations (in the size of the decision tree) and (ii) $O(\sqrt{T})$ cumulative regret in expectation compared to any fixed strategy, at all times $T$.
        This is made possible by new results that we derive, which may have independent uses as well:
        1) geometry of the dilated entropy regularizer,
        2) autocorrelation matrix of the natural sampling scheme for sequence-form strategies,
        3) construction of an unbiased estimator for  linear losses for sequence-form strategies, and
        4) a refined regret analysis for mirror descent when using the dilated entropy regularizer.
    \end{abstract}

    \section{Introduction}

%
\emph{Tree-form sequential decision making (TFSDM)} models multi-stage online decision-making problems~\citep{Farina19:Online}. In TFSDM, an agent interacts sequentially with a potentially reactive environment in two ways: (i) selecting actions at decision points and (ii) partially observing the environment at observation points. Decision points and observation points alternate along a tree structure. TFSDM captures the online decision process that each player faces in an extensive-form game, as well as Markov decision processes and partially-observable Markov decision processes where the agent conditions on observed history.
Regret minimization, one of the main mathematical abstractions in the field of online learning, has proved to be an extremely versatile tool for TFSDM. For instance, over the past decade, regret minimization algorithms such as counterfactual regret minimization (CFR)~\cite{Zinkevich07:Regret} and related newer, faster  algorithms have become popular for solving zero-sum games~\cite{Tammelin15:Solving,Brown15:Regret,Brown17:Dynamic,Brown17:Reduced,Brown19:Solving}.
These newer algorithms served as an important component in the computational game-solving pipelines that achieved several recent milestones in computing superhuman strategies in two-player limit Texas hold'em~\citep{Bowling15:Heads}, two-player no-limit Texas hold'em~\citep{Brown17:Safe,Brown17:Superhuman} and multi-player no-limit Texas hold'em~\cite{Brown19:Superhuman}.

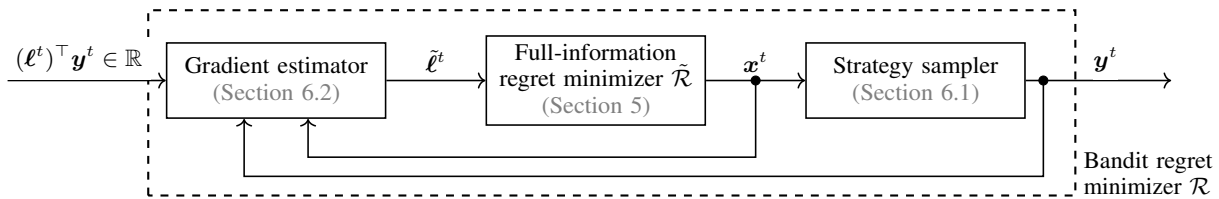
\begin{figure*}[th]\centering
        \begin{tikzpicture}[scale=.85]\small
            \node[draw,semithick,align=center,text width=2.7cm,minimum height=1cm] (G) at (0,0) {Gradient estimator\\\textcolor{gray}{(\cref{sec:l tilde})}};
            \node[draw,semithick,align=center,text width=2.7cm,minimum height=1cm] (Q) at (5,0) {Full-information regret minimizer $\tilde{\cR}$\\\textcolor{gray}{(\cref{sec:r tilde})}};
            \node[draw,semithick,align=center,text width=2.7cm,minimum height=1cm] (S) at (10,0) {Strategy sampler\\\textcolor{gray}{(\cref{sec:sampling scheme})}};
          \draw[thick,dashed] (-2,-1.8) rectangle (12.5,1.1);
          \node[right,text width=2.2cm] at (12.5,-1.45) {Bandit regret minimizer $\cR$};

          \coordinate (D1) at (7.5,0);
          \coordinate (D2) at (12,0);
          \draw[fill=black] (D1) circle (.07);
          \draw[fill=black] (D2) circle (.07);
          \draw[semithick,->] (D1) -- ++(0, -1.2) -- ++(-7, 0) -- (.5,0 |- G.south);
          \draw[semithick,->] (D2) -- ++(0, -1.5) -- ++(-12.5, 0) -- (-.5,0 |- G.south);
          \draw[semithick,->] (-4.2,0) --node[pos=.45,above,fill=white]{$(\vec{\ell}^t)^\top \vec{y}^t\in\bbR$} (G.west);
          \draw[semithick,->] (G.east) --node[above]{$\tilde{\vec{\ell}}^t$} (Q.west);
          \draw[semithick,->] (Q.east) --node[above,inner xsep=0]{$\xbar^t$} (S.west);
          \draw[semithick,->] (S.east) --node[above,pos=.55,inner xsep=0,fill=white]{$\vec{y}^t$} (14,0);
        \end{tikzpicture}
        \caption{Overview of the construction of our bandit regret minimizer $\cR$.
        }
        \label{fig:overview}
\end{figure*}

However, those methods rely on having access to \emph{counterfactuals}, that is, information on what would have happened had the agent chosen a different action at any decision point.
While this assumption is reasonable when regret minimization algorithms are used in self-play (for instance, as a way to converge to a Nash equilibrium in an extensive-form game), it limits their applicability in \emph{online} decision-making settings, where the algorithm is deployed to learn strategies (for instance, exploitative strategies) against a non-stationary and potentially adversary opponent.
In the \emph{bandit} setting that assumption is reversed, and no counterfactual information is available to the online decision maker.
Despite this latter setting being well understood for almost 20 years in \emph{one-shot} decision making, surprisingly little is known about the bandit optimization setting in \emph{sequential} decision making.
Part of the reason for this gap is that the multi-stage nature of TFSDM poses challenges that are not present in one-shot decision making, such as 1) not knowing what the environment would have done in parts of the tree that were not reached (and not even knowing the current path of play if you do not observe the environment's actions), and 2) having an exponential number of available sequential policies to choose from.

In this paper we give the first algorithm for the well-established bandit linear optimization problem in TFSDM and show that it achieves $O(\sqrt{T})$ expected regret compared to any fixed strategy even when playing against a reactive environment, while at the same time only requiring a single linear time tree traversal per iteration.
To our knowledge, there has been only one prior approach to bandit linear optimization that offers both (i) iterations that are polynomial in the number of sequences in the decision process and (ii) $\tilde{O}(\sqrt{T})$ expected regret compared to any fixed strategy~\citep{Abernethy08:Competing}. That work was for general convex sets. By focusing on TFSDM, we achieve faster iterations and convergence in fewer iterations. Our algorithm runs in linear time per iteration unlike the prior algorithm which requires that an eigendecomposition of a Hessian matrix be computed at each iteration---a cubic-time operation. Our expected regret is $O(\sqrt{T})$ instead of the prior algorithm's $O(\sqrt{T\log T})$.
One application of our algorithm and theory is to find exploitative strategies for an agent in extensive-form games against a non-adaptive opponent (that is, an opponent that cannot learn from our prior play in previous iterations of the extensive-form game) but one that can randomize and condition its actions on its observations of our play within any iteration of the extensive-form game. We provide experiments in this setting. To our knowledge, this is the first implementation of bandit optimization for TFSDM.

A known weakness of the approach of~\citet{Abernethy08:Competing}, which is also a weakness in our approach, is that the bound on regret holds not with high probability but only in expectation.
This weakness can be eliminated in theory if iterations are allowed to take exponential time in the number of sequences in the decision process~\cite{Bartlett08:High,Hazan16:Optimal} or a recent manuscript suggests that it can be achieved by accepting slower $O(T^{\nicefrac{2}{3}})$ convergence~\cite{Braun16:Efficient}. Another approach achieves $\tilde{O}(n^{9.5} \sqrt{T})$ regret, where $n$ is the size of the input TFSDM problem, at the cost of having each iteration incur into a factor that grows proportionally to the time horizon $T$~\citep{Bubeck17:Kernel}.
Due to this weakness, in our approach, the one of~\citet{Abernethy08:Competing} and other methods that do not enjoy a high-probability regret bound, when used in self play in two-player zero-sum games, the average regrets of the players might not converge to zero---but if they do, the average strategies converge to a Nash equilibrium.
It is an open problem (except for relatively simple settings like simplex~\cite{Auer02:Nonstochastic} and sphere~\cite{Abernethy09:Beating}) whether in-high-probability $\tilde{O}(\sqrt{T})$ regret bounds can be obtained in the bandit setting in polynomial-time iterations.
\citet{Abernethy09:Beating} presented a template for deriving such bounds, but several pieces therein need to be instantiated to complete the proof of bounds. The theory of the present paper offers solutions for some of those pieces for general TFSDM problems, as we will discuss, so our paper may help pave the way to solving the open problem for TFSDM.

\subsection{Overview of Our Approach}\label{sec:overview}
In this subsection we give an overview of the key ideas behind our method. We assume some basic familiarity with the concept of full-information and bandit regret minimizers; both concepts are recalled in \cref{sec:regret minimization}.

The approach we follow in this paper combines several tools and insights.
We construct a bandit regret minimizer $\mathcal{R}$ starting from a \emph{full-information} regret minimizer $\tilde{\mathcal{R}}$, that is, one that has access to the full loss vector at each iteration. Our bandit regret minimizer $\mathcal{R}$ works as follows:
\begin{itemize}[leftmargin=7mm]
    \item[(i)] the next strategy $\vec{y}^t$
  for $\mathcal{R}$
  is computed starting from the strategy $\xbar^{t}$ output by $\tilde{\mathcal{R}}$. We employ a specific unbiased \emph{sampling scheme} to sample $\vec{y}^t$ from $\xbar^t$.
      At all times $t$, we guarantee that $\bbE[\vec{y}^t | \vec{y}^1, \dots, \vec{y}^{t-1}] = \xbar^t$;
  \item[(ii)] each loss evaluation (that is, the negative of the reward of the strategy that we played in the most recent iteration) $(\vec{\ell}^{t})\trans \vec{y}^t \in \bbR$ is used to constructs an artificial loss vector $\tilde{\vec{\ell}}^t$ in a specific way that makes it an unbiased estimator of $\vec{\ell}^t$. This artificial loss vector is then passed to $\tilde{\mathcal{R}}$.
\end{itemize}
The construction of $\cR$ is summarized pictorially in \cref{fig:overview}.
We implement $\tilde{\mathcal{R}}$ using the \emph{online mirror descent} algorithm paired with a type of regularizer called the \emph{dilated entropy distance-generating function (DGF)}. The reasons behind this choice are twofold. First, it enables an efficient implementation of $\tilde{\mathcal{R}}$, since projections onto sequential strategy spaces based on the dilated entropy DGF amount to a (linear-time) traversal of the decision process. Second, it serves as the basis for defining a \emph{local, time-dependent} norm $\pn{\cdot}{t}$ that combines well with the regret bound of online mirror descent. 
Two steps are critical in the proof of the regret bound for the overall regret minimizer $\mathcal{R}$. First, we show that, in expectation, $\|\tilde{\vec{\ell}}\|_{\ast,t}$ is upper bounded by a small (time-independent) constant $c$ (the same property would not hold for a generic time-independent norm). This, combined with the local-norm regret bound mentioned above, can be used to show that the regret cumulated by $\tilde{\mathcal{R}}$ is $O(\sqrt{T})$ in expectation. Second, we use the unbiasedness of $\vec{y}^t$ and $\tilde{\vec{\ell}}^t$ to conclude that the expected regret accumulated by $\mathcal{R}$ matches that of $\tilde{\mathcal{R}}$. 

\subsection{Relationships to Related Research}\label{sec:related work}
The idea of constructing a bandit regret minimizer starting from a full-information regret minimizer was used in~\citet{Abernethy09:Beating}.
A general construction of an unbiased estimator $\tilde{\vec{\ell}}^t$ of $\vec{\ell}^t$ starting from the loss evaluation $(\vec{\ell}^t)\trans\vec{y}^t$ appears in~\citet{Bartlett08:High}.
We generalize their argument to handle strategy domains where the vector space spanned by all decision vectors is rank deficient (this is the case for sequential strategy spaces), and give several new, fundamental properties about the autocorrelation matrix of the standard sampling scheme for sequence-form strategies. 
The idea of using time-dependent norms to obtain a tighter regret analysis than time-independent norms appeared in, for example,~\citet{Abernethy08:Competing,Abernethy09:Beating,Shalev-Shwartz12:Online}, while the use of the dilated entropy regularizer in the context of sequential decision making and extensive-form games for other purposes goes back to the original work by~\citet{Hoda10:Smoothing}, with important newer practical observations by~\citet{Kroer18:Faster}.

\textsc{Exp3}~\citep{Auer02:Nonstochastic} is credited to be the first bandit regret minimizer for simplex domains. \textsc{GeometricHedge}~\citep{Dani08:Price} is a general-purpose bandit regret minimizer that can be applied to any set of decisions. However, it  requires one to compute a barycentric spanner~\citep{Awerbuch04:Adaptive}, which in our setting would have prohibitive pre-processing cost. Furthermore, it runs in exponential time per iteration in the general case, and it is not known whether that can be avoided in our setting. 


\citet{Lanctot09:Monte} suggested as a side note that a specific \emph{online} variant of their Monte Carlo CFR (MCCFR) algorithm (as opposed to the usual \emph{self-play} MCCFR algorithm) could be used for online decision making without counterfactuals. Their paper did not provide theoretical guarantees for that online variant. The well-established bandit linear optimization setting considered in this paper is quite different from the one that online MCCFR implicitly operates on. First, in bandit optimization (our setting), each strategy is output \emph{before} the environment reveals feedback, and the only feedback that the environment gives is a single real-valued reward $(\vec{\ell}^t)^{\!\top} \vec{x}^t$. In contrast, in online MCCFR the feedback is not just the final payoff, as online MCCFR needs to know which path was followed in the game tree and the terminal leaf, so that regrets can be updated for all decision nodes of the player on the path from the root to the leaf. So, even if a version of online MCCFR with theoretical guarantees were developed, it would \emph{not} be an algorithm for bandit linear optimization, but rather an algorithm for a different (and easier, since more feedback is given to the decision maker) online learning setting. Depending on the applications, that setting---which, to our knowledge, has never been investigated nor formally proposed---might be more or less natural than bandit linear optimization. Since the bandit linear optimization model does not require that the decision maker observe the path of play, it can be used to model settings in which (i) there is no path in the game tree, because the loss given by the environment does not represent playing against an opponent; (ii) the decision maker does not interact immediately with the environment: the output strategy is evaluated at a later time by the environment and feedback is given only then; (iii) the environment does not inform the decision maker of the specific trajectory taken in the interaction out of privacy concerns; or any combination of the above.

\section{Review of Sequential Decision Making}
\label{subseq:SDP}


%
The decision process of an TFSDM problem is structured as a tree of decision points---in which an action must be selected by the agent---and observation points---in which the
environment reveals a signal to the agent.
%
%
We denote the set of decision points in the TFSDM problem as $\cJ$, and the set of observation points as $\cK$. At each decision point $j \in \cJ$, the agent selects an action from the set $A_j$ of available actions. At each observation point $k \in \cK$, the agent observes a signal $s_k$ from the environment out of a set of possible signals $S_k$.
We denote by $\rho$ the transition function of the process. Picking action $a \in A_j$ at decision point $j\in\cJ$ results in the process transitioning to $\rho(j,a) \in \cJ \cup \cK \cup \{\terminalnode\}$, where $\terminalnode$ denotes the end of the decision process. Similarly, the process transitions to $\rho(k,s) \in \cJ \cup \cK \cup \{\terminalnode\}$ after the agent observes signal $s \in S_k$ at observation point $k \in \cK$.
In line with the game theory literature, we call a pair $(j,a)$ where $j\in \cJ$ and $a \in A_j$ a \emph{sequence}. The set of all sequences is denoted as $\Sigma \defeq \{(j,a): j\in\cJ, a\in A_j\}$. For notational convenience, we will often denote an element $(j,a)$ in $\Sigma$ as $ja$ without using parentheses.
Given a sequence $ja \in \Sigma$, we denote by $\vec{u}_{ja}$ the vector such that $(\vec{u}_{ja})_{j'a'} = 1$ if the (unique) path from the root node to action $a'$ at decision point $j'$ passes through action $a$ at decision point $j$, and $(\vec{u}_{ja})_{j'a'} = 0$ otherwise.
Finally, given a node $v \in \cJ \cup \cK$, we denote by $p_v$ its \emph{parent sequence}, defined as the last sequence (that is, decision point-action pair) encountered on the path from the root to $v$. If the agent does not act before $v$ (that is, $v$ is the root of the process or only observation points are encountered on the path from the root to $v$), we let $p_v = \emptyseq$. We use the symbol $N_v$ to denote the number of decision points in the subtree rooted at $v$. If $v$ itself is a decision point, $v$ is included in the count.

\paragraph{Strategies in TFSDM problems}
%
A strategy for an agent in an TFSDM problem specifies a distribution over the set of actions $A_j$ at each decision point $j \in \cJ$. We represent a strategy using the \emph{sequence-form representation}, that is, as a vector $\xbar \in \Rp^{|\Sigma|}$ whose entries are indexed by $\Sigma$. The entry ${x}_{ja}$ contains the product of the probabilities of all actions at all decision points on the path from the root of the process to action $a$ at decision point $j\in \cJ$. In order to be a valid sequence-form strategy, the entries in $\xbar$ must satisfy the following consistency constraints~\citep{Romanovskii62:Reduction,Koller94:Fast,Stengel96:Efficient}:
\begin{equation}\label{eq:sf}
\begin{split}
  \sum_{a\in A_j} {x}_{ja} = {x}_{p_j} \quad\forall j\in \cJ \text{ s.t. } p_j \neq \emptyseq,\\
  \sum_{a\in A_j} {x}_{ja} = 1 \quad\forall j\in \cJ \text{ s.t. } p_j = \emptyseq.
\end{split}
\end{equation}
Since $\emptyseq$ is not an element in $\Sigma$, there is no entry in $\xbar$ that corresponds to $\emptyseq$, and the notation ${x}_\emptyseq$ is invalid. We will slightly abuse notation and refer to ${x}_\emptyseq$ to mean the constant value $1$. 
Finally, we let $\pure \subseteq \Rp^{|\Sigma|}$ be the finite set of all \emph{pure} (also known as \emph{deterministic}) sequence-form strategies, that is, strategies that assign probability $1$ to exactly one action at each decision point.
%
The set of all sequence-form strategies, denoted $\seqf$, is the convex hull $\seqf \defeq \co\pure$ of the set of pure strategies $\pure$.


    \section{Regret Minimization}
\label{sec:regret minimization}
A regret minimizer is an abstraction for a repeated decision maker. The decision maker repeatedly interacts with an unknown (possibly adversarial) environment by choosing points $\vec{x}^1, \dots, \vec{x}^T$ from a set $\cX \subseteq \bbR^n$ of feasible decisions and incurring a linear loss $(\vec{\ell}^1)\trans \vec{x}^1, \dots, (\vec{\ell}^T)\trans\vec{x}^T$ after each iteration.  For the purposes of this paper, the points are strategies (policies) for the
agent, so we will use the terms point and strategies interchangeably in this section.

%
The quality metric for a regret minimizer is its \emph{regret}, which measures the difference in loss against the best \emph{fixed} (that is, time-independent) decision in hindsight. Formally, given a decision $\vec{z} \in \cX$, the regret cumulated against $\vec{z}$ up to time $T$ is defined as
\[
  R^T(\vec{z}) \defeq \sum_{t = 1}^{T} (\vec{\ell}^t)\trans (\vec{x}^t -  \vec{z}).
\]
A ``good'' (aka. \emph{Hannan consistent}) minimizer is such that the regret compared to \emph{any} $\vec{z}\in\cX$ grows sublinearly in $T$.
This paper is interested in two types of regret minimizers, which differ in the feedback that the algorithm receives.

\paragraph{Full-Information Setting.}
In the \emph{full-information} setting, at all time steps $t = 1,\dots,T$, the regret minimizer interacts with the environment as follows:
\begin{itemize}[leftmargin=5mm]
    \item $\textsc{NextStrategy}()$: the agent outputs the next point $\vec{x}^t \in \cX \subseteq \bbR^n$. The next decision can depend on the past decisions $\vec{x}^1, \dots, \vec{x}^{t-1}$ as well as the corresponding feedback $\vec{\ell}^1, \dots, \vec{\ell}^{t-1}$, which we define next;
\item $\textsc{ObserveLoss}(\vec{\ell}^t)$: the environment selects a loss vector $\vec{\ell}^t \in \bbR^n$ and the agent observes $\vec{\ell}^t$. The loss vector can depend on the decisions $\vec{x}^1,\dots,\vec{x}^{t}$ that were output by the regret minimizer so far.
\end{itemize}
Our construction of $\tilde{\mathcal{R}}$ (\cref{sec:dgf}) provides a full-information regret minimizer for the set $\cX = \seqf$.
So, $\tilde{\mathcal{R}}$'s decisions are  (potentially randomized) sequence-form strategies.

\paragraph{Bandit Setting.}
In the \emph{bandit} setting the environment does \emph{not} reveal the selected loss vector $\vec{\ell}^t$ at each iteration, but only the evaluation $(\vec{\ell}^t)\trans \vec{x}^t$ of the loss function for the latest decision $\vec{x}^t$. Formally, at all time steps $t = 1,\dots,T$, the  regret minimizer interacts with the environment as follows:
\begin{itemize}[leftmargin=5mm]
    \item $\textsc{NextStrategy}()$: the agent outputs the next point $\vec{x}^t \in \cX \subseteq \bbR^n$. As in the full-information setting, the next strategy can depend on the past strategies and corresponding feedbacks, which we define next;
    \item $\textsc{ObserveLossEvaluation}((\vec{\ell}^t)\trans \vec{x}^t)$: the environment selects a loss vector $\vec{\ell}^t \in \bbR^n$ and the agent observes $(\vec{\ell}^t)\trans \vec{x}^t$. We assume without loss of generality that $(\vec{\ell}^t)\trans \vec{x}^t\in [0,1]$ at all $t$. The loss vector can depend on the decisions $\vec{x}^1,\dots,\vec{x}^{t-1}$ that were output by the regret minimizer \emph{before} time $t$, but \emph{not} on $\vec{x}^t$.
\end{itemize}
Since the regret minimizer only observes $(\vec{\ell}^t)\trans \vec{x}^t$, it cannot compute any \emph{counterfactual} information (that is, compute the value of the loss at a decision other than the one that was output).
Currently, the bandit setting represents the hardest setting in which the information-theoretic upper bound of $\tilde{O}(\sqrt{T})$ regret is known to be attainable, but very little is known about sequential decision making under that setting, and existing algorithms are not computationally practical.%
\footnote{A third online learning setting---called the \emph{semi-bandit optimization setting}---has been proposed in the literature~\cite{Gyorgy07:Line,Kale10:Non,Audibert14:Regret,Neu13:Efficient}. The feedback that the decision maker receives at all times $t$ in that setting is the component-wise product $\vec{\ell}^t \circ \vec{x}^t$. The semi-bandit feedback provides counterfactual information. Instead, in this paper we are interested in the bandit setting, where no counterfactual information is available.} 
%
%

%

\section{Dilated Entropy and Local Norms}\label{sec:dgf}

The dilated entropy distance-generating function (DGF) is a regularizer that induces a notion of distance that is suitable for the sequence-form strategies spaces. This regularizer was first introduced in the context of extensive-form games~\citep{Hoda10:Smoothing}. \citet{Kroer18:Faster}---with earlier results by~\citet{Kroer15:Faster}---analyzed several properties of this function, including its $1$-strong convexity with respect to the $\ell_1$ and $\ell_2$ norms. They also showed that the dilated entropy DGF leads to state-of-the-art convergence guarantees in iterative methods for computing Nash equilibrium in two-player zero-sum extensive-form games of perfect recall. We define this kind of DGF as follows.
    \begin{definition}
    \label{def:dilated entropy}
        Let $\co \pure$ be the set of sequence-form strategies for the TFSDM problem. The \emph{dilated entropy} distance-generating function for $\co \pure$ is the function $\regu : \Rpp^{|\Sigma|} \to \Rp$ defined as
        \begin{equation*}
            \regu : \xbar \mapsto \sum_{j\in\cJ} w_j \mleft({x}_{p_j}\log |A_j| + \sum_{a\in A_j} {x}_{ja} \log \frac{{x}_{ja}}{{x}_{p_j}} \mright),
        \end{equation*}
        where the weights $w_j$ are defined recursively according to:
\begin{equation*}
    w_j = 2 + 2 \max_{a \in A_{j}} \{w_{\rho(j,a)}\}, \quad
     w_k = \sum_{s \in S_{k}} w_{\rho(j,s)},\quad w_{\terminalnode} = 0.
\end{equation*}
    \end{definition}
The range of $\regu$ is a game-dependent constant, and usually polynomial in the size of the TFSDM problem~\citep{Kroer17:Theoretical}.
The (unique) minimum of $\regu$ is attained by the sequence-form strategy that at each decision point uniformly randomizes among all available actions (that is, $x_{ja} = x_{p_j}/|A_j|$ for all $j\in\cJ,a\in A_j$).

The dilated entropy DGF has the benefit that its gradient and its Fenchel conjugate function can be evaluated efficiently via a linear-time pass of the decision process~\citep{Hoda10:Smoothing}. In particular, for all $\vec{z} \in\Rpp^{|\Sigma|}$, there exists an exact algorithm, denoted \textsc{Gradient}, to compute $\nabla \regu(\vec{z})$ in linear time in $|\Sigma|$. Also, there exists an exact algorithm, denoted \textsc{ArgConjugate}, to compute $
      \nabla \regu^*(\vec{z}) = \argmax_{\hat{\vec{x}} \in \co \pure}\{\vec{z}\trans \hat{\vec{x}} - \regu(\hat{\vec{x}})\}
    $
    in linear time in $|\Sigma|$.
This makes $\regu$ an appealing candidate regularizer in many TFSDM optimization algorithms, including the full-information regret minimizer $\tilde{\mathcal{R}}$ that we use in this paper.
In \cref{app:dgf} in the full version of this paper\footnote{The full version of this paper, including appendix, is available on arXiv.} we give pseudocode for \textsc{Gradient} and \textsc{ArgConjugate}.

As mentioned in the introduction, the analysis of our bandit regret minimizer needs to take into consideration the particular geometry of the dilated entropy DGF. Specifically, at each point $\xbar \in \seqf$ in the sequence-form strategy space, the dilated entropy DGF induces a pair of primal-dual \emph{local} norms $(\pn{\cdot}{\xbar},\dn{\cdot}{\xbar})$ defined for all $\vec{z}\in\bbR^{|\Sigma|}$ as
\begin{align*}
  \pn{\vec{z}}{\xbar} \!\defeq\! \sqrt{\vec{z}\trans \nabla^2 \regu(\xbar)\, \vec{z}};\quad
  \dn{\vec{z}}{\xbar} \!\defeq\! \sqrt{\vec{z}\trans (\nabla^2 \regu(\xbar))^{-1}\vec{z}},
\end{align*}
where $\nabla^2 \regu(\xbar)$ denotes the Hessian matrix of $\regu$ at $\vec{\xbar}$. Since $\nabla^2 \regu(\xbar)$ is positive-definite, it is known that $\dn{\cdot}{\xbar}$ is well-defined and that it is indeed dual to $\pn{\cdot}{\xbar}$, in the sense that $\dn{\vec{z}}{\xbar} = \max\{\vec{z}\trans\vec{w} : \pn{\vec{w}}{\xbar} \le 1\}$ for all $\vec{z} \in \bbR^{|\Sigma|}$.

    To our knowledge, we are the first to explore the local norms induced by the dilated entropy DGF.
    These norms are a fundamental ingredient in our construction, and here we give several properties that we will use in later sections.
%
    In \cref{app:local norms} in the full version of this paper we give several results regarding analytical properties of these norms, including a useful characterization of the inverse Hessian matrix of the DGF $\regu$ at a generic strategy $\xbar \in \seqf$ in terms of sum of dyadics.
%

    \section{Construction of $\tilde{\cR}$}\label{sec:r tilde}

\begin{figure*}
    \makeatletter\let\@latex@error\@gobble\makeatother
    \scalebox{.9}{
    \begin{minipage}[t]{.54\linewidth}
    \begin{algorithm}[H]
      \caption{Full-information regret minimizer $\tilde{\mathcal{R}}$\hspace*{-5mm}}
        \label{algo:R tilde}\DontPrintSemicolon
          \KwData{$\eta$ is a step-size parameter.}
        \BlankLine
        \Fn{\normalfont\textsc{Setup}()}{
            \For{$j \in \cJ$ in top-down order}{\vspace{.3mm}
                \textbf{for} $a \in A_j$ \textbf{do} ${x}^1_{ja} \gets \frac{{x}_{p_j}}{|A_j|}$\;
            }
        }
        \Hline{}
        \vspace{.8mm}
        \textbf{function} \textsc{NextStrategy}(): \textbf{return} $\xbar^t$\;
        \vspace{.8mm}
        \Hline{}
        \vspace{.8mm}
        \Fn{\normalfont\textsc{ObserveLoss}($\tilde{\vec{\ell}}^{t}$)}{
            $\vec{g} \gets \eta\tilde{\vec{\ell}}^{t} - \textsc{Gradient}(\xbar^{t})$\Comment*{\color{commentcolor}\cref{sec:dgf}]\hspace*{-4mm}}
            $\xbar^{t+1}\gets\textsc{ArgConjugate}(-\vec{g})$\Comment*{\color{commentcolor}\cref{sec:dgf}]\hspace*{-4mm}}
        }
    \end{algorithm}
    \end{minipage}}
    \hfill
    \scalebox{.9}{
    \begin{minipage}[t]{.54\linewidth}
        \begin{algorithm}[H]
            \caption{Bandit regret minimizer $\mathcal{R}$ \phantom{$\tilde{\mathcal{R}}$}}
            \label{algo:R}\DontPrintSemicolon
            \Fn{\normalfont\textsc{Setup}()}{
                $\tilde{\mathcal{R}}$.\normalfont\textsc{Setup()}\Comment*{\color{commentcolor}\cref{algo:R tilde}]\hspace*{-4mm}}
            }
            \Hline{}
            \Fn{\normalfont\textsc{NextStrategy}()}{
                $\xbar^t \gets \tilde{\mathcal{R}}.\normalfont\textsc{NextStrategy}()$\Comment*{\color{commentcolor}\cref{algo:R tilde}]\hspace*{-4mm}}
                $\vec{y}^t \gets \textsc{Sample}(\xbar^t)$\Comment*{\color{commentcolor}\cref{sec:sampling scheme}]\hspace*{-4mm}}
                \Return{$\vec{y}^t$}\;
            }
            \Hline{}
            \Fn{\normalfont\textsc{ObserveLossEvaluation}($l \defeq(\vec{\ell}^{t})\trans \vec{y}^t$)\!\!\!\!}{
                $\tilde{\vec{\ell}}^t \gets \textsc{LossEstimate}(l,\xbar^t, \vec{y}^t)$\Comment*{\color{commentcolor}Algorithm~\ref{algo:ell tilde}]\hspace*{-4mm}}
                $\tilde{\mathcal{R}}.\normalfont\textsc{ObserveLoss}(\tilde{\vec{\ell}}^t)$\Comment*{\color{commentcolor}\cref{algo:R tilde}]\hspace*{-4mm}}
            }
        \end{algorithm}
    \end{minipage}}
    \vspace{-2mm}
\end{figure*}

%
%
Our full-information regret minimizer $\tilde{\mathcal{R}}$ is constructed using online mirror
descent---one of the most well-studied full-information regret minimization algorithms in online
learning---instantiated with the dilated entropy DGF $\regu$ (\cref{def:dilated entropy}) as the
regularizer and the set $\seqf\subseteq \bbR^{|\Sigma|}$ of sequence-form strategies in the game as
the domain of feasible iterates. Pseudocode for $\tilde{\cR}$ is given in \cref{algo:R tilde}, where
$\eta > 0$ is a stepsize parameter that can be tuned at will.


Those properties are key to the analysis of the regret cumulated by \cref{algo:R tilde} as a
function of the local dual norms of the loss vectors $\tilde{\vec{\ell}}^t$; that analysis is rather
lengthy and we defer it to \cref{app:R tilde} in the full version of this paper. Here, we only state a key result.

\begin{theorem}\label{thm:omd regret bound}
    Let $D$ be the maximum depth of any node in the decision process, and let $\vec{z} \in \seqf$.
    If $\tilde{\vec{\ell}}^t \in \Rp^{|\Sigma|}$ at all times $t$, then at all times $T$ the regret
    $\tilde{R}^T(\vec{z})$ cumulated by $\tilde{\mathcal{R}}$ satisfies:
    \begin{equation}\label{eq:omd regret bound}
        \tilde{R}^T(\vec{z}) \le \frac{\regu(\vec{z})}{\eta} + \eta\sqrt{3D} \cdot \sum_{t=1}^T \dn{\tilde{\vec{\ell}}^{t}}{\xbar^{t}}^2.
    \end{equation}
\end{theorem}

Incidentally, since the range of $\regu$ over $\seqf$ only depends on the TFSDM problem structure and not on the time
$T$, by setting $\eta = \Theta(1/\sqrt{T})$, we obtain a regret bound of the form
\begin{align*}
    \bbE\mleft[\tilde{R}^T(\vec{z})\mright] &= O\mleft(\frac{1}{\sqrt{T}}\,\bbE\mleft[\sum_{t=1}^T
    \dn{\tilde{\vec{\ell}}^t}{\xbar^t}^2\mright] \mright).
\end{align*}
In the next section, we show how to construct $\tilde{\vec{\ell}}$ in the right hand side and we
prove that the right hand side is small in expectation. Then in \cref{sec:full algo} we prove that
the expectation of the regret on the left hand side equals the expectation of the regret of the
bandit regret minimizer $\mathcal{R}$.

    \section{Unbiased Loss Estimate and Construction of $\mathcal{R}$}
    As mentioned in \cref{sec:overview}, two different components are crucial for our bandit regret minimizer $\mathcal{R}$: the sampling scheme and the construction of the unbiased loss estimates.

    \subsection{Sampling Scheme for TFSDM}\label{sec:sampling scheme}
At each time step $t$, the bandit regret minimizer $\mathcal{R}$ internally calls  $\tilde{\mathcal{R}}.\textsc{NextStrategy}()$ and receives a sequence-form strategy $\xbar^{t} \in \seqf$. Then, $\mathcal{R}$ samples and returns a \emph{pure} sequence-form strategy $\vec{y}^t \in \pure$ such that $\bbE_t[\vec{y}^t] = \xbar^t$. We use the natural sampling scheme for sequence-form strategies: at each decision point $j$, we pick an action $a\in A_j$ according to the distri\-bution ${x}^t_{ja} / {x}^t_{p_j}$ induced by the sequence-form strategy $\xbar^t$.
It is well known (and straightforward to verify---see \cref{app:unbiased} in the full version of this paper) that this sampling scheme is unbiased.

As we will show, in order to balance exploration and exploitation along the structure of the TFSDM problem and construct unbiased loss estimates, an analysis of the autocorrelation matrix $\vec{C}^t \defeq \bbE[\vec{y}^t(\vec{y}^t)^{\!\top} \,|\, \vec{y}^1,\dots, \vec{y}^{t-1}]$ of the sampling scheme, as well as its inverse, can be used. To our knowledge, we are the first to study the autocorrelation matrix of the natural sampling scheme for sequence-form strategies. We do so in \cref{app:autocorrelation}.

    \subsection{Computation of the Loss Estimate $\tilde{\vec{\ell}}^t$}
\label{sec:l tilde}
Our construction of the unbiased loss estimate extends and generalizes that of \citet{Dani08:Price}, in that it can be applied even when the set of strategies is rank deficient, for example here where the pure strategies $\pure$ of the sequence form strategy space only span a strict subspace of the natural Euclidean space $\bbR^{|\Sigma|}$ to which the sequence-form strategies belong.
    In particular, we relax the notion of unbiasedness to mean the weaker condition that the projection $\tilde{\vec{\ell}}^t$ onto the direction\footnote{The direction $\lin\cX$ of a set $\cX$ is the subspace defined as
      $\lin \cX \defeq \Span \{\vec{u} - \vec{v}: \vec{u},\vec{v} \in \cX\}.$} $\lin \pure$ of $\pure$ be an unbiased estimator of the projection of the original (and unknown) $\vec{\ell}^t$
    onto $\lin \pure$:
    \begin{equation}\label{eq:relax unbias}
      \bbE_t[\tilde{\vec{\ell}}^t]\trans \vec{w} = (\vec{\ell}^t)\trans \vec{w} \qquad\forall\, \vec{w} \in \lin \pure, \tag{$\star$}
    \end{equation}
    where $\bbE_t[\,\cdot\,]$ is an abbreviation for $\bbE_t[\,\cdot\,|\, \vec{y}^1, \dots, \vec{y}^{t-1}]$, that is,
    the expectation conditional on the previous decisions of $\mathcal{R}$.
    The main technical tool in our construction is the use of a generalized inverse of the autocorrelation matrix of $\vec{y}^t$, as shown by the next proposition (the proof is in \cref{app:R} in the full version of this paper).

    \begin{restatable}{proposition}{propftilde}\label{prop:f tilde}
        Let $\pi^t$ be a conditional distribution over $\pure$, given the previous decisions $\vec{y}^1,\dots,\vec{y}^{t-1}$, and suppose that the
    support of $\pi^t$ has full rank (that is, $\Span\supp\pi^t \!= \Span \pure$). Let
    $
        \mat{C}^t \defeq \bbE_{t}[\vec{y}^t (\vec{y}^t)\trans]
    $
    be the autocorrelation matrix of $\vec{y}^t$, and let $\mat{C}^{t-}$ be any generalized inverse of $\mat{C}^{t}$, that is, any matrix such that
    $
      \mat{C}^t \mat{C}^{t-} \mat{C}^t = \mat{C}^t.
    $
    Furthermore, let $\vec{b}^t$ be such that $\bbE_t[\vec{b}^t] \perp \lin \pure$. Then, the random variable
        \begin{equation}\label{eq:f tilde}
            \tilde{\vec{\ell}}^t \defeq [(\vec{\ell}^t)\trans \vec{y}^t] \cdot \mat{C}^{t-}\, \vec{y}^t + \vec{b}^t
        \end{equation}
        satisfies~\eqref{eq:relax unbias}.
    \end{restatable}

    Crucially, the loss estimate $\tilde{\vec{\ell}}^t$ in \cref{prop:f tilde} can be constructed using only the bandit feedback (loss evaluation) $(\vec{\ell}^t)\trans \vec{y}^t$ that was received at time $t$ after the regret minimizer output $\vec{y}^t$ as its decision.
At each time $t$, we use \cref{prop:f tilde} to construct the loss estimate $\tilde{\vec{\ell}}^t$. The main conceptual leap is to identify
\begin{itemize}[leftmargin=7mm]
    \item[(i)] a choice of generalized inverse $\mat{C}_*^{t-}$ for the autocorrelation matrix $\mat{C}^t$ of $\vec{y}^t$ returned by \cref{algo:sampling scheme}; and
    \item[(ii)] a particular choice for the (random) vector $\vec{b}_*^t$ such that $\bbE_t[\vec{b}_*^t] \perp\lin\pure$ so that (a) the expression $[(\vec{\ell}^t)\trans \vec{y}^t]\mat{C}_*^{t-}\vec{y}^t+\vec{b}_*^t$ can be evaluated in $O(|\Sigma|)$ time and (b) the resulting loss function $\tilde{\vec{\ell}}^t$ is nonnegative, as required by $\tilde{\mathcal{R}}$ (\cref{thm:omd regret bound}).
\end{itemize}
At a high level, the particular construction that we use generates $\mat{C}_*^{t-}$ and $\vec{b}_*^t$ inductively in a bottom-up fashion by traversing the decision process, and heavily relies on the combinatorial structure of the autocorrelation matrix $\mat{C}^t$. All details and proofs are in \cref{app:f tilde} in the full version of this paper.

\begin{figure}[ht!]
        \SetInd{0.2em}{0.4em}%
        \makeatletter\let\@latex@error\@gobble\makeatother
        \begin{algorithm}[H]\small
          \caption{\normalfont$\textsc{LossEstimate}(l, \xbar^t, \vec{y}^t)$\hspace*{-.8cm}}
            \label{algo:ell tilde}\DontPrintSemicolon
            $\tilde{\vec{\ell}}^t \gets \vec{0}\in\Rpp^{|\Sigma|}$\;
            \Subr{\normalfont$\textsc{Traverse}(v, \alpha_v)$}{
                \uIf{$v \in \cK$} {
                    \For{$s\in S_v$}{
                        $\displaystyle\textsc{Traverse}\mleft(\rho(v,s), \frac{\alpha_v}{|S_v|} + \frac{|S_v|-1}{|S_v|}(1 - l)y^t_{p_v}\mright)$\hspace*{-1cm}\;
                    }
                }
                \Else(\Comment*[f]{\color{commentcolor}that is, $v\in\cJ$]\hspace*{-4mm}}){
                    \For{$a \in A_v$}{
                        \uIf{$\rho(v,a) \neq \terminalnode$}{
                            $\displaystyle\vec{\ell}^t_{ja} \gets \frac{y^t_{va}}{{x}^t_{va}} (N_{v} - N_{\rho(v,a)})$\;
                            $\displaystyle\textsc{Traverse}\mleft(\rho(v,a), \frac{{x}_{p_v}}{{x}_{va}}\alpha_v \mright)$\;
                        }\ElseIf{$\rho(v,a) = \terminalnode$}{
                            $\displaystyle\vec{\ell}^t_{ja} \gets \frac{\alpha_v}{{x}^t_{p_v}}+ \frac{y^t_{va}}{{x}^t_{va}}(l + N_v - 1)$\;
                        }
                    }
                }
            }
            \Hline{}
            $\textsc{Traverse}(r, 0)$\Comment*{\color{commentcolor}$r$: root of the decision process]\hspace*{-4mm}}
            \textbf{return} $\tilde{\vec{\ell}}^t$\;
        \end{algorithm}
\end{figure}
The resulting algorithm is \cref{algo:ell tilde}, where we let $l \defeq (\vec{\ell}^t)^\top \vec{y}^t$ denote the bandit feedback at iteration $t$ in accordance to \cref{algo:R}.
\begin{restatable}{proposition}{proplossestimate}\label{prop:loss estimate}
    At all times $t$, the vector $\tilde{\vec{\ell}}^t$ returned by $\textsc{LossEstimate}(l, \xbar^t, \vec{y}^t)$ satisfies~\eqref{eq:relax unbias}. Furthermore, \cref{algo:ell tilde} amounts to a single traversal of the tree structure of the TFSDM problem and runs in linear time in the number of sequences $|\Sigma|$.
\end{restatable}
In the special case where the decision process only has one decision point (i.e., the strategy space is a simplex), the loss estimate constructed by \cref{algo:ell tilde} coincides with that of the EXP3 algorithm of~\citet{Auer02:Nonstochastic}. However, for general sequential decision processes, the loss estimate is significantly more complicated and not based on an importance-sampling-based argument anymore.



%
Finally, because of the assumption that $l = (\vec{\ell}^t)\trans \vec{y}^t \in [0,1]$, which can be assumed without loss of generality, the loss estimate constructed as just described is non-negative: $\tilde{\vec{\ell}}^t \in \Rp^{|\Sigma|}$. So, \cref{thm:omd regret bound} applies.

\subsection{Norm of the Loss Estimate}\label{sec:norm of f tilde}
In theory, each entry of $\tilde{\vec{\ell}}^t$ can be arbitrarily large since ${x}_{ja}^t$ can be arbitrarily small. As a consequence, the Euclidean norm $\|\tilde{\vec{\ell}}^t\|_2$ of the loss estimate can be arbitrarily large, even in expectation. This shows the importance of having \cref{eq:omd regret bound} expressed in terms of the \emph{local} norms $\dn{\cdot}{\xbar^t}$ instead of a generic time-invariant norm.
Indeed, it is possible to give guarantees on the expectation of the local dual norm of $\tilde{\vec{\ell}}^t$, as we do in the next theorem.
\begin{restatable}{theorem}{thmexpecteddualnorm}\label{thm:expected dual norm}
    At all times $t$, the loss estimate $\tilde{\vec{\ell}}^t \in \Rp^{|\Sigma|}$ returned by $\textsc{LossEstimate}(l, \xbar^t, \vec{y}^t)$, where $r$ is the root of the sequential decision process, satisfies
    \[
        \bbE_t\mleft[\dn{\tilde{\vec{\ell}}^t}{\xbar^t}^2\mright] \le 4\cdot|\Sigma|^3.
    \]
\end{restatable}
The proof is in \cref{app:f tilde norm} in the full version of this paper. \cref{thm:expected dual norm} is one of the deepest results in this paper. It ties together the sampling scheme (\cref{sec:sampling scheme}), the construction of the loss estimates (\cref{sec:l tilde}), and the geometry of the local norms (\cref{sec:dgf}) induced by the dilated entropy DGF.
It combines properties of the particular choice of generalized inverse $\mat{C}_*^{t-}$ and orthogonal vector $\vec{b}_*^t \perp \lin\pure$ with an inductive argument on the TFSDM problem structure.

    \section{The Full Algorithm}\label{sec:full algo}

We construct our bandit regret minimizer $\mathcal{R}$ (\cref{algo:R}) starting from the full-information regret minimizer $\tilde{\mathcal{R}}$ of \cref{algo:R tilde}. The resulting algorithm is surprisingly easy to implement, and requires only two linear traversals of the decision process per iteration.
The regret ${R}^T(\vec{z})$ of $\mathcal{R}$ is linked to the regret $\tilde{R}^T(\vec{z})$ of $\tilde{\mathcal{R}}$: using the definition of regret and the law of total expectation together with the standard bandit optimization assumption that $\vec{\ell}^t$ is conditionally independent from $\vec{y}^t$, as well as \cref{lem:sample unbiased} and \eqref{eq:relax unbias}, we immediately find that $\bbE[R^T(\vec{z})] = \bbE[\tilde{R}^T(\vec{z})]$
%
\cref{thm:omd regret bound} gives an upper bound for the regret $\tilde{R}^T(\vec{z})$ of $\tilde{\mathcal{R}}$ as a function of the sequence of the loss estimates $\tilde{\vec{\ell}}^1, \dots, \tilde{\vec{\ell}}^T$. Taking expectations in~\cref{eq:omd regret bound} and using 
\cref{thm:expected dual norm},
\begin{align*}
  \bbE[\tilde{R}^T(\vec{z})] &\le \frac{\regu(\vec{z})}{\eta} + \eta\sqrt{3D}\cdot \bbE\mleft[\sum_{t=1}^T \bbE_t\mleft[\dn{\tilde{\vec{\ell}}^t}{\xbar^t}^2\mright]\mright]\\
  &\le \frac{\regu(\vec{z})}{\eta} + 4\eta\, |\Sigma|^3\,\sqrt{3D}\cdot T.
\end{align*}
Setting $\eta = 1/(2|\Sigma|^{\nicefrac{3}{2}}\sqrt{T})$, we obtain the following theorem, which is the bottom-line result of this paper.

\begin{theorem}\label{thm:regret fixed}
  Let $D$ be the maximum depth of any node in the decision process. Then, assuming $(\vec{\ell}^t)\trans \vec{y}^t \in [0,1]$ at all times $t = 1,\dots, T$, the regret $R^T(\vec{z})$ cumulated by \cref{algo:R}
satisfies
\[
  \bbE[R^T(\vec{z})] \le 2(\regu(\vec{z}) + \sqrt{3D})\, |\Sigma|^{\nicefrac{3}{2}}\cdot \sqrt{T}\qquad \forall\,\vec{z} \in \co \pure.
\]
\end{theorem}
\cref{thm:regret fixed} shows that the expected regret cumulated by our algorithm is $O(\sqrt{T})$. This improves on the algorithm of~\citet{Abernethy08:Competing}, whose regret is $O(\sqrt{T\log T})$ and that assumes that $T \ge |\Sigma|$.
We conclude this section with a word of caution. Our algorithm, like the one of~\citet{Abernethy08:Competing}, guarantees that $\max_{\vec{z}\in\seqf} \bbE[R^T(\vec{z})]$ is small, but \emph{not} that $\bbE[\max_{\vec{z}\in\seqf} R^T(\vec{z})]$ is small. Depending on the application, this may or may not be sufficient. This limitation is well known (e.g.,~\citep{Abernethy09:Beating}) and is one of the main drivers behind regret minimizers that provide high-probability regret bounds. In the conclusions we will discuss how the techniques of the present paper are relevant toward that effort.

    \section{Experimental Evaluation}

We implemented our bandit regret minimizer (\cref{algo:R}) and the algorithm of \citet{Abernethy08:Competing} (using a logarithmic barrier) (from now on, denoted AHR), which is, to our knowledge, the only prior algorithm that is known to guarantee $\tilde{O}(\sqrt{T})$ regret and polynomial-time iterations in the bandit optimization setting. We compared them on four domains: a simple $2\times 3$ matrix game (which is an instance of an TFSDM problem with no observation points and only one decision node), and three standard extensive-form games in the computational game theory literature, namely Kuhn poker~\citep{Kuhn50:Simplified}, 3-rank Goofspiel\cite{Ross71:Goofspiel}, and Leduc poker~\citep{Southey05:Bayes}. The sequential decision making problem faced by the first player has 13 sequences in Kuhn poker, 262 sequences in Goofspiel, and 337 sequences in Leduc poker. A complete description of those games is available in \cref{app:games} in the full version of this paper.
 The two algorithms face the same strong opponent that at each iteration plays according to a fixed strategy $\bar{\vec{s}}$ that is part of a Nash equilibrium of the game. For our method, we use the theoretical step size multiplied by $5$, while for the method of \citet{Abernethy08:Competing} we multiply their step size parameter by $2$; these changes do not affect the theoretical guarantees but improved the practical performances of both algorithms.
\cref{fig:experiments} shows the regret of the algorithms compared to always playing the best-response strategy against $\bar{\vec{s}}$.

We also report the empirical performance of online MCCFR (as proposed in a side note by \citet{Lanctot09:Monte}) although, as we discussed at length in the introduction, online MCCFR---unlike the other two algorithms--- 1) is not an algorithm for the bandit optimization setting (as discussed in the introduction, it also needs to observe the actions of the opponent, while our algorithm and AHR do not; in the experiment we give online MCCFR that additional benefit), and 2) does not have a known guarantee of sublinear regret in this setting.
We ran each algorithm $100$ times. \cref{fig:experiments} shows these runs with thin light-colored lines. For the non-anytime algorithms (ours and AHR), we divided the desired runtime (e.g., 3 hours in Leduc poker) by the average time per iteration of each algorithm. We also plot the average regret across the $100$ runs of each algorithm with a thick line.

In all games, our method yielded lower regret than AHR at all times. In the matrix game and Kuhn poker, our algorithm converges to a smaller regret than online MCCFR, despite the fact that we do not get to observe the path of play and online MCCFR does. In the larger and deeper games (Goofspiel and Leduc poker), our algorithm still clearly exhibits the guaranteed $O(\sqrt{T})$ cumulative regret, but has higher regret than online MCCFR. Empirically, the regret cumulated by AHR seems to match the theoretical $\tilde{O}(\sqrt{T})$  guarantee well in the small games, but not in Goofspiel and Leduc poker. The reason for this is twofold. First, the runtime cost of each iteration of AHR in those latter games is roughly three orders of magnitude higher than either our method or online MCCFR. A significant fraction (roughly $20$\%) of that runtime is due to the fact that AHR needs to compute an eigendecomposition of a Hessian matrix of the log-barrier at the current point, an expensive operation whose overhead grows roughly cubically with the number of sequences in the game. We use the Eigen library \citep{Gael18:Eigen} to compute the eigendecomposition at each time $t$. Due to  this reason, AHR  performs significantly fewer iterations in the allotted time (three hours). The second issue that we identified with AHR is that it tends to suffer from serious numerical difficulties as the size of the TFSDM problem grows.

\begin{figure}[t]
      \centering
      \begin{tikzpicture}
          \node[] at (0,0) {\includegraphics[scale=.78]{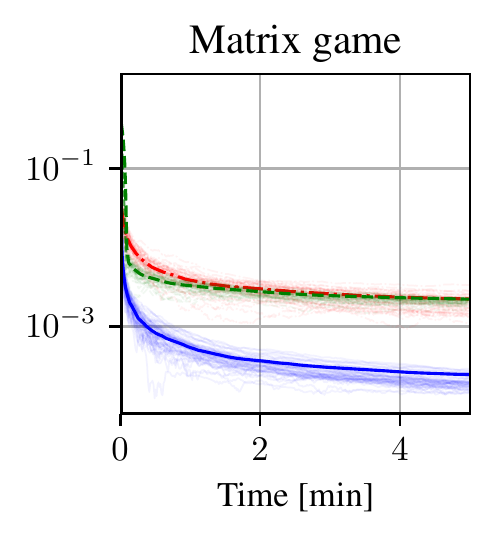}};
          \node[] at (4,0) {\includegraphics[scale=.78]{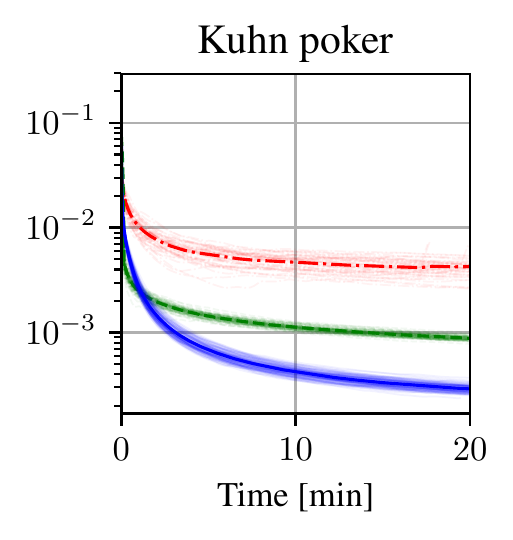}};
          \node[] at (0,-4.1){\includegraphics[scale=.78]{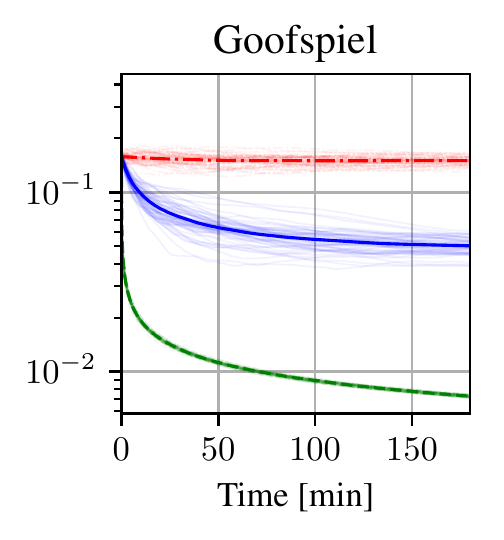}};
          \node[] at (4,-4.1){\includegraphics[scale=.78]{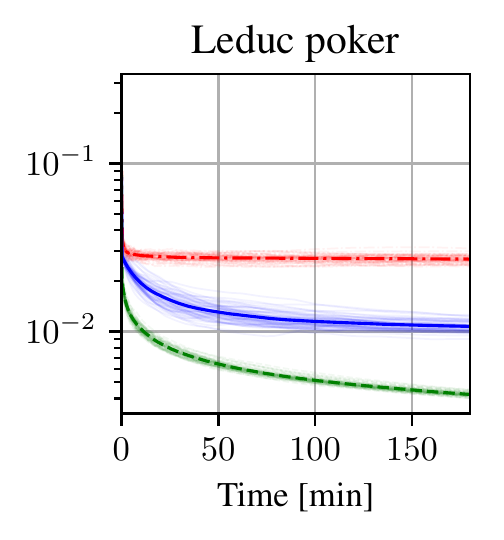}};
          \node[rotate=90] at (-2.2,0) {\small Average regret};
          \node[rotate=90] at (-2.2,-4.1) {\small Average regret};
      \end{tikzpicture}
      \vspace{-2mm}
      \newcommand{\linesty}[1]{\raisebox{.7mm}{\tikz\draw[very thick,#1] (0,0)--(.55,0);}}
        \definecolor{pred}{HTML}{FF0000}
        \definecolor{pblue}{HTML}{0000FF}
        \definecolor{pgreen}{HTML}{008000}
        \definecolor{ppurple}{HTML}{800080}
                \parbox{0.8\linewidth}{\centering\small
                \linesty{pred,dash dot}~AHR \citep{Abernethy08:Competing}\\
                \linesty{pblue}~Our method \qquad
                \linesty{pgreen,dashed}~Online MCCFR
            }
            \vspace{1mm}
      \caption{Evolution of the average regret in different bandit linear optimization algorithms (AHR and ours), as well as the online MCCFR algorithm (not an algorithm for bandit linear optimization).}
      \label{fig:experiments}
      \vspace{-3mm}
\end{figure}

    \section{Conclusions and Future Research}\label{sec:conclusion}

In this paper, we developed an algorithm for the bandit linear optimization on tree-form sequential decision making (TFSDM) problems.
Our bandit regret minimizer is superior to that of \citet{Abernethy08:Competing} both computationally (each iteration runs in linear time in the number of sequences in the problem) and in terms of cumulated regret (the regret is $O(\sqrt{T})$ instead of $O(\sqrt{T\log T})$.
We also presented the first implementations of bandit optimization for TFSDM. 
Our method combines and contributes a number of ideas and tools. First, we gave several new results concerning the local analytic properties of the dilated entropy regularizer (the leading regularizer for TFSDM). We use those analytic properties to obtain a stronger regret bound for the online mirror descent algorithm instantiated with the dilated entropy regularizer. Second, we study several properties of the natural sampling scheme for sequence-form strategies. Those properties are key to efficiently constructing an unbiased estimator of the loss vector $\vec{\ell}^t$ starting from the loss evaluation $(\vec{\ell}^t)\trans \vec{y}^t$ at a pure strategy $\vec{y}^t$. In order to construct the unbiased estimator, we extended and generalized an argument by~\citet{Bartlett08:High} to our context.
Finally, we combined the stronger regret bound for mirror descent together with the unbiased loss estimator to construct our bandit regret minimizer, by showing that the unbiased loss estimator has a dual norm that is bounded by a small time-independent constant.

A known weakness in~\citet{Abernethy08:Competing}, which is also a weakness in our approach, is that the bound on regret
(i) only holds in expectation as opposed to high probability, and
(ii) provides a guarantee on $\max_{\vec{z}\in\seqf} \bbE[R^T(\vec{z})]$ but \emph{not} on  $\bbE[\max_{\vec{z}\in\seqf} R^T(\vec{z})]$.
As discussed in the introduction, this weakness can be eliminated in theory if iterations are allowed to take exponential time in the number of sequences or by accepting a slower convergence rate.
Due to this weakness, our approach and that of~\citet{Abernethy08:Competing} sometimes does not work for equilibrium finding in games through self play.

We leave the problem of designing an algorithm for the bandit linear optimization problem for TFSDM that guarantees both $O(\sqrt{T})$ regret \emph{with high probability} and linear-time iterations as an open future direction. This would solve (i) and thereby also (ii). 
\citet{Abernethy09:Beating} presented a template for deriving such algorithms, but several pieces therein need to be instantiated to complete the proof of bounds. The theory in this paper offers solutions for some of those pieces for general TFSDM problems. Our regularizer, the sampling scheme, the construction of the loss estimates, and the use of local norms can be used within that general framework to provide high-probability results. So, our results may help solve the open problem for TFSDM in the future.

    \section*{Acknowledgments}
    This material is based on work supported by the National Science Foundation under grants IIS-1718457, IIS-1901403, and CCF-1733556, and the ARO under award W911NF2010081. Gabriele Farina is supported by a Facebook fellowship.

%


    \bibliography{dairefs}

\iftrue
    \clearpage
    \onecolumn
    \leftlinenumbers

    \appendix

    \addtolength{\hoffset}{1.25cm}
    \addtolength{\textwidth}{-2.5cm}
    \addtolength{\hsize}{-2.5cm}
    \addtolength{\linewidth}{-2.5cm}
    \addtolength{\columnwidth}{-2.5cm}

    \doparttoc\faketableofcontents\mtcaddpart[Appendix]
    \begin{center}
        \bfseries\LARGE Appendix
        \vskip 5mm
    \end{center}
    \renewcommand{\ptctitle}{}
    \parttoc
        In order to simplify the proofs, we will assume that no two nodes of the same type immediately follow each other. This assumption does not come at the cost of generality, since two consecutive decision points can always be consolidated into an equivalent one by combining their actions. Similarly, two consecutive observation points can be consolidated into an equivalent one by combining the available signals. Because of this assumption, we denote the set of all observation points that are immediately reachable after decision point $j$ as $\next{j} \defeq \{\rho(j,a): a \in A_j\} \setminus \{\terminalnode\}$. Similarly, the set of all decision points that are immediately reachable after observation point $k$ is $\next{k} \defeq \{\rho(k,s): s \in S_k\}$. We also assume without loss of generality that no signal is terminal, that is $\rho(k,s) \neq \terminalnode$. This assumption also does not come at the cost of generality, since it is always possible to attach a decision node with a single terminal action to simulate a terminal signal.

        \section{Additional Notation for Sequential Decision Processes}
        We now introduce additional notation:

        \paragraph{Number of decision nodes in a subtree ($N_v$).} The symbol $N_v$ denotes the number of decision points in the subtree rooted at $v$. If $v$ itself is a decision point, $v$ is included in the count.


        \paragraph{Descendants ($\succeq$).} A partial order $\succeq$ can be established on $\Sigma$ as follows: given two sequences $ja$ and $j'a'$ in $\Sigma$, $j'a' \succeq ja$ if and only if the (unique) path from the root node of the decision process to action $a'$ at decision point $j'$ passes through action $a$ at decision point $j$. Whenever $j'a' \succeq ja$, we say that $j'a'$ is a \emph{descendant} of $ja$.

        \paragraph{Subtree indicator ($\vec{u}_{ja}$).} Given a sequence $ja \in \Sigma$, we denote with $\vec{u}_{ja}$ the vector such that $(\vec{u}_{ja})_{j'a'} = 1$ if $j'a' \succeq ja$, and $(\vec{u}_{ja})_{j'a'} = 0$ otherwise.

        \paragraph{Parent sequence $(p_j)$.} Given a decision point $j \in \cJ$, we denote with $p_j$ its \emph{parent sequence}, defined as the last sequence (that is, decision point-action pair) encountered on the path from the root of the tree to decision point $j$. If the agent does not act before $j$ (i.e., $j$ is the root of the decision process or only observation points are encountered on the path from the root to $j$), we let $p_j = \emptyseq$.

    \subsection{Inductive Definition of $\pure$}
        The set $\pure$ can be constructed recursively in a bottom-up fashion, as follows:
    \begin{itemize}[nolistsep,itemsep=1mm,leftmargin=8mm]
        \item At each terminal decision point $j \in \cJ$, where actions $\{a_1, \dots, a_m\}$ are terminal (that is, $\rho(j,a)=\terminalnode$) and actions $\{a_{m+1}, \dots, a_n\}$ are not, the set of pure strategies is the set
      \begin{align}
        \pure_j \defeq& \mleft\{\begin{pmatrix}\vec{e}_i\\\vec{0}\\\vdots\\\vec{0}\end{pmatrix}: i =1 ,\dots, n\mright\} \cup
        \mleft\{\begin{pmatrix}\vec{e}_{m+1}\\\vec{x}_{k_{m+1}}\\\vdots\\\vec{0}\end{pmatrix} : \vec{x}_{k_{m+1}} \in \pure_{k_{m+1}} \mright\} \cup \dots \cup
        \mleft\{\begin{pmatrix}\vec{e}_{n}\\\vec{0}\\\vdots\\\vec{x}_{k_{n}}\end{pmatrix} : \vec{x}_{k_n} \in \pure_{k_n} \mright\},\label{eq:pure j}
      \end{align}
            where $\vec{e}_i$ the $i$-th canonical basis vector. So, in particular, if all actions at $j$ are terminal,
            \[
                \pure_j = \{\vec{e}_1,\dots,\vec{e}_n\}.
            \]
        \item At each observation point $k \in \mathcal{K}$, the set of pure strategies is simply the Cartesian product of strategies for each of the child subtrees:
            \begin{equation}\label{eq:pure k}
                \pure_k \defeq \pure_{j_1} \times \dots \times \pure_{j_n},
            \end{equation}
      where $\{j_1, \dots, j_n\} = \cC_k$ are the decision points immediately reachable after $k$.
    \end{itemize}

    \subsection{Inductive Definition of $\seqf$}
        The set of mixed sequence-form strategies can also equivalently constructed inductively along the tree structure:
    \begin{itemize}[nolistsep,itemsep=1mm,leftmargin=8mm]
          \item At each terminal decision point $j \in \cJ$, where actions $\{a_1, \dots, a_m\}$ are terminal (that is, $\rho(j,a)=\terminalnode$) and actions $\{a_{m+1}, \dots, a_n\}$ are not, we first fix a distributions over the actions in $A_j$ and then recurse:
      \begin{align}
        \seqf_j &\defeq \mleft\{\begin{pmatrix}\lambda_1\\\dots\\\lambda_n\\\lambda_{m+1} \vec{x}_{k_{m+1}}\\\vdots\\\lambda_n \vec{x}_{k_n}\end{pmatrix}: (\lambda_1,\dots,\lambda_n) \in \Delta^n, \vec{x}_i \in \seqf_{k_i}\ \forall i = m+1,\dots,n\mright\},\label{eq:structure decision point}
      \end{align}
      where $\{k_{m+1}, \dots, k_n\} = \cC_j$ are the observation points immediately reachable after $j$. So, in particular, if all actions at $j$ are terminal, we have
            \begin{equation}\label{eq:structure simplex}
                \seqf_j \defeq \Delta^{|A_j|}.
            \end{equation}
        \item At each observation point $k \in \mathcal{K}$, the set of mixed strategies is the Cartesian product of mixed strategies for each of the child subtrees:
            \begin{equation}\label{eq:structure observation point}
                \seqf_k \defeq \seqf_{j_1} \times \dots \times \seqf_{j_n},
            \end{equation}
      where $\{j_1, \dots, j_n\} = \cC_k$ are the decision points immediately reachable after $k$.
    \end{itemize} 
    \newpage
    \section{Properties of the Dilated Entropy Distance-Generating Function}
\label{app:dgf}

\subsection{Preliminaries}

\textbf{Gradient Computation}\quad We look at the computation of the gradient of $\regu$ at a generic point $\vec{z}\in\Rpp^{|\Sigma|}$. Some elementary algebra reveals that
\begin{align}
  \frac{\partial \regu}{\partial z_{ja}}(\vec{z}) &= w_j\mleft(1 + \log \frac{{z}_{ja}}{{z}_{p_j}}\mright) + \sum_{j'\in\next{\rho(j,a)}}\hspace{-3mm}w_{j'}\mleft(\log |A_{j'}| - \sum_{a'\in A_{j'}}\!\!\frac{{z}_{j'a'}}{{z}_{ja}}\mright)\label{eq:gradient}
\end{align}
for every decision point-action pair $ja \in \Sigma$.
Hence, we can compute $\nabla\regu(\vec{z})$ at any $\vec{z}\in\Rpp^{|\Sigma|}$ in one linear-time traversal of the sequential decision tree as in \cref{algo:gradient}.

\textbf{Fenchel Arg-Conjugate}\quad The Fenchel conjugate of $\regu$ on $\co \pure$ is defined as
\[
  \regu^*: \vec{z} \mapsto \max_{\hat{\vec{x}} \in \co \pure \cap \Rpp^{|\Sigma|}}\{\vec{z}\trans \hat{\vec{x}} - \regu(\hat{\vec{x}})\}
\]
for any $\vec{z} \in \bbR^{|\Sigma|}$. Since the domain of the maximization is not compact, it is not \emph{a priori} obvious that $\regu^*$ is well defined for all $\vec{z}$. Existence of the solution to the maximization problem can be easily proved by exhibiting a point in the domain with zero gradient. Uniqueness follows from the strict convexity of $\regu$ on $\bbR_{>0}^n$. A proof of that fact can be found in the original work by~\citet{Hoda10:Smoothing}.

It is well-known (and easy to check via a straightforward application of Danskin's theorem) that
\begin{equation}\label{eq:argconjugate}
  \nabla \regu^*: \vec{z} \mapsto \argmax_{\hat{\vec{x}} \in \co \pure \cap \Rpp^{|\Sigma|}}\{\vec{z}\trans \hat{\vec{x}} - \regu(\hat{\vec{x}})\}.
\end{equation}
For this reason, we call $\nabla\regu^*$ the \emph{Fenchel arg-conjugate} function of $\regu$ on $\co \pure$. Of course, for any $\vec{z} \in \bbR^{|\Sigma|}$ one can efficiently compute the value of $\regu^*(\vec{z})$ given $\vec{x}^* \defeq \nabla \regu^*(\vec{z})$ (which is guaranteed to be in $\in \Rpp^{|\Sigma|}$) by direct substitution as $\vec{z}\trans \vec{x}^* - \regu(\vec{x}^*)$. In \cref{algo:argconjugate} we give a linear-time algorithm for computing $\nabla\regu^*(\vec{z})$. We refer the reader to the original work by \citet{Hoda10:Smoothing} for a proof of correctness.

\begin{figure}[ht]
    \begin{minipage}[t]{.47\linewidth}
    \begin{algorithm}[H]
        \caption{\textsc{Gradient}($\vec{z}$)}
        \label{algo:gradient}\DontPrintSemicolon
        \KwIn{$\vec{z} \in \bbR^{|\Sigma|}$}
        \KwOut{The value of $\nabla\regu(\vec{z})$}
        \BlankLine
            $\vec{g} \gets \vec{0} \in \bbR^{|\Sigma|}$\;
            \For{$j \in \cJ$ in bottom-up order}{
                \For{$a \in A_j$}{
                    $g_{ja} \gets g_{ja} + w_j \mleft(1 + \log \frac{z_{ja}}{z_{p_j}}\mright)$\;
                    \textbf{if} $p_j \neq \emptyseq$ \textbf{then}
                        $g_{p_j} \gets g_{p_j} - w_j \frac{z_{j'a'}}{z_{ja}}$\;
                }
                \textbf{if} $p_j \neq \emptyseq$ \textbf{then}
                    $g_{p_j} \gets g_{p_j} + w_j \log |A_j|$\;
            }
            \KwRet{$\vec{g}$}\;
    \end{algorithm}
    \end{minipage}\hfill
    \begin{minipage}[t]{.47\linewidth}
    \begin{algorithm}[H]
        \caption{\textsc{ArgConjugate}($\vec{z}$)}
        \label{algo:argconjugate}\DontPrintSemicolon
        \KwIn{$\vec{z} \in \bbR^{|\Sigma|}$}
        \KwOut{The value of $\nabla\regu^*(\vec{z})$}
        \BlankLine
            $\vec{x}^* \gets \vec{0} \in \bbR^{|\Sigma|}$\;
            \For{$j \in \cJ$ in bottom-up order}{
                $s \gets 0$\;
                \For{$a \in A_j$}{
                    $x^*_{ja} \gets \exp\{\frac{z_{ja}}{w_j}\}$\;
                    $s \gets s + x^*_{ja}$\;
                }
                $v \gets  w_j \log |A_j|$\;
                \For{$a \in A_j$}{
                    $x^*_{ja} \gets \frac{x^*_{ja}}{s}$\Comment*[r]{\color{commentcolor}Normalization step]}
                    $v \gets v + z_{ja} x^*_{ja} - x^*_{ja}\log x^*_{ja}$\;
                }
                $z_{ja} \gets z_{ja} + v$\;
            }
            \For{$j\in \cJ$ in top-down order}{
                \For{$a\in A_j$}{
                    $x^*_{ja} \gets x^*_{ja} \cdot x^*_{p_j}$\;
                }
            }
            \KwRet{$\vec{x}^*$}\;
    \end{algorithm}
    \end{minipage}
\end{figure}

\begin{observation}\label{obs:xbar positive}
    At all times $t$, the decision produced by \cref{algo:R tilde} satisfies
    $\xbar^t \in \Rpp^{|\Sigma|}.$
\end{observation}

\subsection{Local Primal Norm}\label{app:local norms}

\begin{lemma}[\citet{Ling19:Large}]
\label{lem:hessian}
The Hessian $\nabla^2 \regu(\vec{z})$ of the dilated entropy DGF at $\vec{z} \in \seqf$ is given by:
\begin{equation*}
\frac{\partial^2}{\partial z_{ja} \partial z_{j' a'}} \regu(\vec{z}) =
  \begin{cases}
    \displaystyle
        \frac{w_j + w_{\rho(j,a)}}{z_{ja}} & \text{if } {ja} = {j' a'}
    \\[4mm]
    \displaystyle
        -\frac{w_{j'}}{z_{ja}} & \text{if } {ja} = {p_{j'}} \text{ and } p_{j'} \neq \emptyseq
    \\[4mm]
    \displaystyle
        -\frac{w_{j}}{z_{p_j}} & \text{if } {j'a'} = {p_j} \text{ and } p_{j} \neq \emptyseq
    \\[2mm]
    0 & \text{otherwise}.
  \end{cases}
\end{equation*}
\end{lemma}

\begin{restatable}{lemma}{lemboundprimal}\label{lem:bound primal}
      Let $\xbar \in \co \pure$ and  $\vec{z} \in \Rp^{|\Sigma|}$. Then,
      $\displaystyle
        \pn{\vec{z}}{\xbar}^2 \le \frac{3}{2}\sum_{j\in\cJ}\sum_{a\in A_j}\frac{w_j}{{x}_{ja}}z_{ja}^2.
      $
    \end{restatable}
    \begin{proof}
      Using the explicit expression of the Hessian of the dilated entropy regularizer (\cref{lem:hessian}) we can write
      \begin{align}
          \| \vec{z} \|_{\xbar}^2  &= \sum_{j\in\cJ}\sum_{a\in A_j} \frac{w_{j} + w_{\rho(j,a)}}{{x}_{ja}} z_{ja}^2 - 2 \sum_{j\in\cJ}\sum_{a\in A_j}\sum_{j' \in \next{\rho(j,a)}} \sum_{a'\in A_{j'}} \frac{w_{j'}}{{x}_{ja}} z_{j'a'} z_{ja} \nonumber\\
            &\le \sum_{j\in\cJ}\sum_{a\in A_j} \frac{w_{j} + w_{\rho(j,a)}}{{x}_{ja}} z_{ja}^2,\label{eq:pn bound}
      \end{align}
      where the inequality holds since $\vec{z}\in\Rp^{|\Sigma|}$. By definition of $w_j$ (\cref{def:dilated entropy}), we have for all $ja\in\Sigma$
      \begin{align}
        w_j + w_{\rho(j,a)} &\le w_j + \max_{a'\in A_j} w_{\rho(j,a')} \nonumber\\
            &= 2 + 3\max_{a' \in A_j} w_{\rho(j,a')} \nonumber\\
            &\le 3 + 3\max_{a' \in A_j} w_{\rho(j,a')}= \frac{3}{2} w_j\label{eq:ineq 3 over 2}.
      \end{align}
      Plugging \eqref{eq:ineq 3 over 2} into \eqref{eq:pn bound} yields the statement.
    \end{proof}

\subsection{Local Dual Norm}

\begin{restatable}{lemma}{leminvhessian}\label{lem:inverse hessian}
      Let $\xbar \in \co \pure$ be a sequence-form strategy. The inverse Hessian $(\nabla^{2}\regu)^{-1}(\xbar)$ at $\xbar$ can be expressed as:
    \begin{align}
    (\nabla^2 \regu)^{-1}(\xbar) &= \sum_{j\in\cJ}\sum_{a \in A_j} \frac{(\xbar\circ \vec{u}_{ja})(\xbar\circ \vec{u}_{ja})^\top}{w_j {x}_{ja}},
    \end{align}
    where $\circ$ denotes componentwise product of vectors.
    \end{restatable}
    \begin{proof}
    Let
    \[
      \mat{H} \defeq \sum_{j\in\cJ}\sum_{a\in A_j} \frac{(\xbar\circ \vec{u}_{ja})(\xbar\circ \vec{u}_{ja})^\top}{w_j {x}_{ja}}
    \]
    be the proposed inverse Hessian matrix. We will prove that $\mat{H} = (\nabla^2 \regu)^{-1}(\xbar)$ by showing that $\nabla^2 \regu(\xbar) \cdot \mat{H} = \mat{I}$ is the identity matrix. We break the proof into two steps:
    \begin{itemize}[leftmargin=*]
      \item \textbf{Step one.} First, we show that for all sequences $ja\in\Sigma$ and $j'a' \in \Sigma$,
        \begin{equation}\label{eq:sigma}
          \circled{A} \defeq \mleft[\nabla^2\regu(\xbar)\cdot(\xbar \circ \vec{u}_{ja})\mright]_{j'a'} =
            \begin{cases}
                w_j  & \text{if } j'a' = ja\\[2mm]
                \displaystyle- w_{j}\frac{{x}_{ja}}{{x}_{p_j}} & \text{if } p_j = j'a'\\[2mm]
                0 & \text{otherwise}.
            \end{cases}
        \end{equation}
        In order to prove \eqref{eq:sigma}, we start from \cref{lem:hessian}:
        \begin{align*}
             \circled{A} &= \sum_{j'' \in \cJ}\sum_{a'' \in A_{j''}}
                \frac{\partial^2 \regu(\xbar)}{\partial {x}_{j'a'} \partial {x}_{j'' a''}} \cdot (\xbar \circ \vec{u}_{ja})_{j''a''}\\
            &= (w_{j'} + w_{\rho(j',a')}) \!\cdot\! (\vec{u}_{ja})_{j'a'} - w_{j'} \!\cdot\! (\vec{u}_{ja})_{p_{j'}} - \!\!\!\sum_{j'' \in \next{\rho(j',a')}}\sum_{a''\in A_{j''}} \! \frac{w_{j''}}{{x}_{j'a'}}{x}_{j''a''} \!\cdot\! (\vec{u}_{ja})_{j''a''}.
        \end{align*}
        We now distinguish four cases, based on how $ja$ relates to $p_{j'}$, $j'a'$, and $j''a''$:
        \begin{itemize}
          \item \emph{First case:} $p_{j'} \succeq ja$, that is $p_{j'}, j'a'$ and $j''a''$ are all descendants of $ja$. Consequently,
              \[
                (\vec{u}_{ja})_{p_{j'}} = (\vec{u}_{ja})_{j'a'} = (\vec{u}_{ja})_{j''a''} = 1
              \]
              for all $j'' \in \next{\rho(j',a')}$ and $a'' \in A_{j''}$. Hence,
              \begin{align*}
                \circled{A} &= w_{j'} + w_{\rho(j',a')} - w_{j''} - \sum_{j'' \in \next{\rho(j',a')}}\sum_{a'' \in A_{j''}} \frac{w_{j''}}{{x}_{j'a'}}{x}_{j''a''} \\
                    &=  w_{\rho(j',a')} - \sum_{j'' \in \next{\rho(j',a')}} \mleft(w_{j''} \sum_{a'' \in A_{j''}} \frac{{x}_{j''a''}}{{x}_{j'a'}}\mright) \\&=  w_{\rho(j',a')} - \sum_{j'' \in \next{\rho(j',a')}} w_{j''}
                    = 0.
              \end{align*}
          \item \emph{Second case:} $ja = j'a'$. In this case, $(\vec{u}_{ja})_{p_{j'}} = 0$, while $( \vec{u}_{ja})_{j'a'} = (\vec{u}_{ja})_{j''a''} = 1$ for all $j'' \in \next{\rho(j',a')}$ and $a'' \in A_{j''}$. Hence,
              \begin{align*}
                \circled{A} &= w_{j'} + w_{\rho(j',a')} - \sum_{j'' \in \next{\rho(j',a')}}\sum_{a'' \in A_{j''}} \frac{w_{j''}}{{x}_{j'a'}}{x}_{j''a''} \\
                    &=  w_{j'} + w_{\rho(j',a')} - \sum_{j'' \in \next{\rho(j',a')}} \mleft(w_{j''} \sum_{a'' \in A_{j''}} \frac{{x}_{j''a''}}{{x}_{j'a'}}\mright) \\
                    &=  w_{j'} + w_{\rho(j',a')} - \sum_{j'' \in \next{\rho(j',a')}} w_{j''} = w_{j'} = w_j.
              \end{align*}
          \item \emph{Third case:} $p_j = j'a'$ (that is, $ja$ immediately follows $j'a'$). Then,
              $
                \circled{A} = - w_{j}\frac{{x}_{ja}}{{x}_{p_j}}
              $
          \item \emph{Otherwise}, $j''a'' \not\succeq ja$ for all $j'' \in \next{\rho(j',a')}$ and $a'' \in A_{j''}$, and therefore
              $
                \circled{A} = 0.
              $
        \end{itemize}
      \item \textbf{Step two.} Given $\sigma \in \Sigma \cup \{\emptyseq\}$, let $\vec{1}_{\sigma} \in \bbR^{|\Sigma|}$ denote the vector that has a $1$ in the entry corresponding to sequence $\sigma$, and $0$ everywhere else (in particular, $\vec{1}_\emptyseq = \vec{0}$). Then,~\eqref{eq:sigma} can be rewritten as
        \begin{equation*}
            \frac{\nabla^2\regu(\xbar)\cdot(\xbar \circ \vec{u}_{ja})}{w_j {x}_{ja}} = \frac{1}{{x}_{ja}}\vec{1}_{ja} - \frac{1}{{x}_{p_j}}\vec{1}_{p_j}.
        \end{equation*}
        Therefore,
        \begin{align*}
            \nabla^2\regu(\xbar) \cdot \mat{H} &= \sum_{j\in J}\sum_{a\in A_j} \frac{\nabla^2\regu(\xbar)\cdot(\xbar \circ \vec{u}_{ja})}{w_j {x}_{ja}}\cdot (\xbar \circ \vec{u}_{ja})\trans = \sum_{j\in \cJ}\sum_{a\in A_j} \mleft(\frac{1}{{x}_{ja}}\vec{1}_{ja} - \frac{1}{{x}_{p_j}}\vec{1}_{p_j}\mright) \cdot (\xbar \circ \vec{u}_{ja})\trans.\\
                &= \sum_{j\in \cJ}\sum_{a\in A_j}\frac{1}{{x}_{ja}}\vec{1}_{ja}\cdot \mleft(\xbar \circ \vec{u}_{ja} - \sum_{j'\in\next{\rho(j,a)}}\sum_{a'\in A_{j'}} \xbar \circ \vec{u}_{j'a'}\mright)\trans\\
                &= \sum_{j\in \cJ}\sum_{a\in A_j}\frac{1}{{x}_{ja}}\vec{1}_{ja}\cdot \mleft[\xbar \circ \mleft(\vec{u}_{ja} - \sum_{j'\in\next{\rho(j,a)}}\sum_{a'\in A_{j'}} \vec{u}_{j'a'}\mright)\mright]\trans.
        \end{align*}
      Using the definition of $\vec{u}_{ja}$, we obtain
        \begin{align*}
            \nabla^2\regu(\xbar) \cdot \mat{H} &= \sum_{j\in \cJ}\sum_{a\in A_j}\frac{1}{{x}_{ja}}\vec{1}_{ja}\cdot (\xbar \circ \vec{1}_{ja})\trans\\
                &= \sum_{j\in \cJ}\sum_{a\in A_j}\vec{1}_{ja}\vec{1}_{ja}\trans\\
                &=\mat{I},
        \end{align*}
        as we wanted to show.\qedhere
    \end{itemize}
    \end{proof}\clearpage

    \begin{restatable}{corollary}{cordualnorm}\label{cor:dual norm}
     Let $\xbar \in \co \pure$ be a sequence-form strategy, and let $\vec{z}\in\bbR^{|\Sigma|}$. The local dual norm of $\vec{z}$ satisfies
    $\displaystyle        \dn{\vec{z}}{\xbar}^{2} = \sum_{j\in\cJ}\sum_{a\in A_j} \frac{(\vec{u}_{ja}\trans(\vec{z} \circ \xbar))^2}{w_j {x}_{ja}}.
    $
    \end{restatable}
\begin{proof}
    By definition of local dual norm, using \cref{lem:inverse hessian}), and applying simple algebraic manipulations:
    \begin{align*}
         \dn{\vec{z}}{\xbar}^{2} &= \vec{z}^\top \mleft( \sum_{j\in\cJ}\sum_{a\in A_j} \frac{(\xbar\circ \vec{u}_{ja})(\xbar\circ \vec{u}_{ja})^\top}{w_j {x}_{ja}} \mright) \vec{z} \\&= \sum_{j\in\cJ}\sum_{a\in A_j} \frac{(\vec{z}^\top (\xbar\circ \vec{u}_{ja}))^2}{w_j {x}_{ja}} \\&= \sum_{j\in\cJ}\sum_{a\in A_j} \frac{(\vec{u}_{ja}\trans(\vec{z} \circ \xbar))^2}{w_j {x}_{ja}}.\qedhere
    \end{align*}
\end{proof}

    \newpage
    \section{Analysis of Mirror Descent using Dilated Entropy DGF}
\label{app:R tilde}
    We study some properties of \cref{algo:R tilde}. The central result, \cref{thm:omd regret bound}, gives a bound on the cumulative regret expressed in term of (dual) local norms centered at the iterates produced by online mirror descent.
    Our first step is to introduce the ``intermediate'' iterate
    \begin{align}
        \tilde{\vec{x}}^1 &\defeq \argmin_{\hat{\vec{x}}\in\Rpp^{|\Sigma|}} \regu(\hat{\vec{x}}),\label{eq:x tilde init}\\
      \tilde{\vec{x}}^{t+1} &\defeq \argmin_{\hat{\vec{x}} \in \Rpp^{|\Sigma|}} \mleft\{(\eta\tilde{\vec{\ell}}^{t} - \nabla \regu(\xbar^{t}))\trans \hat{\vec{x}} + \regu(\hat{\vec{x}})\mright\},\qquad t \ge 1\label{eq:x tilde}
    \end{align}
    which differs from an arg-conjugate (\cref{eq:argconjugate}) in that the minimization problem is unconstrained. In \cref{app:existence x tilde} we prove that $\tilde{\vec{x}}^{t+1}$ is well-defined, in the sense that it always exists unique. In \cref{app:xtilde and omd} we show that it is convenient for analyzing the regret accumulated by online mirror descent~\cite{Abernethy09:Beating}.

    \subsection{Existence and Uniqueness of the Intermediate Iterate}\label{app:existence x tilde}

Using the structure of the dilated entropy DGF together with that of the game tree, we prove the following properties, which will be fundamental in the analysis of online mirror descent based on local norms.

\begin{restatable}{lemma}{lemproxsolution}\label{lem:prox solution}
    At all times $t \ge 1$, each intermediate iterate $\tilde{\vec{x}}^{t+1}$ exists, is unique, and satisfies, for all $ja \in \Sigma$:
        \begin{equation}\label{eq:prox step}
          \frac{\tilde{x}^{t+1}_{ja}}{\tilde{x}^{t+1}_{p_j}} = \frac{{x}^{t}_{ja}}{{x}^{t}_{p_j}}\exp\mleft\{-\eta \frac{\ell^{t}_{ja}}{w_j}-\frac{w_{\rho(j,a)}}{w_j} + \frac{\xi^{t+1}_{ja}}{w_j}\mright\},
        \end{equation}
        where
        \[
           \xi^{t+1}_{ja} \defeq \sum_{j'\in\next{\rho(j,a)}} \hspace{-2mm} w_{j'} \hspace{-2mm} \sum_{a'\in A_{j'}}\!\! \frac{\tilde{x}^{t+1}_{j'a'}}{\tilde{x}^{t+1}_{ja}}.
        \]
    \end{restatable}
    \begin{proof}
    The domain $\Rpp^{|\Sigma|}$ of the optimization problem~\eqref{eq:x tilde} is an open set, and the objective function is a differentiable convex function. Hence, a generic point $\tilde{\vec{x}}^{t+1} \in \Rpp^{|\Sigma|}$ is a (global) minimizer if and only if the gradient of the objective function at $\tilde{\vec{x}}^{t+1}$ is $\vec{0}$. For this, we start by setting the gradient of the objective function to $\vec{0}$:
\[
  \eta\vec{\ell}^{t} - \nabla \regu(\xbar^{t}) + \nabla\regu(\tilde{\vec{x}}^{t+1}) = \vec{0}.
\]
Substituting the expression for $\nabla\regu$ (Equation~\ref{eq:gradient}) into the optimality condition yields
\begin{align*}
    &\eta \ell_{ja} - w_j \log\frac{{x}^{t}_{ja}}{{x}^{t}_{p_j}}
    + \hspace{-3mm}\sum_{j'\in\next{\rho(j,a)}}\hspace{-3mm}w_{j'}\hspace{-1mm}\sum_{a'\in A_{j'}} \frac{{x}^{t}_{j'a'}}{{x}^{t}_{ja}}+ w_j \log\frac{\tilde{x}^{t+1}_{ja}}{\tilde{x}^{t+1}_{p_j}} -\hspace{-3mm}\sum_{j'\in\next{\rho(j,a)}}\hspace{-3mm}w_{j'}\hspace{-1mm}\sum_{a'\in A_{j'}} \frac{\tilde{x}^{t+1}_{j'a'}}{\tilde{x}^{t+1}_{ja}} = 0
\end{align*}
for all $ja \in \Sigma$. Using the fact that $\xbar^{t} \in \seqf$, we can write $\sum_{a'\in A_{j'}}\frac{{x}^{t}_{j'a'}}{{x}^{t}_{ja}} = 1$ and simplify the above condition into
\begin{align*}
    &\eta \ell_{ja} - w_j \log\frac{{x}^{t}_{ja}}{{x}^{t}_{p_j}}
    + w_j \log\frac{\tilde{x}^{t+1}_{ja}}{\tilde{x}^{t+1}_{p_j}} + w_{\rho(j,a)} -\hspace{-3mm}\sum_{j'\in\next{\rho(j,a)}}\hspace{-3mm}w_{j'}\sum_{a'\in A_{j'}} \frac{\tilde{x}^{t+1}_{j'a'}}{\tilde{x}^{t+1}_{ja}} = 0,
\end{align*}
where we used the equality $w_{\rho(j,a)} = \sum_{j' \in \next{\rho(j,a)}} w_{j'}$ (Equation~\ref{def:dilated entropy}). Rearranging the terms we conclude that the gradient of the objective function of~\eqref{eq:x tilde} is $\vec{0}$ if and only if, for all $ja\in\Sigma$,
    \[
        \frac{\tilde{x}^{t+1}_{ja}}{\tilde{x}^{t+1}_{p_j}} = \frac{{x}^{t}_{ja}}{{x}^{t}_{p_j}}\exp\mleft\{-\eta \frac{\ell^{t}_{ja}}{w_j}-\frac{w_{\rho(j,a)}}{w_j} + \frac{\xi^{t+1}_{ja}}{w_j}\mright\},\qquad \xi^{t+1}_{ja} \defeq \sum_{j'\in\next{\rho(j,a)}} \hspace{-2mm} w_{j'} \hspace{-2mm} \sum_{a'\in A_{j'}}\!\! \frac{\tilde{x}^{t+1}_{j'a'}}{\tilde{x}^{t+1}_{ja}},
    \]
    which is~\eqref{eq:prox step}. Crucially, the previous equation uniquely defines a point $\tilde{\vec{x}}^{t+1}$. Indeed, by applying~\eqref{eq:prox step} at any root decision points $j$ (that is, decision points where $p_j = \emptyseq$), using the fact that by definition $\tilde{x}^{t+1}_\emptyseq = {x}^{t+1}_\emptyseq = 1$, we can compute all sequences $ja$ ($a\in A_j$) uniquely as
    \[
        \tilde{x}^{t+1}_{ja} = {x}^t_{ja} \exp\mleft\{-\eta \frac{\ell^{t}_{ja}}{w_j}-\frac{w_{\rho(j,a)}}{w_j} + \frac{\xi^{t+1}_{ja}}{w_j}\mright\} > 0.
    \]
    Now that all $\tilde{x}^{t+1}_{ja}$ have been computed for root decision points $j$, we can inductively compute them for the decision points at depth $2$, then those at depth $3$, and so on. Since all entries in $\xbar^t$ are strictly positive (\cref{obs:xbar positive}), it is immediate to see inductively that all entries of $\tilde{\vec{x}}^{t+1}$ are strictly positive. Hence, the unique vector $\tilde{\vec{x}}^{t+1}$ that makes the gradient of the objective function $\vec{0}$ belongs to the open domain of the optimization, $\Rpp^{|\Sigma|}$. Hence, it is the unique minimizer of~\eqref{eq:x tilde}, as we wanted to show.
    \end{proof}

    \subsection{Relationships with Regret Bound of Mirror Descent}\label{app:xtilde and omd}

It is known that at all $t$, $\xbar^{t} = \argmin_{\hat{\vec{x}}\in\Rpp^{|\Sigma|}} \div{\hat{\vec{x}}}{\tilde{\vec{x}}^t}$, where $\div{\cdot}{\cdot}$ denotes the Bregman divergence induced by $\regu$, that is,
    \[
     \div{\vec{x}}{\vec{c}} \defeq \regu(\vec{x}) - \regu(\vec{c}) - \nabla\regu(\vec{c})^\top(\vec{x}-\vec{c}), \quad \vec{x},\vec{c}\in\Rpp^n.
    \]

    We prove that result here for completeness.
    \begin{lemma}
        At all times $t$, the iterate $\xbar^t$ produced by \cref{algo:R tilde} satisfies
    \begin{equation}\label{eq:tilde to bar}
        \xbar^{t} = \argmin_{\hat{\vec{x}}\in\Rpp^{|\Sigma|}} \div{\hat{\vec{x}}}{\tilde{\vec{x}}^t}.
    \end{equation}
    \end{lemma}
    \begin{proof}
        By induction:
        \begin{itemize}
            \item Base case ($t = 1$). Since $\bbR^n_{>0}$ is an open set, \eqref{eq:x tilde init} implies that
            \[
                \nabla\regu(\tilde{\vec{x}}^1) = \vec{0}.
            \]
            Hence,
            \[
                \bar{\vec{x}}^1 = \argmin_{\hat{\vec{x}}\in\seqf \cap \bbR_{>0}^{|\Sigma|}} \div{\hat{\vec{x}}}{\tilde{\vec{x}}^1} = \argmin_{\hat{\vec{x}}\in\seqf \cap \bbR_{>0}^{|\Sigma|}} \{ \regu(\hat{\vec{x}}) - \regu(\tilde{\vec{x}}^1) \} = \argmin_{\hat{\vec{x}}\in\seqf \cap \bbR_{>0}^{|\Sigma|}}  \regu(\hat{\vec{x}}),
            \]
            which is exactly the first iterate returned by \cref{algo:R tilde}.
            \item Inductive step. We assume that the iterates produced up to time $t$ by \cref{algo:R tilde} coincide with those produced up to time $t$ by \eqref{eq:tilde to bar}. We will prove that the iterates coincide at time $t+1$ as well. Since $\bbR^n_{>0}$ is an open set, from \eqref{eq:x tilde}, $\tilde{\vec{x}}^{t + 1}$ must be a critical point of the objective function, that is
            \begin{align*}
                \vec{0} &= \nabla_{\hat{\vec{x}}}\mleft[\eta(\tilde{\vec{\ell}}^t)^\top \hat{\vec{x}} + \div{\hat{\vec{x}}}{\bar{\vec{x}}^t}\mright](\tilde{\vec{x}}^{t+1})= \eta\tilde{\vec{\ell}}^t + \nabla\regu(\tilde{\vec{x}}^{t+1}) - \nabla\regu(\bar{\vec{x}}^{t}),
            \end{align*}
            from which we find that
            \[
                \nabla\regu(\tilde{\vec{x}}^{t+1}) = \nabla\regu(\bar{\vec{x}}^{t})-\eta\tilde{\vec{\ell}}^t.
            \]
            Hence, expanding \eqref{eq:tilde to bar}, we find
            \begin{align*}
                \bar{\vec{x}}^{t+1} &= \argmin_{\hat{\vec{x}}\in\seqf \cap \bbR_{>0}^{|\Sigma|}} \div{\hat{\vec{x}}}{\tilde{\vec{x}}^{t+1}}\\
                    &= \argmin_{\hat{\vec{x}}\in\seqf \cap \bbR_{>0}^{|\Sigma|}} \mleft\{\regu(\hat{\vec{x}}) - \regu(\tilde{\vec{x}}^{t+1}) - \nabla\regu(\tilde{\vec{x}}^{t+1})^\top(\hat{\vec{x}} - \tilde{\vec{x}}^{t+1})\mright\}\\
                    &= \argmin_{\hat{\vec{x}}\in\seqf \cap \bbR_{>0}^{|\Sigma|}} \mleft\{\regu(\hat{\vec{x}}) - \regu(\tilde{\vec{x}}^{t+1}) - (\nabla\regu(\bar{\vec{x}}^{t})-\eta\tilde{\vec{\ell}}^t)^\top(\hat{\vec{x}} - \tilde{\vec{x}}^{t+1})\mright\}\\
                    &= \argmin_{\hat{\vec{x}}\in\seqf \cap \bbR_{>0}^{|\Sigma|}} \mleft\{\regu(\hat{\vec{x}}) -   (\nabla\regu(\bar{\vec{x}}^{t})-\eta\tilde{\vec{\ell}}^t)^\top\hat{\vec{x}}\mright\}\\
                    &=\textsc{ArgConjugate}(\nabla\regu(\bar{\vec{x}}^{t})-\eta\tilde{\vec{\ell}}^t),
            \end{align*}
            which completes the inductive step.\qedhere
        \end{itemize}
    \end{proof}

    The following known result will be central in the proof of \cref{thm:omd regret bound}.

        \begin{lemma}[\citet{Rakhlin09:Lecture}, Lemma 13]\label{lem:regret bound known}
            Online mirror descent satisfies, at all times $T$ and for all mixed strategies $\vec{z} \in \seqf$, the regret bound
            \[
                R^T(\vec{z}) \le \frac{\regu(\vec{z})}{\eta} + \sum_{t=1}^T (\tilde{\vec{\ell}}^{t})\trans(\xbar^{t} - \tilde{\vec{x}}^{t+1}).
            \]
        \end{lemma}

    \cref{lem:regret bound known} is usually proved by leveraging the Legendreness of the regularizer. Since we have not shown that $\regu$ is Legendre, we provide an alternative, standalone proof of \cref{lem:regret bound known}. We start by stating the generalized Pythagorean inequality for Bregman divergences. Our statement is equivalent to Lemma~11.3 in the book by \citet{Cesa06:Prediction} but, unlike theirs, our proof does not depend on $\regu$ being Legendre. Note that a close inspection of the proof by Cesa-Bianchi and Lugosi already reveals that the gradient condition mentioned by the reviewer is only tangentially used to prove the existence of the Bregman projection $\xbar^t$.

\begin{lemma}[Generalized Pythagorean inequality]\label{lem:pyth}
    Let $\vec{u} \in \seqf \cap \Rpp^{|\Sigma|}$ and let $\vec{c} \in \mathbb{R}^n_{>0}$. Finally, let $\vec{p} \in \seqf \cap \Rpp^{|\Sigma|}$ be the Bregman projection of $\vec{c}$ onto $\seqf \cap \Rpp^{|\Sigma|}$:
        \begin{equation}\label{eq:argmin}
            \vec{p} = \argmin_{\vec{p} \in \seqf \cap \Rpp^{|\Sigma|}} \div{\vec{p}}{\vec{c}}
        \end{equation}
    Then,
    \[
        \div{\vec{u}}{\vec{c}} \ge \div{\vec{u}}{\vec{p}} + \div{\vec{p}}{\vec{c}}.
    \]
\end{lemma}
\begin{proof}
    The necessary first-order optimality condition (e.g., Proposition 2.1.1 in the book by \citet{Borwein10:Convex}) for the optimization problem \eqref{eq:argmin} is
    \[
        (\nabla_{\vec{p}} \div{\vec{p}}{\vec{c}})^\top (\vec{x} - \vec{p}) \ge 0 \quad\forall \vec{x} \in \seqf \cap \Rpp^{|\Sigma|}.
    \]
    (The above condition is necessary under significantly weaker conditions than $\regu$ being Legendre. It is enough for the objective of the minimization to be G\^ateaux differentiable at $\vec{p}$ and the domain $\seqf \cap \Rpp^{|\Sigma|}$ to be convex---both of which are verified in this case.)
    Expanding the gradient of the Bregman divergence and setting in particular $\vec{x} = \vec{u}$, it must be
    \begin{equation}\label{eq:condition}
        (\nabla\regu(\vec{p}) - \nabla\regu(\vec{c}))^\top (\vec{u} - \vec{p}) \ge 0.
    \end{equation}
    Now,
    \begin{align}
        (\nabla\regu(\vec{p}) - \nabla\regu(\vec{c}))^\top (\vec{u} - \vec{p}) &= \nabla\regu(\vec{p})^\top (\vec{u} - \vec{p}) - \nabla\regu(\vec{c})^\top (\vec{u} - \vec{p})\nonumber\\
         &= \nabla\regu(\vec{p})^\top (\vec{u} - \vec{p}) - \nabla\regu(\vec{c})^\top (\vec{u} - \vec{c}) + \nabla\regu(\vec{c})^\top (\vec{p} - \vec{c})\nonumber\\
         &= -\div{\vec{u}}{\vec{p}} + \div{\vec{u}}{\vec{c}} - \div{\vec{p}}{\vec{c}} \label{eq:xxx}
    \end{align}
    Plugging \eqref{eq:xxx} into \eqref{eq:condition} yields the statement.
\end{proof}

\begin{proof}[Proof of \cref{lem:regret bound known}]
    Since $\tilde{\vec{x}}^{t+1}$ exists (\cref{lem:prox solution}) and $\bbR_{>0}^n$ is an open set, $\tilde{\vec{x}}^{t+1}$ must be a critical point of the objective $\eta(\tilde{\vec{\ell}}^t)^\top \hat{\vec{x}} + \div{\hat{\vec{x}}}{\bar{\vec{x}}^t}$, that is
    \begin{align*}
        \vec{0} &= \nabla_{\hat{\vec{x}}}\mleft[\eta(\tilde{\vec{\ell}}^t)^\top \hat{\vec{x}} + \div{\hat{\vec{x}}}{\bar{\vec{x}}^t}\mright](\tilde{\vec{x}}^{t+1})\\
        &= \eta\tilde{\vec{\ell}}^t + \nabla\regu(\tilde{\vec{x}}^{t+1}) - \nabla\regu(\bar{\vec{x}}^{t}),
    \end{align*}
    from which we conclude that
    \begin{equation}\label{eq:etaloss}
        \eta\tilde{\vec{\ell}}^t = \nabla\regu(\bar{\vec{x}}^{t}) - \nabla\regu(\tilde{\vec{x}}^{t+1}).
    \end{equation}
    Hence, for all $\vec{u} \in \seqf \cap \Rpp^{|\Sigma|}$,
    \begin{align}
        \eta(\tilde{\ell}^t)^\top(\bar{\vec{x}}^t - \vec{u}) &= \Big(\nabla\regu(\bar{\vec{x}}^{t}) - \nabla\regu(\tilde{\vec{x}}^{t+1})\Big)^\top(\bar{\vec{x}}^t - \vec{u})\nonumber\\
        &=-\nabla\regu(\bar{\vec{x}}^{t})^\top(\vec{u} - \bar{\vec{x}}^t)  - \nabla\regu(\tilde{\vec{x}}^{t+1})^\top(\bar{\vec{x}}^t - \vec{u})\nonumber\\
        &=-\nabla\regu(\bar{\vec{x}}^{t})^\top(\vec{u} - \bar{\vec{x}}^t) + \nabla\regu(\tilde{\vec{x}}^{t+1})^\top(\vec{u} - \tilde{\vec{x}}^{t+1})  - \nabla\regu(\tilde{\vec{x}}^{t+1})^\top(\bar{\vec{x}}^t - \tilde{\vec{x}}^{t+1})\nonumber\\
        &= \underbrace{\div{\vec{u}}{\bar{\vec{x}}^t}}_{\circled{A}} - \div{\vec{u}}{\tilde{\vec{x}}^{t+1}} + \underbrace{\div{\bar{\vec{x}}^t}{\tilde{\vec{x}}^{t+1}}}_{\circled{B}} .\label{eq:last}
    \end{align}
    We now bound terms \circled{A} and \circled{B} on the right-hand side.
    \begin{itemize}
        \item[\circled{A}] Since $\bar{\vec{x}}^{t}$ is the Bregman projection of $\tilde{\vec{x}}^t$ onto $\seqf \cap \Rpp^{|\Sigma|}$ and $\vec{u}\in\seqf \cap \Rpp^{|\Sigma|}$, using \cref{lem:pyth} with $\vec{c} = \tilde{\vec{x}}^t$ we can write
            \begin{equation}\label{eq:partA}
                \div{\vec{u}}{\tilde{\vec{x}}^t} \ge \div{\vec{u}}{\bar{\vec{x}}^t} + \div{\bar{\vec{x}}^t}{\tilde{\vec{x}}^t} \implies \div{\vec{u}}{\bar{\vec{x}}^t} \le \div{\vec{u}}{\tilde{\vec{x}}^t},
            \end{equation}
            where we used the fact that Bregman divergences are always non-negative, as $\regu$ is convex.
        \item[\circled{B}] Using the fact that Bregman divergences are always non-negative, we have
            \begin{align}
                \div{\bar{\vec{x}}^t}{\tilde{\vec{x}}^{t+1}} &\le \div{\bar{\vec{x}}^t}{\tilde{\vec{x}}^{t+1}} + \div{\tilde{\vec{x}}^{t+1}}{\bar{\vec{x}}^{t}}\nonumber\\
                &= -\nabla\regu(\tilde{\vec{x}}^{t+1})^\top(\bar{\vec{x}}^t - \tilde{\vec{x}}^{t+1}) - \nabla\regu(\bar{\vec{x}}^{t})^\top(\tilde{\vec{x}}^{t+1}- \bar{\vec{x}}^{t})\nonumber\\
                &= \Big(\nabla\regu(\bar{\vec{x}}^{t}) - \nabla\regu(\tilde{\vec{x}}^{t+1})\Big)^\top(\bar{\vec{x}}^t - \tilde{\vec{x}}^{t+1}).\nonumber
            \end{align}
            So, substituting \eqref{eq:etaloss},
            \begin{align}
                \div{\bar{\vec{x}}^t}{\tilde{\vec{x}}^{t+1}} &\le \eta(\tilde{\vec{\ell}}^t)^\top (\bar{\vec{x}}^t - \tilde{\vec{x}}^{t+1}).\label{eq:partB}
            \end{align}
    \end{itemize}
    Substituting \eqref{eq:partA} and \eqref{eq:partB} into \eqref{eq:last}, we obtain
    \begin{align*}
        \eta(\tilde{\vec{\ell}}^t)^\top(\bar{\vec{x}}^t - \vec{u}) &\le \div{\vec{u}}{\tilde{\vec{x}}^t} - \div{\vec{u}}{\tilde{\vec{x}}^{t+1}} + \eta(\tilde{\vec{\ell}}^t)^\top (\bar{\vec{x}}^t - \tilde{\vec{x}}^{t+1}).
    \end{align*}
    Finally, summing over $t=1,\dots, T$ we have
    \begin{align*}\allowdisplaybreaks
        \eta\sum_{t=1}^T \eta(\tilde{\vec{\ell}}^t)^\top(\bar{\vec{x}}^t - \tilde{\vec{x}}^{t+1}) &\le \sum_{t=1}^T \Big(\div{\vec{u}}{\tilde{\vec{x}}^t} - \div{\vec{u}}{\tilde{\vec{x}}^{t+1}}\Big) + \eta\sum_{t=1}^T (\tilde{\vec{\ell}}^t)^\top(\bar{\vec{x}}^t - \tilde{\vec{x}}^{t+1})\\
        &= \div{\vec{u}}{\tilde{\vec{x}}^1} - \div{\vec{u}}{\tilde{\vec{x}}^{T+1}} + \eta\sum_{t=1}^T (\tilde{\vec{\ell}}^t)^\top(\bar{\vec{x}}^t - \tilde{\vec{x}}^{t+1})\\
        &\le \div{\vec{u}}{\tilde{\vec{x}}^1} + \eta\sum_{t=1}^T (\tilde{\vec{\ell}}^t)^\top(\bar{\vec{x}}^t - \tilde{\vec{x}}^{t+1}).
     \end{align*}
     Finally, since $\bbR^n_{>0}$ is an open set, and $\tilde{\vec{x}}^1=\argmin_{\hat{\vec{x}} \in \Rpp^{|\Sigma|}} \regu(\hat{\vec{x}})$ (\cref{eq:x tilde init}), it must be that
     \[
        \nabla \regu(\tilde{\vec{x}}^1) = \vec{0} \implies  \div{\vec{u}}{\tilde{\vec{x}}^1} = \regu(\vec{u}) - \regu(\tilde{\vec{x}}^1) \le \regu(\vec{u}),
     \]
     where the last inequality holds since the range of $\regu$ is non-negative.
     This concludes the proof.
\end{proof}

    \subsection{Analysis of Mirror Descent using the Local Norms of the Entropy DGF}

        \begin{restatable}{proposition}{propproxsolbounds}\label{prop:prox solution bounds}
           Let the quantity $\psi^{t}_{ja}$ be defined for all sequences $ja \in \Sigma$ as
           \[
            \psi^{t}_{ja} \defeq \vec{u}_{ja}\trans (\tilde{\vec{\ell}}^{t} \circ \xbar^{t}) = \sum_{j'a' \ \succeq\ ja} \tilde{\ell}^{t}_{j'a'}\cdot{x}^{t}_{j'a'}.
          \]
          If $\tilde{\vec{\ell}}^t \in \Rp^{|\Sigma|}$, the intermediate iterate $\tilde{\vec{x}}^{t+1}$ satisfies
          \[
            \frac{{x}^{t}_{ja}}{{x}^{t}_{p_j}} \,\exp\left\{-\frac{\eta}{w_j {x}^{t}_{ja}}\psi^{t}_{ja}\right\} \le \frac{\tilde{x}^{t+1}_{ja}}{\tilde{x}^{t+1}_{p_j}} \le \frac{{x}^{t}_{ja}}{{x}^{t}_{p_j}}.
          \]
        \end{restatable}
    \begin{proof}
      For ease of notation, in this proof we will make use of the symbol $\next{ja}$ to mean $\next{\rho(j,a)}$. We prove the proposition by induction:
      \begin{itemize}
          \item \textbf{Base case}. For any $ja$ with $\mathcal{C}_{ja} = \emptyset$ (and thus $\psi^{t}_{ja} = \tilde{\ell}^{t}_{ja} {x}^{t}_{ja}$) we have by \cref{lem:prox solution}
          \[
            \frac{\tilde{x}^{t+1}_{ja}}{\tilde{x}^{t+1}_{p_j}} =
                \frac{{x}^{t}_{ja}}{{x}^{t}_{p_j}}\exp\mleft\{-\eta \frac{\tilde{\ell}^{t}_{ja}}{w_j}\mright\} =
            \frac{{x}^{t}_{ja}}{{x}^{t}_{p_j}} \,\exp\mleft\{-\frac{\eta}{w_j {x}^{t}_{ja}}\psi^{t}_{ja}\mright\},
          \]
          which proves the lower bound. In order to prove the upper bound, it is enough to note that the argument of the $\exp$ is non-positive. Hence,
\[
            \frac{\tilde{x}^{t+1}_{ja}}{\tilde{x}^{t+1}_{p_j}} = \frac{{x}^{t}_{ja}}{{x}^{t}_{p_j}} \,\exp\mleft\{-\frac{\eta}{w_j {x}^{t}_{ja}}\psi^{t}_{ja}\mright\} \le  \frac{{x}^{t}_{ja}}{{x}^{t}_{p_j}}.
\]
          \item \textbf{Inductive step}. Suppose that the inductive hypothesis holds for all sequences $j'a' \succ ja$. Then, we have
          \begin{align}\label{eq:inductive hyp}
            \xi^{t+1}_{ja} &= \sum_{j'\in\next{ja}}\mleft( w_{j'} \sum_{a'\in A_{j'}} \frac{\tilde{x}^{t+1}_{j'a'}}{\tilde{x}^{t+1}_{ja}}\mright) \ge \sum_{j'\in\next{ja}}\mleft( w_{j'} \sum_{a'\in A_{j'}} \frac{{x}^{t}_{j'a'}}{{x}^{t}_{ja}} \exp\mleft\{-\frac{\eta}{w_{j'} {x}^{t}_{j'a'}}\psi^{t}_{j'a'}\mright\}\mright).
          \end{align}
          Furthermore, for all $ja\in\Sigma$, using~\cref{def:dilated entropy} we have
          \begin{align*}
            w_{\rho(j,a)} = \sum_{j' \in \next{ja}} w_j' = \sum_{j'\in\next{ja}} \hspace{-2mm} w_{j'} \hspace{-2mm} \sum_{a'\in A_{j'}}\!\! \frac{{x}^{t}_{j'a'}}{{x}^{t}_{ja}},
          \end{align*}
          where the last equality follows from the fact that $\xbar^{t}$ is a valid sequence-form strategy. Hence, we can rewrite~\eqref{eq:prox step} as
        \begin{equation}\label{eq:prox step expanded}
          \frac{\tilde{x}^{t+1}_{ja}}{\tilde{x}^{t+1}_{p_j}} = \frac{{x}^{t}_{ja}}{{x}^{t}_{p_j}}\exp\mleft\{-\eta \frac{\tilde{\ell}^{t}_{ja}}{w_j}-\frac{1}{w_j}\sum_{j'\in\next{ja}} \hspace{-2mm} w_{j'} \hspace{-2mm} \sum_{a'\in A_{j'}}\!\! \frac{{x}^{t}_{j'a'}}{{x}^{t}_{ja}} + \frac{\xi^{t+1}_{ja}}{w_j}\mright\}.
        \end{equation}

Plugging in the inductive hypothesis~\eqref{eq:inductive hyp} into~\eqref{eq:prox step expanded} and using the monotonicity of $\exp$, we obtain
          \begin{align}
              \frac{\tilde{x}^{t+1}_{ja}}{\tilde{x}^{t+1}_{p_j}} &\ge \frac{{x}^{t}_{ja}}{{x}^{t}_{p_j}} \exp\mleft\{-\eta\frac{\tilde{\ell}^{t}_{ja}}{w_j} - \frac{1}{w_j}\mleft(\sum_{j'\in\next{ja}}w_{j'}\sum_{a'\in A_{j'}} \frac{{x}^{t}_{j'a'}}{{x}^{t}_{ja}}\mleft(1 - \exp\mleft\{-\frac{\eta}{w_{j'} {x}^{t}_{j'a'}}\psi^{t}_{j'a'}\mright\}\mright)\mright)\mright\} \nonumber\\
              &\ge \frac{{x}^{t}_{ja}}{{x}^{t}_{p_j}} \exp\mleft\{-\eta\frac{\tilde{\ell}^{t}_{ja}}{w_j} - \frac{\eta}{w_j}\mleft(\sum_{j'\in\next{ja}}\sum_{a'\in A_{j'}} \frac{1}{{x}^{t}_{ja}}\psi^{t}_{j'a'}\mright)\mright\},\label{eq:almost done}
          \end{align}
          where the second inequality follows from the fact that $1-e^{-x} \le x$ for all $x \in \bbR$. Finally, using the definition of $\psi^{t}_{ja}$ we find
\begin{equation}\label{eq:psi step}
        \sum_{j' \in \next{ja}}\sum_{a' \in A_{j'}} \psi^{t}_{j'a'} = \psi^{t}_{ja} - \tilde{\ell}^{t}_{ja} {x}^{t}_{ja}.
\end{equation}
        Plugging~\eqref{eq:psi step} into~\eqref{eq:almost done} we obtain
        \begin{align*}
              \frac{\tilde{x}^{t+1}_{ja}}{\tilde{x}^{t+1}_{p_j}}
              &\ge \frac{{x}^{t}_{ja}}{{x}^{t}_{p_j}} \exp\mleft\{-\frac{\eta}{w_j {x}^{t}_{ja}}\psi^{t}_{ja}\mright\}.
        \end{align*}
        This completes the proof for the lower bound.

        In order to prove the upper bound, we start from~\eqref{eq:prox step}.
          \begin{align*}
            \frac{\tilde{x}^{t+1}_{ja}}{\tilde{x}^{t+1}_{p_j}} = \frac{{x}^{t}_{ja}}{{x}^{t}_{p_j}}\exp\mleft\{-\eta \frac{\tilde{\ell}^{t}_{ja}}{w_j}-\frac{1}{w_j}\mleft(\sum_{j'\in\next{ja}}w_{j'}\sum_{a'\in A_{j'}} \mleft(\frac{{x}^{t}_{j'a'}}{{x}^{t}_{ja}} - \frac{\tilde{x}^{t+1}_{j'a'}}{\tilde{x}^{t+1}_{ja}}\mright)\mright) \mright\} \le  \frac{{x}^{t}_{ja}}{{x}^{t}_{p_j}}.
          \end{align*}
        Using the inductive hypothesis ${\tilde{x}^{t+1}_{ja}}/{\tilde{x}^{t+1}_{p_j}} \le  {{x}^{t}_{ja}}/{{x}^{t}_{p_j}}$, we obtain
          \begin{align*}
            \frac{\tilde{x}^{t+1}_{ja}}{\tilde{x}^{t+1}_{p_j}} &\le \frac{{x}^{t}_{ja}}{{x}^{t}_{p_j}}\exp\mleft\{-\eta \frac{\tilde{\ell}^{t}_{ja}}{w_j}-\frac{1}{w_j}\mleft(\sum_{j'\in\next{ja}}w_{j'}\sum_{a'\in A_{j'}} \mleft( \frac{{x}^{t}_{j'a'}}{{x}^{t}_{ja}} - \frac{{x}^{t}_{j'a'}}{{x}^{t}_{ja}}\mright)\mright) \mright\} \le  \frac{{x}^{t}_{ja}}{{x}^{t}_{p_j}}\\
                &\le \frac{{x}^{t}_{ja}}{{x}^{t}_{p_j}}\exp\mleft\{-\eta \frac{\tilde{\ell}^{t}_{ja}}{w_j}\mright\} \le \frac{{x}^{t}_{ja}}{{x}^{t}_{p_j}}.\qedhere
          \end{align*}
      \end{itemize}
    \end{proof}

    An immediate corollary of~\cref{prop:prox solution bounds} is the following.
    \begin{restatable}{corollary}{corcoarsebounds}\label{cor:coarse bounds}
        For all $ja \in \Sigma$,
        \begin{equation*}
        0 < \exp \mleft\{ - \sum_{j' a' \preceq\, ja} \frac{\eta}{w_j {x}^{t}_{j' a'}} \psi^{t}_{j' a'}  \mright\} \le \frac{\tilde{x}^{t+1}_{ja}}{{x}^{t}_{ja}} \le \frac{\tilde{x}^{t+1}_{p_j}}{{x}^{t}_{p_j}} \le 1.
        \end{equation*}
        In particular,
        \begin{equation}\label{eq:bound on ratio}
            0 \le 1 - \frac{\tilde{x}^{t+1}_{ja}}{{x}^{t}_{ja}} \le \sum_{j' a' \preceq\, ja} \frac{\eta}{w_j {x}^{t}_{j' a'}} \psi^{t}_{j' a'}.
        \end{equation}
    \end{restatable}
    \begin{proof}
        The first statement follows from applying \cref{prop:prox solution bounds} repeatedly on the path from the root of the decision tree to decision point $j$. The second statement holds from the first statement by noting that
    \[
        1 - \frac{\tilde{x}^{t+1}_{ja}}{{x}^{t}_{ja}} \le 1 - \exp \mleft\{ - \sum_{j' a' \preceq\, ja} \frac{\eta}{w_j {x}^{t}_{j' a'}} \psi^{t}_{j' a'}  \mright\} \le  \sum_{j' a' \preceq\, ja} \frac{\eta}{w_j {x}^{t}_{j' a'}} \psi^{t}_{j' a'},
    \]
    where we used the fact that $1 - e^{-x} \le x$ for all $x\in \bbR$.
    \end{proof}

    Furthermore, the following lemma will be important in the proof of \cref{prop:omd unproj step}.

    \begin{restatable}{lemma}{lemuppropagation}\label{lem:up propagation}
        For all sequences $ja$,
        \begin{align*}
            \sum_{j'a'\ \succeq\ ja} \frac{w_{j'} y_{j' a'}}{w_j y_{ja}} \le  2.
        \end{align*}
    \end{restatable}
    \begin{proof}
        By induction.

        \begin{itemize}
            \item \textbf{Base case}. For any terminal decision $ja \in \Sigma$ (that is, $\mathcal{C}_{\rho(j,a)} = \emptyset$), we have
            \[
             \sum_{j'a' \ \succeq\ ja} \frac{w_{j'} y_{j'a'}}{w_j y_{ja}}= \frac{w_j {x}_{ja}}{w_j {x}_{ja}} = 1 \le 2.
            \]
            \item \textbf{Inductive step}. Suppose that the inductive hypothesis holds for all sequences $j' a' \succ ja$. Then,
            \begin{align*}
                 \sum_{j'a' \ \succeq\ ja} \frac{w_{j'} y_{j' a'}}{w_j y_{ja}}
                 &= 1 + \sum_{j' \in \next{\rho(j,a)}}\sum_{a' \in A_{j'}} \mleft(\frac{w_{j'} y_{j'a'}}{w_{j} y_{ja}} \sum_{j''a'' \ \succeq\ j'a'} \frac{w_{j''} y_{j''a''}}{w_{j'} y_{j'a'}}\mright)\\
                 &\le 1 + 2 \sum_{j' \in \next{\rho(j,a)}}\sum_{a' \in A_{j'}} \frac{w_{j'} y_{j' a'}}{w_j y_{ja}}\\
                 &= 1 + 2\sum_{j' \in \next{\rho(j,a)}} \frac{w_{j'}}{w_j}
                 = 1 + \frac{2w_{\rho(j,a)}}{w_j}
                 \le 2,
            \end{align*}
            where the first inequality follows by the inductive hypothesis, and the second inequality holds by definition of the weights in the dilated DGF (Equation~\ref{def:dilated entropy}). \qedhere
        \end{itemize}
    \end{proof}

    \cref{prop:prox solution bounds} is also a fundamental step for the following proposition, which bounds the length of the step (as measured according to the local norm $\pn{\cdot}{\xbar^{t}}$) between the last decision $\xbar^{t}$ and the next intermediate iterate $\tilde{\vec{x}}^{t+1}$ as a function of the stepsize parameter $\eta$ and the dual local norm of the loss $\tilde{\vec{\ell}}^{t}$ that was last observed:

\begin{restatable}{proposition}{propunprojstep}\label{prop:omd unproj step}
        Let $D$ be the maximum depth of any node in the decision process. If $\tilde{\vec{\ell}}^t \in \Rp^{|\Sigma|}$, then
        \begin{equation}
            \pn{\xbar^{t} - \tilde{\vec{x}}^{t+1}}{\xbar^{t}} \le \eta\sqrt{3D} \cdot\dn{\tilde{\vec{\ell}}^{t}}{\xbar^{t}}.
        \end{equation}
\end{restatable}
    \begin{proof}
     By \cref{cor:coarse bounds}, $\xbar^{t} - \tilde{\vec{x}}^{t+1} \in \Rp^{|\Sigma|}$. Hence, we can apply \cref{lem:bound primal}:
        \begin{align*}
            \pn{\xbar^{t} - \tilde{\vec{x}}^{t+1}}{\xbar^{t}}^2 &\le  \frac{3}{2}\sum_{j\in\cJ}\sum_{a\in A_j} \frac{w_j}{{x}^{t}_{ja}}({x}^{t}_{ja} - \tilde{x}^{t+1}_{ja})^2
            =  \frac{3}{2}\sum_{j\in\cJ}\sum_{a\in A_j} w_j {x}^{t}_{ja} \mleft(1-\frac{\tilde{x}^{t+1}_{ja}}{{x}^{t}_{ja}}\mright)^2.
        \end{align*}
    Using Inequality~\eqref{eq:bound on ratio},
        \begin{align*}
            \pn{\xbar^{t} - \tilde{\vec{x}}^{t+1}}{\xbar^{t}}^2 &\le  \frac{3\eta^2}{2} \sum_{j\in\cJ}\sum_{a\in A_j} w_j {x}^{t}_{ja}  \mleft( \sum_{\atOrAbove{j'a'}{ja}} \frac{\psi^{t}_{j'a'}}{w_{j'} {x}^{t}_{j' a'}} \mright)^2\\
            &\le  \frac{3D\eta^2}{2} \sum_{j\in\cJ}\sum_{a\in A_j} \sum_{\atOrAbove{j'a'}{ja}} \frac{w_j {x}^{t}_{ja}  }{w_{j'}^2 ({x}^{t}_{j'a'})^2} (\psi^{t}_{j'a'})^2\\
            &= \frac{3D \eta^2}{2} \sum_{j\in\cJ}\sum_{a\in A_j} \sum_{\atOrAbove{j'a'}{ja}} \frac{w_j {x}^{t}_{ja}  }{w_{j'} {x}^{t}_{j' a'}}  \frac{(\psi^{t}_{j'a'})^2}{w_{j'} {x}^{t}_{j' a'}},
        \end{align*}
        where the second inequality follows from applying Cauchy-Schwarz. Now using double counting we derive
                \begin{align*}
            \frac{3D \eta^2}{2} \sum_{j\in\cJ}\sum_{a\in A_j} \sum_{\atOrAbove{j'a'}{ja}} \frac{w_j {x}^{t}_{ja}  }{w_{j'} {x}^{t}_{j' a'}}  \frac{(\psi^{t}_{j'a'})^2}{w_{j'} {x}^{t}_{j' a'}} &= \frac{3D \eta^2}{2} \sum_{j\in\cJ}\sum_{a\in A_j}\mleft( \frac{(\psi^{t}_{ja})^2}{w_{j} {x}^{t}_{j a}}  \sum_{\atOrBelow{j'a'}{ja}} \frac{w_{j'} {x}^{t}_{j' a'}}{w_{j} {x}^{t}_{j a}}\mright)\\
              &\le 3D \eta^2 \sum_{j\in\cJ}\sum_{a\in A_j} \frac{(\psi^{t}_{ja})^2}{w_{j} {x}^{t}_{j a}}.
        \end{align*}
        where the second inequality follows from \cref{lem:up propagation}.
        Finally, plugging in definition of $\psi^{t}_{ja}$ and using \cref{cor:dual norm}, we have
        \[
            \pn{\xbar^{t} - \tilde{\vec{x}}^{t+1}}{\xbar^{t}}^2 \le 3D \eta^2 \sum_{j\in\cJ}\sum_{a\in A_j} \frac{(\vec{u}_{ja}\trans(\tilde{\vec{\ell}}^{t} \circ \xbar^{t}))^2}{w_{j} {x}^{t}_{j a}} = 3D \eta^2 \dn{\tilde{\vec{\ell}}^{t}}{\xbar^{t}}^2.
        \]
        Taking the square root of both sides yields the statement.
    \end{proof}

    Finally, \cref{lem:regret bound known} can be used to derive a regret bound for $\tilde{\mathcal{R}}$ expressed in term of local norms. In particular, using the generalized Cauchy-Schwarz inequality together with \cref{prop:omd unproj step}, we obtain
    \begin{align*}
      (\tilde{\vec{\ell}}^{t})\trans(\xbar^{t} - \tilde{\vec{x}}^{t+1}) &\le \dn{\tilde{\vec{\ell}}^{t}}{\xbar^{t}}\cdot\pn{\xbar^{t} - \tilde{\vec{x}}^{t+1}}{\xbar^{t}} \\
        &\le \eta\sqrt{3D}\cdot\dn{\tilde{\vec{\ell}}^{t}}{\xbar^{t}}^2.
    \end{align*}

    Substituting the last inequality into the bound of~\cref{lem:regret bound known}, we obtain \cref{thm:omd regret bound}.

    \newpage
    \section{Sampling Scheme and Its Autocorrelation Matrix}\label{app:R}

\subsection{General Observations}
\begin{restatable}{lemma}{lemimagespan}\label{lem:image equals span}
Let $\pi$ be a distribution with finite support, and let $\vec{y} \sim \pi$. Then $\img\bbE\big[\vec{y} \vec{y}\trans\big] = \Span\supp\pi$.
\end{restatable}
\begin{proof}
We prove the statement by showing that the nullspace of $\bbE[\vec{y}\vec{y}\trans]$ is equal to the orthogonal complement of $\Span\supp\pi$, in symbols:
\[
    \ker \bbE[\vec{y}\vec{y}\trans] = (\Span\supp\pi)^\bot.
\]
This will immediately imply the statement using the well-known relationship $\img \bbE[\vec{y}\vec{y}\trans] = (\ker \bbE[\vec{y}\vec{y}\trans])^\bot$.

We start by showing $\ker \bbE[\vec{y}\vec{y}\trans] \subseteq (\Span\supp\pi)^\bot$. Take $\vec{z} \in \ker \bbE[\vec{y}\vec{y}\trans]$. Then,
\begin{align*}
  \bbE[\vec{y}\vec{y}\trans] \vec{z} = 0
  &\implies \vec{z}\trans \bbE[\vec{y}\vec{y}\trans] \vec{z} = 0 \\
  &\implies \bbE[(\vec{z}\trans \vec{y})^2] = 0 \\
  &\implies \vec{z}\trans \vec{y} = 0 \quad \forall \vec{y} \in \supp \pi \\
  &\implies \vec{z}\trans \vec{y} = 0 \quad \forall \vec{y} \in \Span \supp \pi.
\end{align*}

We now look at the other direction, that is $(\Span\supp\pi)^\bot \subseteq \ker \bbE[\vec{y}\vec{y}\trans]$. Take $\vec{z} \in (\Span\supp\pi)^\bot$. Then,
\begin{align*}
  \bbE[\vec{y}\vec{y}\trans] \vec{z} = \bbE[\vec{y}(\vec{y}\trans \vec{z})] = \bbE[\vec{y} \cdot 0] = \vec{0}.
\end{align*}
This implies $\vec{z} \in \ker \bbE[\vec{y}\vec{y}\trans]$, and concludes the proof.
\end{proof}

\begin{lemma}\label{lem:C Cinv y equal y}
    Suppose that a distribution $\pi$ over $\pure$ is known, such that the support of $\pi$ is full-rank (that is, $\Span\supp\pi = \Span \pure$), and let $\vec{y} \sim \pi$. Furthermore, let $\mat{C}^{-}$ be any generalized inverse of the autocorrelation matrix $\mat{C}$. Then, for all $\vec{z}\in \Span \pure$,
\begin{align*}
  \mat{C} \mat{C}^{-} \vec{z} = \vec{z}, \text{ and }
  \vec{z}\trans \mat{C}^{-}\, \mat{C} = \vec{z}\trans\!\!.
\end{align*}
\end{lemma}
\begin{proof}
Since $\img \mat{C} = \Span\supp\pi$ (see Lemma \ref{lem:image equals span}) and $\Span\supp\pi = \Span\pure$ by hypothesis, it must be $\vec{z}\in\img \mat{C}$. Hence, there exists $\vec{v} \in \bbR^{|\Sigma|}$ such that $\vec{z} = \mat{C}\,\vec{v}$, and therefore
\begin{align*}
    \mat{C} \mat{C}^{-} \vec{z} = \mat{C} \mat{C}^{-} \mat{C} \vec{v} = \mat{C} \vec{v} = \vec{z},
\end{align*}
    where the second equality follows by definition of generalized inverse. The proof of the second equality in the statement is analogous.
\end{proof}

\begin{lemma}\label{lem:z Cinv ybar equals 1}
  Suppose that a distribution $\pi$ over $\pure$ is known, such that the support of $\pi$ is full-rank (that is, $\Span\supp\pi = \Span \pure$), and let $\vec{y} \sim \pi$ and $\bar{\vec{y}} = \mathbb{E}[\vec{y}]$. Furthermore, let $\mat{C}^{-}$ be any generalized inverse of the autocorrelation matrix $\mat{C}$.  Then, for all $\vec{z} \in \Span \pure$,
  \[
    \vec{z}\trans \mat{C}^{-} \bar{\vec{y}} = \bar{\vec{y}}\trans \mat{C}^{-} \vec{z} = 1.
  \]
\end{lemma}
\begin{proof}
Since $\img \mat{C} = \Span\supp\pi$ (see Lemma \ref{lem:image equals span}) and $\Span\supp\pi = \Span\pure$ by hypothesis, there exists $\vec{v} \in \bbR^{|\Sigma|}$ such that $\vec{z} = \mat{C}\vec{v}$. Furthermore, $\bar{\vec{y}} = \mat{C}\,\vec{\tau}$ where $\vec{\tau}$ is any vector such that $\vec{z}\trans \vec{\tau} = 1$ for all $\vec{z} \in \seqf$ (such vector must exist because $\vec{0}$ is not in the affine hull of $\pure$). Hence,
  \begin{align*}
    \vec{z}\trans \mat{C}^{-} \bar{\vec{y}} &= \vec{v}\trans \mat{C} \mat{C}^{-} \mat{C}\vec{\tau} = \vec{v}\trans \mat{C}\vec{\tau} = \vec{z}\trans\vec{\tau} = 1.
  \end{align*}
  The proof that $\bar{\vec{y}}\trans \mat{C}^{-}\vec{z} = 1$ is analogous.
\end{proof}

Since $\mathbb{E}_t[\vec{y}^t] = \bar{\vec{x}}^t$, Lemma \ref{lem:C Cinv y equal y} and \ref{lem:z Cinv ybar equals 1} imply that $\mat{C}^t \mat{C}^{t-} \vec{z} = \vec{z}$ for all $\vec{z} \in \pure$, and that $\vec{z}\trans \mat{C}^{t-} \bar{\vec{x}}^t = 1$ for all $\vec{z} \in \pure$.

\subsection{Unbiasedness of the Sampling Scheme}
\label{app:unbiased}

In \cref{algo:sampling scheme} we give pseudocode for the sampling scheme described in \cref{sec:sampling scheme}.

    \begin{algorithm}[H]\small
        \caption{\textsc{Sample}($\xbar^t$)}
        \label{algo:sampling scheme}\DontPrintSemicolon
        \KwIn{$\xbar^t \in\seqf$ sequence-form strategy}
        \KwOut{$\vec{y}^t \in \pure$ such that $\bbE[\vec{y}^t] = \xbar^t$}
        \BlankLine
        $\vec{y}^t \gets \vec{0}$\;
        \Subr{\normalfont\textsc{RecursiveSample}($v$)}{
            \uIf{$v \in \cJ$}{
                Sample an action $a \sim ({x}_{va}^t/{x}_{p_v})_{a \in A_v}$\;
                $y_{va}^t \gets 1$\;
                $\textsc{RecursiveSample}(\rho(v,a))$\;
            }
            \ElseIf{$v \in \cK$}{
                \textbf{for} $s \in S_k$ \textbf{do}
                    $\textsc{RecursiveSample}(\rho(v,s))$\hspace*{-1cm}\;
            }
        }
        \Hline{}
        $\textsc{RecursiveSample}(r)$\Comment*{\color{commentcolor}$r$: root of the decision process]}
        \KwRet{$\vec{y}^t$}\;
    \end{algorithm}

\begin{restatable}{lemma}{lemunbiasedsample}\label{lem:sample unbiased}
    The sampling scheme given by \cref{algo:sampling scheme} is unbiased, that is, $\bbE_t[\vec{y}^t] = \xbar^t$.
\end{restatable}

\begin{proof}
We prove by induction over the structure of the sequential decision process that for all $v \in \cJ \cup \cK$,
\begin{align*}
 \bbE_t[\vec{y}_v^t] = \xbar_v^t
\end{align*}

\begin{itemize}
\item \textbf{First case:} $v\in\cJ$ is a terminal decision point. This serves as the base case for the induction.

    Let $A_v = \{a_1, \dots, a_n\}$. Then, $\substrategy{v} = ({x}^{t}_{va_1}, \dots, {x}^{t}_{va_n}) \in \Delta^n$ and
    \begin{align*}
      \bbE_t[\vec{y}_v^t] = \sum_{i = 1}^{n} {x}^{t}_{va_i} \vec{e_i} = \substrategy{v}.
    \end{align*}
\item\textbf{Third case:} $v \in \cJ$ is a (generic) decision point. Let $\{a_1, \dots, a_m\}$ the terminal actions at $v$, and $\{a_{m+1}, \dots, a_n\}$ the remaining, non-terminal actions. Furthermore, let $\next{j} = \{k_{m+1}, \dots, k_n\}$ be the set of observation points that are immediately reachable after $v$. From~\cref{eq:structure decision point}, $\substrategy{j}$ must be in the form $\substrategy{j} = (\lambda^{t}_1, \dots, \lambda^{t}_n, \lambda^{t}_{m+1} \substrategy{k_{m+1}}, \dots, \lambda^{t}_n \substrategy{k_n})$, where $\vec{\lambda}^t = (\lambda^{t}_1, \dots, \lambda^{t}_n) \in \Delta^n$. It follows that
    \begin{align*}
      \bbE_t[\vec{y}_v^t] =
      \sum_{i=1}^m \lambda_i^t \begin{pmatrix}\vec{e}_i\\\vec{0}\\\vdots\\\vec{0}\end{pmatrix} + \sum_{i=m+1}^n \lambda_i^t \begin{pmatrix}\vec{e}_i\\\vdots\\\bbE_t[\vec{y}_{k_i}^t]\\\vdots\end{pmatrix}
      =
      \xbar_v^t.
    \end{align*}
\item \textbf{Third case:} $v \in \cK$ is an observation point. Let $\next{v} = \{j_1, \dots, j_n\}$ be the set of decision points that are immediately reachable after $v$. From~\eqref{eq:structure observation point}, $\substrategy{v}$ is in the form $\substrategy{v} = (\substrategy{j_1}, \dots, \substrategy{j_n}) \in \prod_{i=1}^n\seqf_{j_i}$. It follows that
    \begin{align*}
      \bbE_t[\vec{y}_v^t] =
      \bbE_t\left[\begin{pmatrix}
        \vec{y}_{j_1}^t \\
        \vdots \\
        \vec{y}_{j_n}^t
      \end{pmatrix}\right]
      =
      \begin{pmatrix}
          \bbE_t[\vec{y}_{j_1}^t] \\
          \vdots \\
          \bbE_t[\vec{y}_{j_n}^t]
      \end{pmatrix} =
      \begin{pmatrix}
          \substrategy{j_1} \\
          \vdots \\
          \substrategy{j_n}
      \end{pmatrix}
      = \substrategy{v},
    \end{align*}
    where the second equality follows from the independence of the sampling scheme and the third equality follows from the inductive hypothesis.\qedhere
\end{itemize}
\end{proof}

\subsection{Autocorrelation Matrix of the Sampling Scheme}
\label{app:autocorrelation}

%
%

\subsubsection{Decision Points}

Let $v \in \cJ$ be a decision point, where actions $a_1, \dots, a_m$ are terminal, while $a_{m+1}, \dots, a_n$ are not. In order to sample a pure sequence-form strategy we first sample $a \in \{a_1, \dots, a_n\}$ according to the distribution specified by $\{{x}_{va}\}_{a \in A_v} \in \Delta^{|A_v|}$. Then, we set $y_{va}^t = 1$. Finally, if the sampled action is \emph{not} terminal, we recursively sample $\vec{y}_{\rho(v,a)}^t$ by calling into the sampling scheme for $\rho(v,a)$.

\begin{restatable}{lemma}{lemautconvexhull}\label{lem:aut convex hull}
    Let $\pure$ be a decision process rooted in decision point $j$. Let $\{a_1,\dots,a_m\}$ be the terminal actions at $j$, and let $\{a_{m+1}, \dots, a_n\}$ be the remaining, non-terminal actions. Furthermore, let
\[
    \substrategy{j} = (\lambda^{t}_1, \dots, \lambda^t_m, \lambda^t_{m+1}, \dots, \lambda^{t}_n, \lambda^{t}_{m+1} \substrategy{k_{m+1}}, \dots, \lambda^{t}_n \substrategy{k_n}),
\]
where $(\lambda^{t}_1, \dots, \lambda^{t}_n) \in \Delta^n$, in accordance with~\eqref{eq:structure decision point}. Let $\mat{C}_{k_i}^t$ for $i \in 1,\dots, n$ be the autocorrelation matrix of the unbiased sampling scheme picking $\vec{y}_{k_i}^t \in \pure_{k_i}$ using $\xbar_{k_i} / \lambda_i$ for all $i\in\{m+1, \dots, n\}$.  The autocorrelation matrix of the sampling scheme picking $\vec{y}_{j}^t$ is
    \begin{align*}
      \mat{C}^{t}_{j} =
      \mleft(\!\begin{array}{c@{\hspace{1mm}}c@{\hspace{1mm}}c@{\hspace{.5mm}}|c@{\hspace{1mm}}c@{\hspace{1mm}}c@{\hspace{.5mm}}|@{\hspace{1mm}}c@{\hspace{1mm}}c@{\hspace{1mm}}c}
        \lambda^t_1&& &&& &&\\
        &\ddots& &&& &&\\
        &&\lambda^t_m &&& &&\\[1mm]
    \hline&&&&&\\[-3mm]
        &&&\lambda^{t}_{m+1} &&& \lambda^{t}_{m+1} (\xbar^{t}_{k_{m+1}})\trans && \\[-1mm]
        &&&&\ddots    && &\ddots& \\[-1mm]
        &&&&& \lambda^{t}_n& && \lambda^{t}_n (\xbar^{t}_{k_n})\trans\\[1mm]
    \hline&&&&&\\[-3mm]
        &&&\lambda^{t}_{m+1}\xbar^{t}_{k_{m+1}}\!\! &&& \lambda^{t}_{m+1} \mat{C}^{t}_{k_{m+1}}\!\! &&\\[-1mm]
        &&&&\ddots&&&\ddots&\\
        &&&&&\lambda^{t}_n\xbar^{t}_{k_n}&&&\!\!\!\lambda^{t}_n \mat{C}^{t}_{k_n}
      \end{array}\!\mright)\!.
    \end{align*}
\end{restatable}
\begin{proof}
It follows from the definition of the sampling scheme that
\begin{align*}
  \mat{C}^t_{j} &= \bbE_{t}[\vec{y}_j^t (\vec{y}_j^t)\trans]\\
   &= \sum_{i=1}^m \lambda^t_i
  \bbE_{t}\left[(\vec{e}_i\trans, \vec{0}, \dots, \vec{0})\trans (\vec{e}_i\trans, \vec{0}, \dots, \vec{0}) \right] \\[-2mm]
    &\hspace{3.5cm} + \sum_{i = m+1}^{n} \lambda^t_i
  \bbE_{t}\left[(\vec{e}_i\trans, \vec{0}, \dots, \vec{0}, (\vec{y}_{k_i}^t)\trans,  \vec{0}, \dots, \vec{0})\trans (\vec{e}_i\trans, \vec{0}, \dots, \vec{0}, (\vec{y}_{k_i}^t)\trans,  \vec{0}, \dots, \vec{0}) \right].
\end{align*}
    Expanding the outer products and summing yields the statement.
\end{proof}

\subsubsection{Observation Points}

Let $v \in \cK$ be an observation point. In order to sample $\vec{y}^t_v$ given a $\xbar^t_v = (\xbar^t_{j_1}, \dots, \xbar^t_{j_n})$, we call into the sampling schemes for nodes $j_1, \dots, j_n$ by making $n$ independent calls to $\textsc{Sample}(j_i, \xbar^t_{j_i})$ for $i = 1,\dots, n$.

\begin{restatable}{lemma}{lemautcartesian}\label{lem:aut cartesian}
    Let $\pure$ be a decision process rooted in observation point $k$. Let $\next{k} = \{j_1, \dots, j_n\}$ be the set of decision points that are immediately reachable after $k$. In accordance with~\eqref{eq:structure observation point}, $\substrategy{k}$ is in the form $\substrategy{k} = (\substrategy{j_1}, \dots, \substrategy{j_n})$. Let $\mat{C}_{j_i}^t$ for $i \in 1,\dots, n$ be the autocorrelation matrix of the unbiased sampling scheme picking $\vec{y}_{j_i}^t \in \pure_{j_i}$ using $\xbar_{j_i}$.  The autocorrelation matrix of the sampling scheme picking $\vec{y}_{k}^t \in \pure$ is
    \begin{align*}
        \mat{C}^{t}_{k} &=
        \!\begin{pmatrix}
           \mat{C}^{t}_{j_1} & \xbar^{t}_{j_1} (\xbar^{t}_{j_2})\trans & \cdots & \xbar^{t}_{j_1}(\xbar^{t}_{j_n})\trans \\[1mm]
            \xbar^{t}_{j_2} (\xbar^{t}_{j_1})\trans & \mat{C}^{t}_{j_2} & \cdots & \xbar^{t}_{j_2}(\xbar^{t}_{j_n})\trans\\
            \vdots & \vdots & \ddots & \vdots \\[1mm]
            \xbar^{t}_{j_n} (\xbar^{t}_{j_1})\trans & \xbar^{t}_{j_n}(\xbar^{t}_{j_2})\trans & \cdots & \mat{C}^{t}_{j_n}
           \end{pmatrix}.
    \end{align*}
\end{restatable}
\begin{proof}
It follows from the definition of the sampling scheme that
\begin{align*}
   \mat{C}_{t,k} &= \bbE_{t}[\vec{y}_k^t (\vec{y}_k^t)\trans]\\
   &= \bbE_{t}\left[
   \begin{pmatrix}
      \vec{y}_{j_1} \\
      \vec{y}_{j_2} \\
      \vdots \\
      \vec{y}_{j_n}
    \end{pmatrix}
    \begin{pmatrix}
      \vec{y}_{j_1} \\
      \vec{y}_{j_2} \\
      \vdots \\
      \vec{y}_{j_n}
    \end{pmatrix}\trans
    \right]
    = \!\begin{pmatrix}
   \bbE_{t}[\vec{y}_{j_1}^t(\vec{y}_{j_1}^t)\trans] & \bbE_{t}[\vec{y}_{j_1}^t] \bbE_{t}[\vec{y}_{j_2}^t]\trans & \cdots & \bbE_{t}[\vec{y}_{j_1}^t] \bbE_{t}[\vec{y}_{j_n}^t]\trans \\[1mm]
    \bbE_{t}[\vec{y}_{j_2}^t] \bbE_{t}[\vec{y}_{j_1}^t]\trans & \bbE_{t}[\vec{y}_{j_2}^t(\vec{y}_{j_2}^t)\trans] & \cdots & \bbE_{t}[\vec{y}_{j_2}^t]\bbE_{t}[\vec{y}_{j_n}^t]\trans\\
    \vdots & \vdots & \ddots & \vdots \\[1mm]
    \bbE_{t}[\vec{y}_{j_n}^t]\bbE_{t}[\vec{y}_{j_1}^t]\trans & \bbE_{t}[\vec{y}_{j_n}^t]\bbE_{t}[\vec{y}_{j_2}^t]\trans & \cdots &  \bbE_{t}[\vec{y}_{j_n}^t(\vec{y}_{j_n}^t)\trans]
   \end{pmatrix}\\
   &= \!\begin{pmatrix}
   \mat{C}^{t}_{j_1} & \xbar^{t}_{j_1} (\xbar^{t}_{j_2})\trans & \cdots & \xbar^{t}_{j_1}(\xbar^{t}_{j_n})\trans \\[1mm]
    \xbar^{t}_{j_2} (\xbar^{t}_{j_1})\trans & \mat{C}^{t}_{j_2} & \cdots & \xbar^{t}_{j_2}(\xbar^{t}_{j_n})\trans\\
    \vdots & \vdots & \ddots & \vdots \\[1mm]
    \xbar^{t}_{j_n} (\xbar^{t}_{j_1})\trans & \xbar^{t}_{j_n}(\xbar^{t}_{j_2})\trans & \cdots & \mat{C}^{t}_{j_n}
   \end{pmatrix}.
\end{align*}
\end{proof}

\subsection{Generalized Inverse of the Autocorrelation Matrix}
\label{app:f tilde}


\begin{restatable}{proposition}{propinverseconvexhull}\label{prop:inverse convex hull}
    Let $j \in \cJ$ be a decision point. Let $\{a_1,\dots,a_m\}$ be the terminal actions at $j$, and let $\{a_{m+1}, \dots, a_n\}$ be the remaining, non-terminal actions. Then,
    \begin{align*}
    \mat{C}^{t-}_{j,\ast} \defeq \mleft(\!\begin{array}{c@{\hspace{2.5mm}}c@{\hspace{2.5mm}}c|c@{\hspace{2.5mm}}c@{\hspace{2.5mm}}c|c@{\hspace{1mm}}c@{\hspace{1mm}}c}
    \frac{1}{\lambda_1^t} &&&&&\\
    &\ddots&&&&\\
    &&\frac{1}{\lambda_m^t} &&&\\[2mm]
    \hline
    &&& 0&      & & &       & \\[-1mm]
    &&&  &\ddots& & &       & \\[1mm]
    &&&  &      &0 &&       & \\[1mm]
    \hline&&&&&\\[-3mm]
    &&&  &&& \frac{1}{\lambda^{t}_{m+1}} \mat{C}^{t-}_{k_{m+1}}&&\\[-1mm]
    &&&  &&&&\ddots&\\
    &&&  &&& &&\frac{1}{\lambda^{t}_n} \mat{C}^{t-}_{k_n}\\
    \end{array}\!\mright)
    \end{align*}
    is a generalized inverse for the autocorrelation matrix $\mat{C}^{t}_{j}$ defined in \cref{lem:aut convex hull}. The matrix is well defined in virtue of \cref{obs:xbar positive}.
\end{restatable}
\begin{proof}\allowdisplaybreaks
Let $\mat{C}^{t-}_{j}$ be the matrix proposed by the statement. Using \cref{lem:aut convex hull},

{\scriptsize\begin{align*}
&\mat{C}^{t}_{j}\, \mat{C}^{t-}_{j}\, \mat{C}^{t}_{j}\\  &= \mat{C}^{t}_{j}
\mleft(\!\begin{array}{c@{\hspace{2.5mm}}c@{\hspace{2.5mm}}c|c@{\hspace{2.5mm}}c@{\hspace{2.5mm}}c|c@{\hspace{1mm}}c@{\hspace{1mm}}c}
    \frac{1}{\lambda_1^t} &&&&&\\
    &\ddots&&&&\\
    &&\frac{1}{\lambda_m^t} &&&\\[2mm]
    \hline
    &&& 0&      & & &       & \\[-1mm]
    &&&  &\ddots& & &       & \\[1mm]
    &&&  &      &0 &&       & \\[1mm]
    \hline&&&&&\\[-3mm]
    &&&  &&& \frac{1}{\lambda^{t}_1} \mat{C}^{t-}_{k_{m+1}}&&\\[-1mm]
    &&&  &&&&\ddots&\\
    &&&  &&& &&\frac{1}{\lambda^{t}_n} \mat{C}^{t-}_{k_n}\\
    \end{array}\!\mright)
\mleft(\!\begin{array}{c@{\hspace{1mm}}c@{\hspace{1mm}}c@{\hspace{.5mm}}|c@{\hspace{1mm}}c@{\hspace{1mm}}c@{\hspace{.5mm}}|@{\hspace{1mm}}c@{\hspace{1mm}}c@{\hspace{1mm}}c}
        \lambda^t_1&& &&& &&\\
        &\ddots& &&& &&\\
        &&\lambda^t_m &&& &&\\[1mm]
    \hline&&&&&\\[-3mm]
        &&&\lambda^{t}_{m+1} &&& \lambda^{t}_{m+1} (\xbar^{t}_{k_{m+1}})\trans && \\[-1mm]
        &&&&\ddots    && &\ddots& \\[-1mm]
        &&&&& \lambda^{t}_n& && \lambda^{t}_n (\xbar^{t}_{k_n})\trans\\[1mm]
    \hline&&&&&\\[-3mm]
        &&&\lambda^{t}_{m+1}\xbar^{t}_{k_{m+1}}\!\! &&& \lambda^{t}_{m+1} \mat{C}^{t}_{k_{m+1}}\!\! &&\\[-1mm]
        &&&&\ddots&&&\ddots&\\
        &&&&&\lambda^{t}_n\xbar^{t}_{k_n}&&&\!\!\!\lambda^{t}_n \mat{C}^{t}_{k_n}
      \end{array}\!\mright)\\[3mm]
&= \mleft(\!\begin{array}{c@{\hspace{1mm}}c@{\hspace{1mm}}c@{\hspace{.5mm}}|c@{\hspace{1mm}}c@{\hspace{1mm}}c@{\hspace{.5mm}}|@{\hspace{1mm}}c@{\hspace{1mm}}c@{\hspace{1mm}}c}
        \lambda^t_1&& &&& &&\\
        &\ddots& &&& &&\\
        &&\lambda^t_m &&& &&\\[1mm]
    \hline&&&&&\\[-3mm]
        &&&\lambda^{t}_{m+1} &&& \lambda^{t}_{m+1} (\xbar^{t}_{k_{m+1}})\trans && \\[-1mm]
        &&&&\ddots    && &\ddots& \\[-1mm]
        &&&&& \lambda^{t}_n& && \lambda^{t}_n (\xbar^{t}_{k_n})\trans\\[1mm]
    \hline&&&&&\\[-3mm]
        &&&\lambda^{t}_{m+1}\xbar^{t}_{k_{m+1}}\!\! &&& \lambda^{t}_{m+1} \mat{C}^{t}_{k_{m+1}}\!\! &&\\[-1mm]
        &&&&\ddots&&&\ddots&\\
        &&&&&\lambda^{t}_n\xbar^{t}_{k_n}&&&\!\!\!\lambda^{t}_n \mat{C}^{t}_{k_n}
      \end{array}\!\mright)
\mleft(\!\begin{array}{c@{\hspace{1mm}}c@{\hspace{1mm}}c@{\hspace{.5mm}}|c@{\hspace{1mm}}c@{\hspace{1mm}}c@{\hspace{.5mm}}|@{\hspace{1mm}}c@{\hspace{1mm}}c@{\hspace{1mm}}c}
        1&&&&&&\\
        &\ddots&&&&&\\
        &&1&&&&\\
        \hline
        &&& 0&&& &&\\[1mm]
        &&& &\ddots&& && \\
        &&& &&0& &&\\[1mm]
    \hline&&&&&\\[-3mm]
        &&& \mat{C}^{t-}_{k_{m+1}} \xbar^{t}_{k_{m+1}} &&& \mat{C}^{t-}_{k_{m+1}} \mat{C}^{t}_{k_{m+1}} &&\\[-1mm]
        &&& &\ddots&&&\ddots&\\
        &&& &&\mat{C}^{t-}_{k_n} \xbar^{t}_{k_n}&&& \mat{C}^{t-}_{k_n} \mat{C}^{t}_{k_n}
      \end{array}\!\mright)
\\[3mm]
&= \mleft(\!\begin{array}{c@{\hspace{1mm}}c@{\hspace{1mm}}c@{\hspace{.5mm}}|c@{\hspace{1mm}}c@{\hspace{1mm}}c@{\hspace{.5mm}}|@{\hspace{1mm}}c@{\hspace{1mm}}c@{\hspace{1mm}}c}
        \lambda^t_1&& &&& &&\\
        &\ddots& &&& &&\\
        &&\lambda^t_m &&& &&\\[1mm]
    \hline&&&&&\\[-3mm]
        &&&\lambda^{t}_{m+1} &&& \lambda^{t}_{m+1} (\xbar^{t}_{k_{m+1}})\trans && \\[-1mm]
        &&&&\ddots    && &\ddots& \\[-1mm]
        &&&&& \lambda^{t}_n& && \lambda^{t}_{n} (\xbar^{t}_{k_n})\trans\\[1mm]
    \hline&&&&&\\[-3mm]
        &&&\lambda^{t}_{m+1}\xbar^{t}_{k_{m+1}}\!\! &&& \lambda^{t}_{m+1} \mat{C}^{t}_{k_{m+1}}\!\! &&\\[-1mm]
        &&&&\ddots&&&\ddots&\\
        &&&&&\lambda^{t}_n\xbar^{t}_{k_n}&&&\!\!\!\lambda^{t}_n \mat{C}^{t}_{k_n}
      \end{array}\!\mright) = \mat{C}^{t}_{j},
\end{align*}}

where we used \cref{lem:C Cinv y equal y,lem:z Cinv ybar equals 1} in the third equality. This concludes the proof.
\end{proof}

\begin{restatable}{proposition}{propinversecartesian}\label{prop:inverse cartesian}
    Let $k \in \cK$ be an observation point, and let $\next{k} = \{j_1, \dots, j_n\}$ be the decision points immediately reachable after $k$. Finally, for all $i = 1, \dots, n$, let $\mat{C}^{t-}_{j_i}$ be any generalized invers for $\mat{C}^{t}_{j_i}$, and let
    \begin{align*}
      &\vec{\mu}^{t}_{k} \defeq \begin{pmatrix}
            \mat{C}^{t-}_{j_1}\,\xbar^{t}_{j_1}\\
            \vdots\\
            \mat{C}^{t-}_{j_n}\,\xbar^{t}_{j_n}
      \end{pmatrix}.
    \end{align*}
    The matrix
\begin{align*}
   \mat{C}_{k,\ast}^{t-} \defeq \begin{pmatrix}
   \mat{C}^{t-}_{j_1}\!\!\! & &\\
   & \ddots &\\
   && \!\!\! \mat{C}^{t-}_{j_n}\
   \end{pmatrix} - \frac{n-1}{n^2}\cdot\vec{\mu}^{t}_{k}(\vec{\mu}^{t}_{k})\trans.
   \end{align*}
    is a generalized inverse for the autocorrelation matrix $ \mat{C}^{t}_{k}$ defined in \cref{lem:aut cartesian}.
\end{restatable}
\begin{proof}
In order to reduce the notational burden, let
\[
    \mat{C}^{t\,\sim}_{k} \defeq \begin{pmatrix}
   \mat{C}^{t-}_{j_1} & &\\
   & \ddots &\\
   && \mat{C}^{t-}_{j_n}\
   \end{pmatrix}.
\]
With that, we have
\allowdisplaybreaks
\begin{align}
\mat{C}^{t}_{k}\,\mat{C}^{t\,\sim}_{k}\, \mat{C}^{t}_{k} &=
\begin{pmatrix}
           \mat{C}^{t}_{j_1} & \cdots & \xbar^{t}_{j_1}(\xbar^{t}_{j_n})\trans \\
            \vdots & \ddots & \vdots \\[1mm]
            \xbar^{t}_{j_n} (\xbar^{t}_{j_1})\trans & \cdots & \mat{C}^{t}_{j_n}
           \end{pmatrix}
\begin{pmatrix}
   \mat{C}^{t-}_{j_1} & &\\
   & \ddots &\\
   && \mat{C}^{t-}_{j_n}\
   \end{pmatrix}
\begin{pmatrix}
           \mat{C}^{t}_{j_1} & \cdots & \xbar^{t}_{j_1}(\xbar^{t}_{j_n})\trans \\
            \vdots & \ddots & \vdots \\[1mm]
            \xbar^{t}_{j_n} (\xbar^{t}_{j_1})\trans & \cdots & \mat{C}^{t}_{j_n}
           \end{pmatrix} \nonumber\\[2mm]
&= \begin{pmatrix}
           \mat{C}^{t}_{j_1} & \cdots & \xbar^{t}_{j_1}(\xbar^{t}_{j_n})\trans \\
            \vdots & \ddots & \vdots \\[1mm]
            \xbar^{t}_{j_n} (\xbar^{t}_{j_1})\trans & \cdots & \mat{C}^{t}_{j_n}
           \end{pmatrix}
\begin{pmatrix}
           \mat{C}^{t-}_{j_1}\,\mat{C}^{t}_{j_1} & \cdots & \mat{C}^{t-}_{j_1}\,\xbar^{t}_{j_1}(\xbar^{t}_{j_n})\trans \\
            \vdots & \ddots & \vdots \\[1mm]
            \mat{C}^{t-}_{j_n}\,\xbar^{t}_{j_n} (\xbar^{t}_{j_1})\trans & \cdots & \mat{C}^{t-}_{j_n}\,\mat{C}^{t}_{j_n}
           \end{pmatrix} \nonumber\\[2mm]
&= \begin{pmatrix}
  \mat{C}^{t}_{j_1} + (n-1)\,\xbar^{t}_{j_1}(\xbar^{t}_{j_1})\trans &
  n\,\xbar^{t}_{j_1}(\xbar^{t}_{j_2})\trans & \cdots &
  n\, \xbar^{t}_{j_1}(\xbar^{t}_{j_n})\trans\\[1mm]
  n\,\xbar^{t}_{j_2}(\xbar^{t}_{j_1})\trans & \mat{C}^{t}_{j_2} + (n-1)\,\xbar^{t}_{j_2}(\xbar^{t}_{j_2})\trans &
  \cdots &
  n\, \xbar^{t}_{j_2}(\xbar^{t}_{j_n})\trans\\
  \vdots & \vdots & \ddots & \vdots \\[1mm]
  n\,\xbar^{t}_{j_n}(\xbar^{t}_{j_1})\trans & n\, \xbar^{t}_{j_n}(\xbar^{t}_{j_2})\trans & \cdots
   &\mat{C}^{t}_{j_n} + (n-1)\,\xbar^{t}_{j_n}(\xbar^{t}_{j_n})\trans &
\end{pmatrix}\nonumber\\[2mm]
  &= \mat{C}^{t}_{k} + (n-1)\,\xbar^{t}_{k}(\xbar^{t}_{k})\trans,\label{eq:C Ctilde C}
\end{align}
where we repeatedly used \cref{lem:C Cinv y equal y,lem:z Cinv ybar equals 1}. At the same time, we have
\begin{align}
  \mat{C}^{t}_k\, \vec{\mu}^{t}_{k} &=
    \begin{pmatrix}
           \mat{C}^{t}_{j_1} & \cdots & \xbar^{t}_{j_1}(\xbar^{t}_{j_n})\trans \\
            \vdots & \ddots & \vdots \\[1mm]
            \xbar^{t}_{j_n} (\xbar^{t}_{j_1})\trans & \cdots & \mat{C}^{t}_{j_n}
    \end{pmatrix}
    \begin{pmatrix}
      \mat{C}^{t-}_{j_1} \xbar^{t}_{j_1} \\
      \vdots \\
      \mat{C}^{t-}_{j_n} \xbar^{t}_{j_n}
    \end{pmatrix}
    = \begin{pmatrix}
         n \,\xbar^{t}_{j_1} \\
         \vdots \\
         n \,\xbar^{t}_{j_n}
       \end{pmatrix} = n \,\xbar^{t}_{k},\label{eq:C mu}
\end{align}
where again we used \cref{lem:C Cinv y equal y,lem:z Cinv ybar equals 1}.
Putting~\eqref{eq:C Ctilde C} and~\eqref{eq:C mu} together, we obtain
\begin{align*}
\mat{C}^{t}_{k}\mleft(\mat{C}^{t\,\sim}_{k} - \frac{n-1}{n^2}\cdot \vec{\mu}^{t}_{k}(\vec{\mu}^{t}_{k})\trans\mright)\mat{C}^{t}_{k} =
  \mat{C}^{t}_{k} + (n-1)\,\xbar^{t}_{k}(\xbar^{t}_{k})\trans - (n-1)\,\xbar^{t}_{k}(\xbar^{t}_{k})\trans
  = \mat{C}^{t}_{k},
\end{align*}
as we wanted to show.
\end{proof}

    \newpage
    \section{Loss Estimate}

\propftilde*
\begin{proof}
    For all $\vec{z}\in\lin\pure$,
    \begin{align*}
        \vec{z}\trans \bbE_t[\tilde{\vec{\ell}}^t] &= \vec{z}\trans\bbE_t\mleft[((\vec{y}^t)\trans\!\vec{\ell}^t)\, \mat{C}^{t-}\, \vec{y}^t\mright] + \vec{z}^\top \bbE_t[\vec{b}^t]\\
        &= \vec{z}\trans \bbE_t\mleft[ \mat{C}^{t-}\, \vec{y}^t (\vec{y}^t)\trans\, \vec{\ell}^t \mright]\\
         &= \vec{z}\trans \mat{C}^{t-}\, \mat{C}^t \vec{\ell}^t.
    \end{align*}
    Using the inclusion $\lin \pure \subseteq \Span \pure$ together with \cref{lem:C Cinv y equal y} gives the statement.
\end{proof}

\cref{prop:inverse convex hull,prop:inverse cartesian} give explicit formulas for inductively constructing a generalized inverse $\mat{C}_\ast^{t-}$ of the autocorrelation matrix $\mat{C}^t$ of the natural sampling scheme for sequential decision processes needed in \cref{prop:f tilde}.

In the next subsection we focus on studying the particular orthogonal vector $\vec{b}^t_*$ that we use in our algorithm.

\subsection{Orthogonal Vector}

The primary role of the orthogonal vector $\vec{b}^t$ is to prevent the subtraction $-(n-1)/n\vec{\mu}_k^t$ in \cref{prop:inverse cartesian} from creating negative entries in the loss estimate. It is important to prevent negative entries in the loss estimate because our refined local-norm-based bound for the online mirror descent algorithm paired with the dilated entropy DGF (\cref{thm:omd regret bound}) is contingent on it.

The insight that we use in the construction of $\vec{b}_*^t$ is that the vector $\vec{\mu}_k^t$ in \cref{prop:inverse cartesian} is orthogonal to $\lin\pure_k$ for all $k \in \cK$ because of \cref{lem:z Cinv ybar equals 1}. So, we construct the vector $\vec{b}^t \perp \lin\pure$ inductively to cancel the effect of all $\vec{\mu}_k^t$, as follows:
\begin{itemize}[nolistsep,itemsep=1mm,leftmargin=8mm]
    \item Consider a decision point $j \in \cJ$. Let $\{a_1, \dots, a_m\}$ at $j$ be the terminal actions, and let the remaining actions $\{a_{m+1},\dots, a_n\}$ be non-terminal, leading to observation nodes $k_{m+1},\dots, k_n$ respectively. Furthermore, let the $a$ be the action that was selected at $j$ by the pure sequence-form strategy $\vec{y}^t$. If $a \le m$, we let
        \begin{equation}\label{eq:bj terminal}
          \vec{b}_{j,*}^t = \mleft(\frac{1}{\lambda^t_a}(N_j - 1) \vec{e}_a, \vec{0}, \dots, \vec{0}\mright).
        \end{equation}
        Otherwise, $a \in \{m+1,\dots,n\}$ is a non-terminal action, and we recursively let
        \begin{equation}\label{eq:bj nonterminal}
          \vec{b}_{j,*}^t = \mleft(\frac{1}{\lambda^t_a}(N_j - N_{k_a}) \vec{e}_a, \vec{0}, \dots, \frac{1}{\lambda^t_a}\vec{b}^t_{k_a,*} , \vec{0}\mright).
        \end{equation}

        So, in particular, at all terminal decision points $j\in\cJ$ we have $\vec{b}^t_j = \vec{0}$ since $N_j = 1$.
    \item At all observation points $k\in \cK$, we let
    \begin{equation}\label{eq:bk}
        \vec{b}^{t}_{k,*} = (\vec{b}^t_{j_1,*}, \dots, \vec{b}^t_{j_n,*}) + \frac{n-1}{n} \vec{\mu}^t_k,
    \end{equation}
  where $\{j_1, \dots, j_n\} = S_k$.
\end{itemize}

The vector $\vec{b}_*^t$ we just defined satisfies the following property.

\begin{lemma}\label{lem:bt}
    At all times $t$ and for all nodes $v \in \cJ\cup\cK$,
    \[
        \bbE_t[\vec{b}^t_{v,*}]^\top \vec{z}_v = N_v - 1\qquad \forall\, \vec{z}_v \in \seqf_v.
    \]
\end{lemma}
\begin{proof}
    We prove the lemma by induction on the structure of the decision process:
    \begin{itemize}
      \item \textbf{First case:} $v \in \cJ$ is a terminal decision point. In this case $N_v = 1$. Furthermore, by construction $\vec{b}^t_{v,*} = \vec{0}$, so $\bbE_t[\vec{b}_{v,*}^t]^\top \vec{z}_v = 0 = N_v - 1$.
      \item \textbf{Second case:} $v \in \cJ$ is a (generic) decision point. In this case, taking the expectation of $\vec{b}_v^t$ (\cref{eq:bj terminal,eq:bj nonterminal}) over all possible choices of action $a$ at $v$ (action $a$ is chosen with probability $\lambda^t_a$, yields
          \[
            \bbE[\vec{b}_{v,*}^t] = \mleft(N_v - 1, \dots, N_v - 1, N_v - N_{k_{m+1}}, \dots, N_v - N_{k_{n}}, \bbE_t[\vec{b}_{k_{m+1},*}^t], \dots, \bbE_t[\vec{b}_{k_{n},*}^t]\mright).
          \]
          From \eqref{eq:structure decision point}, any $\vec{z}_v \in \seqf_v$ can be written in the from $(\lambda'_1, \dots, \lambda'_n, \lambda'_{m+1} \vec{z}_{k_{m+1}}, \dots, \lambda'_n \vec{z}_{k_n})$ for appropriate $(\lambda'_1, \dots, \lambda'_n) \in \Delta^n$ and $\vec{z}_{k_i} \in \seqf_{k_i}$ for $i \in \{m+1, \dots, n\}$. So,
         \[
         \bbE[\vec{b}^t_{v,*}]^\top \vec{z}_v = \sum_{i=1}^m (N_v - 1) \lambda'_i + \sum_{i=m+1}^n \mleft( \lambda'_i (N_v - N_{k_i}) + \lambda'_i \bbE[\vec{b}^t_{k_i,*}]^\top \vec{z}_{k_i}\mright).
         \]
         Finally, using the inductive hypothesis $\bbE[\vec{b}^t_{k_i,*}]^\top \vec{z}_{k_i} = N_{k_i} - 1$ for all $i \in \{m+1, \dots, n\}$,
         \begin{align*}
            \bbE[\vec{b}^t_{v,*}]^\top \vec{z}_v &= \sum_{i=1}^m (N_v - 1) \lambda'_i + \sum_{i=m+1}^n (N_v - N_{k_i}) + (N_{k_i} - 1)\\
                &= \sum_{i=1}^m (N_v - 1) \lambda'_i  + \sum_{i=m+1}^n \lambda'_i (N_v - 1) = N_v - 1,
         \end{align*}
         where we used the fact that $(\lambda'_1, \dots, \lambda'_n) \in \Delta^n$ in the last equality.
    \item \textbf{Third case:} $v \in \cK$ is an observation point. In this case, the expectation of $\vec{b}_v^t$ (\cref{eq:bk}) is
        \[
            \bbE[\vec{b}^t_{v,*}] = (\bbE[\vec{b}^t_{j_1,*}], \dots \bbE[\vec{b}^t_{j_n,*}]) + \frac{n-1}{n} \vec{\mu}_k^t.
        \]
        So, using the inductive hypothesis and \cref{lem:z Cinv ybar equals 1}, we find that
        \[
            \bbE[\vec{b}^t_{v,*}]^\top \vec{z}_v = \sum_{i=1}^n (N_{j_i} - 1) + \frac{n-1}{n} \cdot n = -1 + \sum_{i=1}^n N_{j_i}.
        \]
        Using the fact that $N_v = \sum_{i=1}^n N_{j_i}$ yields the statement.\qedhere
    \end{itemize}
\end{proof}

Since any point in $\lin\pure$ is the difference of two points in $\seqf$, \cref{lem:bt} immediately implies the following.

\begin{corollary}
  At all times $t$,
  $
    \bbE_t[\vec{b}^t_{v,*}] \perp \vec{z}_v \quad \forall \vec{z}_v \in \lin \pure.
  $
\end{corollary}

So, our (inductive) construction of $\vec{b}^t$ satisfies the requirements of \cref{prop:f tilde}. Below we will show that it always leads to non-negative loss estimates.

\subsection{Algorithm for Constructing the Loss Estimate in Linear Time}

Let $l^t \defeq (\vec{\ell}^t)^\top \vec{y}^t$ denote the bandit feedback (loss evaluation) received at time $t$ by the decision maker. In this section, we will show that the unbiased loss estimate $\tilde{\vec{\ell}}^t = l^t\cdot \mat{C}^{t-} \vec{y}^t + \vec{b}^t$, as defined in \cref{prop:f tilde}, can be inductively computed in linear time in the size of the decision process as in \cref{algo:ell tilde} (which is a copy of the algorithm in the body of the paper).

\begin{minipage}[ht]{\linewidth}\centering
        \SetInd{0.2em}{0.4em}%
\scalebox{1}{\begin{algorithm}[H]\small
          \caption{\normalfont$\textsc{LossEstimate}(l, \xbar^t, \vec{y}^t)$\hspace*{-.8cm}}\label{algo:ell tilde2}
            \DontPrintSemicolon
            $\tilde{\vec{\ell}}^t \gets \vec{0}\in\Rpp^{|\Sigma|}$\;
            \Subr{\normalfont$\textsc{Traverse}(v, \alpha_v)$}{
                \uIf{$v \in \cK$} {
                    \For{$s\in S_v$}{
                        $\displaystyle\textsc{Traverse}\mleft(\rho(v,s), \frac{\alpha_v}{|S_v|} + \frac{|S_v|-1}{|S_v|}(1 - l)y^t_{p_v}\mright)$\hspace*{-1cm}\;
                    }
                }
                \Else(\Comment*[f]{\color{commentcolor}that is, $v\in\cJ$]\hspace*{-4mm}}){
                    \For{$a \in A_v$}{
                        \uIf{$\rho(v,a) \neq \terminalnode$}{
                            $\displaystyle\vec{\ell}^t_{ja} \gets \frac{y^t_{va}}{{x}^t_{va}} (N_{v} - N_{\rho(v,a)})$\;
                            $\displaystyle\textsc{Traverse}\mleft(\rho(v,a), \frac{{x}_{p_v}}{{x}_{va}}\alpha_v \mright)$\;
                        }\ElseIf{$\rho(v,a) = \terminalnode$}{
                            $\displaystyle\vec{\ell}^t_{ja} \gets \frac{\alpha_v}{{x}^t_{p_v}}+ \frac{y^t_{va}}{{x}^t_{va}}(l + N_v - 1)$\;
                        }
                    }
                }
            }
            \Hline{}
            $\textsc{Traverse}(r, 0)$\Comment*{\color{commentcolor}$r$: root of the decision process]\hspace*{-4mm}}
            \textbf{return} $\tilde{\vec{\ell}}^t$\;
        \end{algorithm}}
    \end{minipage}

\proplossestimate*
\begin{proof}\allowdisplaybreaks
Our construction revolves around the following quantity, indexed over nodes $v \in \cJ \cup \cK$:
\[
    \vec{h}^t_{v}(\alpha) \defeq \frac{y^t_{p_v}}{{x}^t_{p_v}} \, \tilde{\vec{\ell}}^t_v + \frac{\alpha}{{x}^t_{p_v}} \mat{C}_{v,*}^{t-} \xbar_v^t = l^t\cdot \frac{y^t_{p_v}}{{x}^t_{p_v}}\mat{C}_{v,*}^{t-}\, \vec{y}_{v}^t + \beta\,\vec{b}_{v,*}^t + \frac{\alpha}{{x}^t_{p_v}} \mat{C}_{v,*}^{t-} \xbar_v^t.
\]

The loss estimate $\tilde{\vec{\ell}}^t$ coincides with $\vec{h}_r^t(0)$ where $r$ is the root of the sequential decision problem. We now show that $\vec{h}^t_v(\alpha,\beta)$ can be constructed inductively over the structure of the sequential decision process:
\begin{itemize}[nolistsep,itemsep=1mm,leftmargin=5mm]
  \item If $v\in\cJ$ is a decision node, let $\{a_1, \dots, a_m\}$ be the terminal actions at $v$, and let $\{a_{m+1},\dots, a_n\}$ be the non-terminal actions at $v$ (if any). Finally, let $k_i = \rho(v, a_i)$ for $i=m+1,\dots, n$. In accordance with \cref{eq:structure decision point}, $\xbar^t_v = (\lambda_1^t, \dots, \lambda_n^t, \lambda_{k_{m+1}}^t \xbar_{k_i}^t ,\dots, \lambda_{k_n}^t\xbar_{k_n}^t)$. Similarly, we let $y^t_v = (\nu_1^t, \dots, \nu_n^t, \nu_{k_{m+1}}^t \vec{y}_{k_i}^t ,\dots, \nu_{k_n}^t\vec{y}_{k_n}^t)$.

      Using \cref{prop:inverse convex hull} and \cref{eq:bj terminal,eq:bj nonterminal}, we can write
      \begin{align}\allowdisplaybreaks
          \vec{h}^t_v(\alpha) &=  l^t\frac{y^t_{p_v}}{{x}^t_{p_v}} \begin{pmatrix}
            \frac{\nu^t_1}{\lambda^t_1}\\
            \vdots\\
            \frac{\nu^t_{m}}{\lambda^t_{m}}\\
            0\\
            \vdots\\
            0\\
            \frac{\nu^t_{m+1}}{\lambda^t_{m+1}}\mat{C}^{t-}_{k_{m+1},*} \vec{y}^t_{k_{m+1}}\\
            \vdots\\
            \frac{\nu^t_{n}}{\lambda^t_{n}}\mat{C}^{t-}_{k_{n},*} \vec{y}^t_{k_{n}}\\
          \end{pmatrix} + \frac{y^t_{p_v}}{{x}^t_{p_v}}\begin{pmatrix}
            \frac{\nu^t_1}{\lambda_1^t}(N_v - 1)\\
            \vdots\\
            \frac{\nu^t_m}{\lambda_m^t}(N_v - 1)\\
            \frac{\nu^t_{m+1}}{\lambda_{m}^t}(N_v - N_{k_{m}})\\
            \vdots\\
            \frac{\nu^t_{n}}{\lambda_n^t}(N_v - N_{k_{n}})\\
            \frac{\nu^t_{m+1}}{\lambda^t_{m+1}}\vec{b}^t_{k_{m+1},*}\\
            \vdots\\
            \frac{\nu^t_{n}}{\lambda^t_{n}}\vec{b}^t_{k_{n},*}\\
          \end{pmatrix} + \frac{\alpha}{{x}^t_{p_v}} \begin{pmatrix}
            1\\
            \vdots\\
            1\\
            0\\
            \vdots\\
            0\\
            \mat{C}^{t-}_{k_{m+1},*} \xbar^t_{k_{m+1}}\\
            \vdots\\
            \mat{C}^{t-}_{k_{n},*} \xbar^t_{k_{n}}\\
          \end{pmatrix} \nonumber\\
          &= \begin{pmatrix}
            \frac{\alpha}{{x}^t_{p_v}}+ \frac{y^t_{p_v}\nu_1^t}{{x}^t_{p_v}\lambda^t_1}(l^t  + N_v - 1) \\
            \vdots\\
            \frac{\alpha}{{x}^t_{p_v}} + \frac{y^t_{p_v}\nu_{m}^t}{{x}^t_{p_v}\lambda^t_{m}}(l^t  + N_v - 1)\\
            \frac{y^t_{p_v}\nu_{m+1}^t}{{x}^t_{p_v}\lambda^t_{m+1}}(N_v - N_{k_{m+1}})\\
            \vdots\\
            \frac{y^t_{p_v}\nu_{n}^t}{{x}^t_{p_v}\lambda^t_{n}}(N_v - N_{k_{n}})\\
            \frac{y^t_{p_v}\nu_{m+1}^t}{{x}^t_{p_v}\lambda^t_{m+1}}(l^t \cdot \mat{C}^{t-}_{k_{m+1},*} \vec{y}^t_{k_{m+1}} + \vec{b}^t_{k_{m+1},*}) + \frac{\alpha\lambda^t_{m+1}}{{x}^t_{p_v}\lambda^t_{m+1}} \mat{C}^{t-}_{k_{m+1},*} \vec{y}^t_{k_{m+1}}\\
            \vdots\\
            \frac{y^t_{p_v}\nu_{n}^t}{{x}^t_{p_v}\lambda^t_{n}}(l^t \cdot \mat{C}^{t-}_{k_{n},*} \vec{y}^t_{k_{n}} + \vec{b}^t_{k_{n},*}) + \frac{\alpha\lambda^t_{n}}{{x}^t_{p_v}\lambda^t_{n}} \mat{C}^{t-}_{k_{n},*} \vec{y}^t_{k_{n}}\\
          \end{pmatrix} \nonumber\\
          &= \begin{pmatrix}
            \frac{\alpha}{{x}^t_{p_v}} + \frac{y^t_{va_1}}{{x}^t_{va_1}}(l^t  + N_v - 1) \\
            \vdots\\
            \frac{\alpha}{{x}^t_{p_v}} + \frac{y^t_{va_{m}}}{{x}^t_{va_{m}}}(l^t  + N_v - 1)\\
            \frac{y^t_{va_{m+1}}}{{x}^t_{va_{m+1}}}(N_v - N_{k_{m+1}})\\
            \vdots\\
            \frac{y^t_{va_{n}}}{{x}^t_{va_{n}}}(N_v - N_{k_{n}})\\
            \vec{h}_{k_{m+1}}^t\mleft(\alpha\lambda^t_{m+1}\mright)\\
            \vdots\\
            \vec{h}_{k_{n}}^t\mleft(\alpha\lambda^t_{n}\mright)\\
          \end{pmatrix}\label{eq:explore j}
      \end{align}
  \item If $v\in \cK$ is an observation point, let $\{j_1,\dots,j_n\} = \next{v}$ be the decision points that are immediately reachable from $v$. From \cref{eq:structure observation point} and \cref{eq:pure k} we can write $\xbar^t_v = (\xbar^t_{j_1}, \dots, \xbar^t_{j_n})$ and $\vec{y}^t_v = (\vec{y}^t_{j_1}, \dots, \vec{y}^t_{j_n})$. Using \cref{prop:inverse cartesian} and \cref{eq:bk}, we can recursively expand $\mat{C}^{t-}_{v,*}$ and $\vec{b}^t_{v,*}$ and write
      \begin{align}
            \vec{h}_v^t(\alpha) &= l^t \frac{y^t_{p_v}}{{x}^t_{p_v}}\begin{pmatrix}
                \mat{C}^{t-}_{j_1,*}&&\\
                &\ddots&\\
                &&\mat{C}^{t-}_{j_n,*}
            \end{pmatrix} \begin{pmatrix}
                \vec{y}^t_{j_1}\\
                \vdots\\
                \vec{y}^t_{j_n}
            \end{pmatrix} + \frac{y^t_{p_v}}{{x}^t_{p_v}}(1 - l^t) \frac{n-1}{n}\begin{pmatrix}
                \mat{C}^{t-}_{j_1,*} \xbar^t_{j_1}\\
                \vdots\\
                \mat{C}^{t-}_{j_n,*} \xbar^t_{j_n}\\
            \end{pmatrix}
            + \frac{y^t_{p_v}}{{x}^t_{p_v}} \begin{pmatrix}
                \vec{b}^{t}_{j_1,*}\\
                \vdots\\
                \vec{b}^{t}_{j_n,*}\\
            \end{pmatrix} \nonumber\\
            &\hspace{10cm}
            + \frac{\alpha}{n{x}^t_{p_v}} \begin{pmatrix}
                \mat{C}^{t-}_{j_1,*} \xbar^t_{j_1}\\
                \vdots\\
                \mat{C}^{t-}_{j_n,*} \xbar^t_{j_n}\\
            \end{pmatrix} \nonumber\\
        &= \begin{pmatrix}
                \vec{h}_{j_1}^{t}\mleft(\frac{\alpha}{n} + (1 - l^t)\frac{n-1}{n}y^t_{p_v}\mright)\\
                \vdots\\
                \vec{h}_{j_n}^{t}\mleft(\frac{\alpha}{n} + (1 - l^t)\frac{n-1}{n}y^t_{p_v}\mright)\\
            \end{pmatrix}.\label{eq:explore k}
      \end{align}
\end{itemize}

\cref{eq:explore j,eq:explore k} together show that $\vec{h}^t_v(\alpha)$ can be computed inductively along the structure of the decision problem.

Routine $\textsc{Traverse}(v,\alpha_v)$ in \cref{algo:ell tilde} implementes the recursive construction $\vec{h}^t(\alpha)$ by operationalizing \cref{eq:explore j,eq:explore k} into code. Specifically, for all $v \in \cJ\cup\cK$ and $\alpha_v$, $\textsc{Traverse}(v,\alpha_v)$ computes $\vec{h}^t_v(\alpha_v)$. Lines~7-12 of \cref{algo:ell tilde2} correspond to \cref{eq:explore j}, while Line~5 corresponds to \cref{eq:explore k}.

Since the algorithm performs constant work per each sequence, the algorithm runs in linear time in the number of sequences $|\Sigma|$ of the sequential decision process.
\end{proof}

\subsection{Expected Local Dual Norm of Loss Estimate}
\label{app:f tilde norm}

We start from a technical lemma.

\begin{lemma}\label{lem:Cinv xbar le 2}
    For all nodes $v \in \cJ \cup \cK$,
    \[
        \dn{\mat{C}^{t-}_{v,\ast} \,\xbar_v}{\xbar_v}^2 \le 2.
    \]
\end{lemma}
\begin{proof}
    We prove the slightly stronger bound
    \[
        \dn{\mat{C}^{t-}_{v,\ast} \,\xbar_v}{\xbar_v}^2 \le 2 - \frac{1}{w_v}.
    \]
     by induction on the structure of the decision process:
    \begin{itemize}
      \item \textbf{First case:} $v \in \cJ$ is a terminal decision point. This is the base case.
      In this case, $\xbar^t = (\frac{1}{\lambda^t_1}, \dots, \frac{1}{\lambda^t_n}) \in \Delta^n$, and from \cref{prop:inverse convex hull} we have
          \[
            \mat{C}_{v,*}^{t-} = \begin{pmatrix}\frac{1}{\lambda_1^t}&&\\&\ddots&\\&&\lambda_n^t\end{pmatrix}.
          \]
          So, using \cref{prop:inverse convex hull} and \cref{cor:dual norm},
          \[
            \mat{C}^{t-}_{v,\ast} \,\xbar_v = (1, \dots, 1) \implies \dn{\mat{C}^{t-}_{v,\ast} \,\xbar_v}{\xbar_v}^2 = \sum_{i=1}^n \frac{\lambda_i^t}{w_v} = \frac{1}{w_v} = \frac{1}{2} \le 2 - \frac{1}{2} = 2 - \frac{1}{w_v}.
          \]
      \item \textbf{Second case:} $v \in \cJ$ is a (generic) decision point. In this case,
          \[
            \xbar^t = (\lambda^t_1, \dots, \lambda^t_n, \lambda^t_{m+1} \xbar^t_{k_{m+1}}, \dots, \lambda^t_n \xbar^t_{k_n})
          \]
          for appropriate $(\lambda^t_1, \dots, \lambda^t_n) \in \Delta^n $ (\cref{eq:structure decision point}). So, using \cref{prop:inverse convex hull},
         \[
            \mat{C}^{t-}_{v,\ast} \,\xbar_v = (1, \dots, 1, 0, \dots, 0, \mat{C}^{t-}_{m+1} \xbar^t_{k_{m+1}}, \dots, \mat{C}^{t-}_n \xbar^t_{k_n}).
         \]
         Using \cref{cor:dual norm}, we have
         \begin{align*}
            \dn{\mat{C}^{t-}_{v,\ast}}{\xbar_v}^2 &= \sum_{i=1}^n \frac{\lambda_i^t}{w_v} + \sum_{i=m+1}^n \frac{\lambda_i^t}{w_v} \mleft((\xbar_{k_i}^t)^\top \mat{C}^{t-}_{k_i,*} \xbar^t_{k_i}\mright)^2 + \sum_{i=m+1}^n \lambda_i^t \dn{\mat{C}^{t-}_{k_i}\xbar^t_{k_i}}{\xbar^t_{k_i}}^2 \\
            &= \sum_{i=1}^n \frac{\lambda_i^t}{w_v} + \sum_{i=m+1}^n \lambda_i^t \dn{\mat{C}^{t-}_{k_i}\xbar^t_{k_i}}{\xbar^t_{k_i}}^2\\
            &= \frac{1}{w_v} + \sum_{i=m+1}^n \lambda_i^t \dn{\mat{C}^{t-}_{k_i}\xbar^t_{k_i}}{\xbar^t_{k_i}}^2\\
            &\le \frac{1}{w_v} + \max_{i=m+1}^n \dn{\mat{C}^{t-}_{k_i}\xbar^t_{k_i}}{\xbar^t_{k_i}}^2\\
            &\le \frac{1}{w_v} + \max_{i=m+1}^n \mleft\{2 - \frac{1}{w_{k_i}}\mright\}
         \end{align*}
         where the second equality holds from \cref{lem:z Cinv ybar equals 1}, the first inequality follows from using the fact that $(\lambda^t_1, \dots, \lambda^t_n) \in \Delta^n$, and the second inequality follows from the inductive hypothesis. Using the fact that $w_{v} \ge 2 w_{k_i}$ for all $i=m+1, \dots, n$ (\cref{def:dilated entropy}), we obtain
        \[
            \dn{\mat{C}^{t-}_{v,\ast}}{\xbar_v}^2 \le \frac{1}{w_v} + 2 - \frac{2}{w_v} = 2 - \frac{1}{w_v},
        \]
        completing the inductive step.
    \item \textbf{Third case:} $v \in \cK$ is an observation point. Let $\{j_1, \dots, j_n\} = \next{v}$. \cref{prop:inverse cartesian} yields
        \[
            \mat{C}^{t-}_{v,*} \xbar^t_{v} = \frac{1}{n}\mleft(\mat{C}^{t-}_{j_1,*} \xbar^t_{j_i}, \dots, \mat{C}^{t-}_{j_n,*} \xbar^t_{j_n}\mright).
        \]
        So,
        \begin{align*}
            \mat{C}^{t-}_{v,*} \xbar^t_{v} &= \frac{1}{n^2}\sum_{i=1}^n \dn{\mat{C}^{t-}_{j_i,*}}{\xbar^t_{j_i}}^2\\
                &\le \frac{1}{n^2} \sum_{i=1}^n \mleft(2 - \frac{1}{w_{j_i}}\mright)\\
                &\le \frac{1}{n} \sum_{i=1}^n \mleft(2 - \frac{1}{w_{j_i}}\mright)\\
                &\le \max_{i=1}^n \mleft\{2 - \frac{1}{w_{j_i}}\mright\}\le 2 - \frac{1}{w_v},
        \end{align*}
        where the last inequality follows from the observation that $w_{j_i} \le w_v$ for all $i = 1, \dots, n$.
        \qedhere
    \end{itemize}
\end{proof}

\thmexpecteddualnorm*
\begin{proof}
    In order to prove the statement, we will prove the slightly stronger statement that for all $v \in \cJ \cup \cK$ and $\alpha \in[0,1]$,
    \[
        \bbE_t\mleft[\dn{\tilde{\vec{\ell}} + \alpha\mat{C}^{t-}_{v,\ast} \xbar^t_v}{\xbar_v^t}^2\mright] \le  2 N_v^2 + |\Sigma|^2\mleft(2 - \frac{1}{w_{v}}\mright) N_v.
    \]
    Our proof is by induction over the sequential decision process structure:
    \begin{itemize}
      \item \textbf{First case:} $v \in \cJ$ is a terminal decision node. Let $\xbar^t_v = (\lambda^t_1, \dots, \lambda^t_n) \in \Delta^n$, in accordance with \cref{eq:structure simplex}. By the sampling scheme, $\vec{y}^t_v = \vec{e}_a$ with probability $\lambda^t_a$. Furthermore,
          \[
            \tilde{\vec{\ell}}^t_v = l^t \mat{C}^{t-}_{v,\ast} \vec{y}^t_{v} + \vec{b}^t_{v,*} + \alpha \mat{C}^{t-}_{v, \ast} \xbar_v^t = \begin{pmatrix}\alpha \\ \vdots \\ \alpha\end{pmatrix} + \frac{l^t}{\lambda^t_a} \vec{e}_a.
          \]
          So,
          \begin{align*}
            \bbE_t\mleft[\dn{\tilde{\vec{\ell}} + \alpha\mat{C}^{t-}_{v,\ast} \xbar^t_v}{\xbar_v^t}^2\mright] &= \sum_{i=1}^n \lambda^t_i \dnlarge{\begin{pmatrix}\alpha \\ \vdots \\ \alpha\end{pmatrix} + \frac{l^t}{\lambda^t_i} \vec{e}_i}{\xbar^t_v}^2\\
            &= \frac{1}{w_v}\sum_{i=1}^n \lambda_i^t \mleft(\lambda_i \mleft(\alpha + \frac{l^t}{\lambda^t_i}\mright)^2 + \sum_{j\neq i} \lambda^t_j \alpha^2 \mright)\\
            &= \frac{1}{2} \sum_{i=1}^n \lambda_i^t \mleft(\alpha^2 + 2\alpha\, l^t + \frac{(l^t)^2}{\lambda^t_i}\mright)\\
            &= \frac{1}{2} \mleft(\alpha^2 + 2\alpha l^t + n (l^t)^2\mright),
          \end{align*}
          where the second equality follows from expanding the definition of local dual norm. Using the hypothesis that $l^t \in [0,1]$, for all $\alpha \in [0,1]$ we find
            \begin{align*}
                \bbE_t\mleft[\dn{\tilde{\vec{\ell}} + \alpha\mat{C}^{t-}_{v,\ast} \xbar^t_v}{\xbar_v^t}^2\mright] \le \frac{3 + n}{2} \le 2 + n^2 \le 2 N_v^2 + |\Sigma|^2\mleft(2 - \frac{1}{w_{v}}\mright) N_v.
            \end{align*}
      \item \textbf{Second case:} $v \in \cK$. Let $\{j_1, \dots, j_n\} = \next{v}$. Furthermore, let $\xbar_v^t = (\xbar_{j_1}^t, \dots, \xbar_{j_n}^t)$ and $\vec{y}_v^t = (\vec{y}_{j_1}^t, \dots, \vec{y}_{j_n}^t)$ in accordance with \eqref{eq:structure observation point} and \eqref{eq:pure k}, respectively. Using \cref{prop:inverse cartesian}, we have
          \begin{align*}
            \tilde{\vec{\ell}} + \alpha\mat{C}^{t-}_{v,\ast} \xbar^t_v &=
                l^t \mat{C}^{t-}_{v,\ast} \vec{y}_v^t + \vec{b}^t_v + \alpha\mat{C}^{t-}_{v,\ast} \xbar^t_v\\
                &= l^t \begin{pmatrix}
                     \mat{C}^{t-}_{j_1,\ast} & & \\
                     &\ddots & \\
                     && \mat{C}^{t-}_{j_n,\ast}
                   \end{pmatrix}\begin{pmatrix}
                        \vec{y}^t_{j_1}\\
                        \vdots\\
                        \vec{y}^t_{j_n}
                   \end{pmatrix}
                   + \begin{pmatrix}
                        \vec{b}^t_{j_1,\ast}\\
                        \vdots\\
                        \vec{b}^t_{j_n,\ast}
                   \end{pmatrix} + \frac{\alpha + (n-1)(1-l^t)}{n}\begin{pmatrix}
                        \mat{C}^{t-}_{j_1,\ast} \xbar^t_{j_1}\\
                        \vdots\\
                        \mat{C}^{t-}_{j_n,\ast} \xbar^t_{j_n}
                   \end{pmatrix}\\
                &= \begin{pmatrix}
                        \tilde{\vec{\ell}}^t_{j_1}\\
                        \vdots\\
                        \tilde{\vec{\ell}}^t_{j_n}
                    \end{pmatrix} + \frac{\alpha + (n-1)(1-l^t)}{n}\begin{pmatrix}
                        \mat{C}^{t-}_{j_1,\ast} \xbar^t_{j_1}\\
                        \vdots\\
                        \mat{C}^{t-}_{j_n,\ast} \xbar^t_{j_n}
                   \end{pmatrix}\\
                &= \begin{pmatrix}
                        \tilde{\vec{\ell}}^t_{j_1} + \alpha'\cdot\mat{C}^{t-}_{j_1,\ast} \xbar^t_{j_1} \\
                        \vdots\\
                        \tilde{\vec{\ell}}^t_{j_n} + \alpha'\cdot\mat{C}^{t-}_{j_n,\ast} \xbar^t_{j_n}
                    \end{pmatrix}, \quad\text{where } \alpha'\defeq \frac{\alpha + (n-1)(1-l^t)}{n}.
          \end{align*}
          Using the assumption that $l^t \in [0,1]$, for all $\alpha \in [0,1]$ we have $\alpha' \in [0,1]$. Hence, using \cref{cor:dual norm} we obtain that for all $\alpha \in [0,1]$,
          \begin{align*}
            \dn{\tilde{\vec{\ell}} + \alpha\mat{C}^{t-}_{v,\ast} \xbar^t_v}{\xbar^t_v}^2 &= \sum_{i=1}^n \dn{\tilde{\vec{\ell}}^t_{j_i} + \alpha'\cdot\mat{C}^{t-}_{j_i,\ast} \xbar^t_{j_i}}{\xbar^t_{j_i}}^2
          \end{align*}
          Taking expectations and using the inductive hypothesis, we have
          \begin{align*}
            \bbE_t\mleft[\dn{\tilde{\vec{\ell}} + \alpha\mat{C}^{t-}_{v,\ast} \xbar^t_v}{\xbar^t_v}^2\mright] &= \sum_{i=1}^n \bbE_t\mleft[\dn{\tilde{\vec{\ell}}^t_{j_i} + \alpha'\cdot\mat{C}^{t-}_{j_i,\ast} \xbar^t_{j_i}}{\xbar^t_{j_i}}^2\mright]\\
            &\le \sum_{i=1}^n \mleft(2 N_{j_i}^2 + |\Sigma|^2\mleft(2 - \frac{1}{w_{v}}\mright) N_{j_i}\mright) \\
            &\le  2 \mleft(\sum_{i=1}^n N_{j_i}\mright)^2 + |\Sigma|^2\mleft(2 - \frac{1}{w_{v}}\mright) \sum_{i=1}^n N_{j_i}\\
            &= 2 N_v^2 + |\Sigma|^2\mleft(2 - \frac{1}{w_{v}}\mright) N_v.
          \end{align*}
        \item \textbf{Third case:} $v\in \cJ$ is a generic decision node. Let $\{1, \dots, m\}$ be the terminal decision actions at $v$, and $\{m+1, \dots, n\}$ be the remaining, non-terminal decision actions, leading to observation points $k_{m+1}, \dots, k_n$, respectively. Furthermore, let $\xbar^t_v = (\lambda_1, \dots, \lambda_n, \lambda_{m+1}\xbar^t_{k_{m+1}}, \dots, \lambda_n\xbar^t_{k_n})$ in accordance with \cref{eq:structure decision point}. Action $a \in \{1, \dots, n\}$ is selected with probability $\lambda_i$.

            We break the analysis according to whether the sampled action $a$ is terminal or not.
            \begin{itemize}
                \item If $a$ is terminal, from the definition of $\mat{C}^{t-}_{v,*}$ (\cref{prop:inverse convex hull}), $\vec{b}^t_{v,*}$ (\cref{eq:bj terminal}) and \cref{cor:dual norm} we have
                    \begin{align*}
                        \dn{\tilde{\vec{\ell}} + \alpha\mat{C}^{t-}_{v,\ast} \xbar^t_v}{\xbar^t_{v}}^2 &= \frac{1}{w_v} \sum_{i=1}^m \lambda_i^t \mleft(\alpha + \frac{\mathds{1}[i=a]}{\lambda_a^t}N_v \mright)^2 \\
                            &\hspace{3cm}+ \sum_{i=m+1}^n \mleft( \frac{\lambda^t_i}{w_v}(\alpha (\xbar_{k_i}^t)^\top \mat{C}^{t-}_{k_i} \xbar_{k_i}^t)^2 + \lambda^t_i \dn{\alpha\mat{C}^{t-}_{k_i}\xbar^t_{k_i}}{\xbar^t_{k_i}}^2\mright)\\
                        &= \frac{\alpha^2}{w_v} + \frac{2\alpha}{w_v}N_v + \frac{N_v^2}{w_v \lambda_a^t}+ \alpha^2\sum_{i=m+1}^n \lambda_i^t \dn{\mat{C}^{t-}_{k_i} \xbar^t_{k_i}}{\xbar_{k_i}}^2,
                    \end{align*}
                    where we used \cref{lem:z Cinv ybar equals 1} and the fact that $(\lambda_1^t, \dots, \lambda_n^t) \in \Delta^n$ in the second equality. Using \cref{lem:Cinv xbar le 2}, the fact that $w_v \ge 2$ for all $v$, and the hypothesis $\alpha \in [0,1]$, we can bound the squared dual norm as
                    \[
                        \dn{\tilde{\vec{\ell}} + \alpha\mat{C}^{t-}_{v,\ast} \xbar^t_v}{\xbar^t_{v}}^2 \le 4N_v + \frac{N_v^2}{w_v\lambda_a^t}.
                    \]
                \item If $a$ is terminal, from the definition of $\mat{C}^{t-}_{v,*}$ (\cref{prop:inverse convex hull}), $\vec{b}^t_{v,*}$ (\cref{eq:bj nonterminal}) and \cref{cor:dual norm} we have
                    \begin{align}
                        \dn{\tilde{\vec{\ell}} + \alpha\mat{C}^{t-}_{v,\ast} \xbar^t_v}{\xbar^t_{v}}^2 &= \frac{1}{w_v} \mleft(\sum_{i=1}^m \lambda_i^t \alpha^2\mright) + \frac{\lambda^t_a}{w_v}\mleft(\frac{1}{\lambda_a^t}(N_v - N_{\rho(v,a)}) +  \frac{1}{\lambda_a^t}(\xbar_{k_i}^t)^\top \mleft(\tilde{\vec{\ell}}^t_{k_a} + (\alpha\lambda_a^t)\mat{C}^{t-}_{k_a} \xbar_{k_i}^t\mright)\mright)^2 \nonumber\\
                            &\hspace{1cm}+ \frac{1}{\lambda^t_a } \dnlarge{\tilde{\vec{\ell}}^t_{k_a} + (\alpha\lambda_a^t)\mat{C}^{t-}_{k_a} \xbar_{k_a}^t}{\xbar^t_{k_a}}^2 \nonumber\\
                            &\hspace{1cm}+ \sum_{\substack{i \in \{m+1,\dots, n\}\\i\neq a}} \mleft( \frac{\lambda^t_i}{w_v}(\alpha (\xbar_{k_i}^t)^\top \mat{C}^{t-}_{k_i} \xbar_{k_i}^t)^2 + \lambda^t_i \dn{\mat{C}^{t-}_{k_i}\xbar^t_{k_i}}{\xbar^t_{k_i}}^2\mright)\label{eq:ugly}
                    \end{align}
                    Now,
                    \[
                        (\tilde{\vec{\ell}}^t_{k_a})^\top \xbar^t_{k_a} = l^t\cdot \vec{y}_{k_a}^t \mat{C}^{t-}_{k_a} \xbar^t_{k_a} + (\vec{b}^t_{k_a})^\top \xbar^t_{k_a} = l^t + N_{k_a} - 1 \le N_{k_a},
                    \]
                    where we used \cref{lem:bt} together with the hypothesis that $l^t \in [0,1]$. So, using the hypothesis $\alpha \in [0,1]$,
                    \begin{align*}
                        \frac{\lambda^t_a}{w_v}\mleft(\frac{1}{\lambda_a^t}(N_v - N_{\rho(v,a)}) +  \frac{1}{\lambda_a^t}(\xbar_{k_i}^t)^\top \mleft(\tilde{\vec{\ell}}^t_{k_a} + (\alpha\lambda_a^t)\mat{C}^{t-}_{k_a} \xbar_{k_i}^t\mright)\mright)^2 &\le \frac{1}{\lambda_a^t w_v}(N_v + \alpha\lambda_a^t)^2\\
                        &\le \frac{N_v^2}{w_v \lambda_a^t} + 2N_v.
                    \end{align*}
                    Substituting that bound into~\eqref{eq:ugly}, and bounding the remaining terms as in the previous case (that is, terminal sampled action $a$), we have
                    \begin{align*}
                        \dn{\tilde{\vec{\ell}} + \alpha\mat{C}^{t-}_{v,\ast} \xbar^t_v}{\xbar^t_{v}}^2 & \le 4N_v + \frac{N_v^2}{w_v \lambda_a^t} + \frac{1}{\lambda_a^t}\dnlarge{\tilde{\vec{\ell}}^t_{k_a} + (\alpha\lambda_a^t)\mat{C}^{t-}_{k_a} \xbar_{k_a}^t}{\xbar^t_{k_a}}^2
                    \end{align*}
            \end{itemize}
            We now take expectations:
            \begin{align*}
                \bbE_t\mleft[\dn{\tilde{\vec{\ell}} + \alpha\mat{C}^{t-}_{v,\ast} \xbar^t_v}{\xbar^t_{v}}^2\mright] &\le \sum_{a=1}^m \lambda_a^t \mleft(4N_v + \frac{N_v^2}{w_v\lambda_a^t}\mright) \\
&\hspace{1cm}+ \sum_{a=m+1}^n \lambda_a^t \mleft(4N_v + \frac{N_v^2}{w_v\lambda_a^t} + \frac{1}{\lambda_a^t}\bbE_t\mleft[\dnlarge{\tilde{\vec{\ell}}^t_{k_a} + (\alpha\lambda_a^t)\mat{C}^{t-}_{k_a} \xbar_{k_a}^t}{\xbar^t_{k_a}}^2\mright]\mright)\\
                &\le 4N_v + \frac{N_v^3}{w_v} + \sum_{i=m+1}^n \bbE_t\mleft[\dnlarge{\tilde{\vec{\ell}}^t_{k_a} + (\alpha\lambda_a^t)\mat{C}^{t-}_{k_a} \xbar_{k_a}^t}{\xbar^t_{k_a}}^2\mright]\\
                &\le 4N_v + |\Sigma|^2\frac{N_v}{w_v} + \sum_{i=m+1}^n \bbE_t\mleft[\dnlarge{\tilde{\vec{\ell}}^t_{k_a} + (\alpha\lambda_a^t)\mat{C}^{t-}_{k_a} \xbar_{k_a}^t}{\xbar^t_{k_a}}^2\mright].
            \end{align*}
            Since $\alpha \lambda_a^t \in [0,1]$, we can use the inductive hypothesis and write
            \begin{align*}
                \bbE_t\mleft[\dn{\tilde{\vec{\ell}} + \alpha\mat{C}^{t-}_{v,\ast} \xbar^t_v}{\xbar^t_{v}}^2\mright] &\le 4N_v + |\Sigma|^2 \frac{N_v}{w_v} + \sum_{i=m+1}^n \mleft(2-\frac{1}{w_{k_i}}\mright)|\Sigma|^2 \frac{N_{k_i}}{w_{k_i}} + 2 N_{k_i}^2\\
                &\le 2\mleft(N_{k_i}^2 + 2\sum_{i=m+1}^n N_{k_i}\mright) + \frac{|\Sigma|^2}{w_v}+ |\Sigma|^2\mleft(\sum_{i=m+1}^n  \mleft(2 - \frac{1}{w_{k_i}} + \frac{1}{w_v}\mright) N_{k_i}\mright),
            \end{align*}
            where the second inequality follows from using $N_{v} = 1+ \sum_{i=m+1}^n N_{k_i}$. Using the fact that $1/w_i \ge 2/w_v$, we further obtain
            \begin{align*}
                \bbE_t\mleft[\dn{\tilde{\vec{\ell}} + \alpha\mat{C}^{t-}_{v,\ast} \xbar^t_v}{\xbar^t_{v}}^2\mright] &\le 2\mleft(N_{k_i}^2 + 2\sum_{i=m+1}^n N_{k_i}\mright) + |\Sigma|^2\mleft(\sum_{i=m+1}^n  \mleft(2 - \frac{1}{w_{v}}\mright) N_{k_i}\mright)\\
                &= 2\mleft(N_{k_i}^2 + 2\sum_{i=1}^n N_{k_i}\mright) + |\Sigma|^2\mleft(2 - \frac{1}{w_{v}}\mright) (N_v - 1) + \frac{|\Sigma|^2}{w_v}.
            \end{align*}
            Since $\frac{1}{w_v} \le 1 \le 2 - \frac{1}{w_v}$, we can further bound the right hand side as
\begin{align*}
                \bbE_t\mleft[\dn{\tilde{\vec{\ell}} + \alpha\mat{C}^{t-}_{v,\ast} \xbar^t_v}{\xbar^t_{v}}^2\mright] &\le 2\mleft(N_{k_i}^2 + 2\sum_{i=m+1}^n N_{k_i}\mright) + |\Sigma|^2\mleft(2 - \frac{1}{w_{v}}\mright) N_v\\
                &\le 2\mleft(\sum_{i=m+1} N_{k_i}\mright)^2 + |\Sigma|^2\mleft(2 - \frac{1}{w_{v}}\mright) N_v\\
                &\le 2 N_v^2 + |\Sigma|^2\mleft(2 - \frac{1}{w_{v}}\mright) N_v
            \end{align*}
            thus completing the inductive step.
    \end{itemize}
    Finally, noting that $N_v \le |\Sigma|$ at all $v$ yields the statement.
\end{proof}

    \section{Game Instances Used in our Experimental Evaluation}\label{app:games}

\emph{Matrix game} is a small matrix game, where the payoff matrix for Player 1 is
\[
    \begin{pmatrix}
        -1   &    1\\
         1   & -0.5\\
       0.9   &   -1
    \end{pmatrix}.
\]
The first and third action of Player~1 have almost opposite payoffs, which empirically complicates the loss estimate and convergence to the best response with bandit feedback.

\emph{Kuhn poker} is a standard benchmark game in the equilibrium-solving community \citep{Kuhn50:Simplified}. In Kuhn poker, the two players put an ante worth $1$ into the pot at the beginning of the game. Then, each player is privately dealt one card from a deck that contains only three cards---specifically, jack, queen, and king. A single round of betting then occurs, with the following rule: first, Player $1$ decides to either check or bet $1$; then,
  \begin{itemize}[nolistsep,nolistsep,itemsep=1mm,leftmargin=5mm]
  \item If Player 1 checks, Player 2 may check or raise $1$.
    \begin{itemize}[nolistsep]
      \item If Player 2 checks, a showdown occurs; otherwise (that is, Player 2 raises), Player 1 can fold or call.
        \begin{itemize}
          \item If Player 1 folds the game ends and Player 2 takes the pot; if Player 1 calls a showdown occurs.
          \end{itemize}
        \end{itemize}
      \item If Player 1 raises Player 2 may fold or call.
        \begin{itemize}[nolistsep]
        \item If Player 2 folds the game ends and Player 1 takes the pot; if Player 2 calls, a showdown occurs.
        \end{itemize}
      \end{itemize}
      If a showdown occurs, the player with the higher card wins the pot and the game ends.

\emph{Leduc poker} is another common benchmark game in the equilibrium-finding
community~\cite{Southey05:Bayes}. It is played with a deck of $3$ unique
ranks, each of which appears exactly twice in the deck. There are two rounds in the game. In the
first round, all players put an ante of $1$ in the pot and are privately dealt a single card. A round of betting then starts. Player 1 acts first, and at most two bets are allowed per player. Then, a card is publicly revealed, and another
round of betting takes place, with the same dynamics described above. After the two betting round, if one of the players has a pair with the public card, that player
wins the pot. Otherwise, the player with the higher card wins the pot. All bets in the first
round are worth $2$, while all bets in the second round are $4$.

\emph{Goofspiel} is another popular parametric benchmark game, originally proposed by \citet{Ross71:Goofspiel}. It is a two-player card game, played with three identical decks
of $k$ cards each, whose values range from $1$ to $k$. In our experiment, we used $k=4$. At the beginning of the game, each player gets dealt a full deck as their hand, and the third deck (the ``prize'' deck) is shuffled and put face down on the board. In each turn, the topmost card from the prize deck is revealed. Then, each player privately picks a card from their hand. This card acts as a bid to win the card that was just revealed from the prize deck. The selected cards are simultaneously revealed, and the highest one
wins the prize card. If the players' played cards are equal, the prize card is split.  The players’ score are computed as the sum of the values of the prize cards they have won.  
    \section{Hyperparameter Selection for Experimental Evaluation}

The AHR algorithm as well as our proposed bandit regret minimizer require us to choose one step-size parameter. While we could have simply used the theoretically correct step-size which ensures $\Tilde{O}(\sqrt{T})$ expected regret we experimented by multiplying the step size with a constant $\alpha$. Multiplying the theoretically correct step-size with a constant does not effect the asymptotic regret bound. We considered $\alpha \in \{0.5, 1, 2, 5, 10\}$. For each choice of $\alpha$ we ran both algorithm 10 times on each of the two larger games (Goofspiel and Leduc poker). \cref{fig:hyperparameters_ahr,fig:hyperparameters_bco} show the performance of the algorithm by \citet{Abernethy08:Competing} (AHR) and our algorithm for every choice of $\alpha$.

\newcommand\scaletwopics{0.49}
\begin{figure}[H]
    \centering
        \includegraphics[width=\scaletwopics\textwidth]{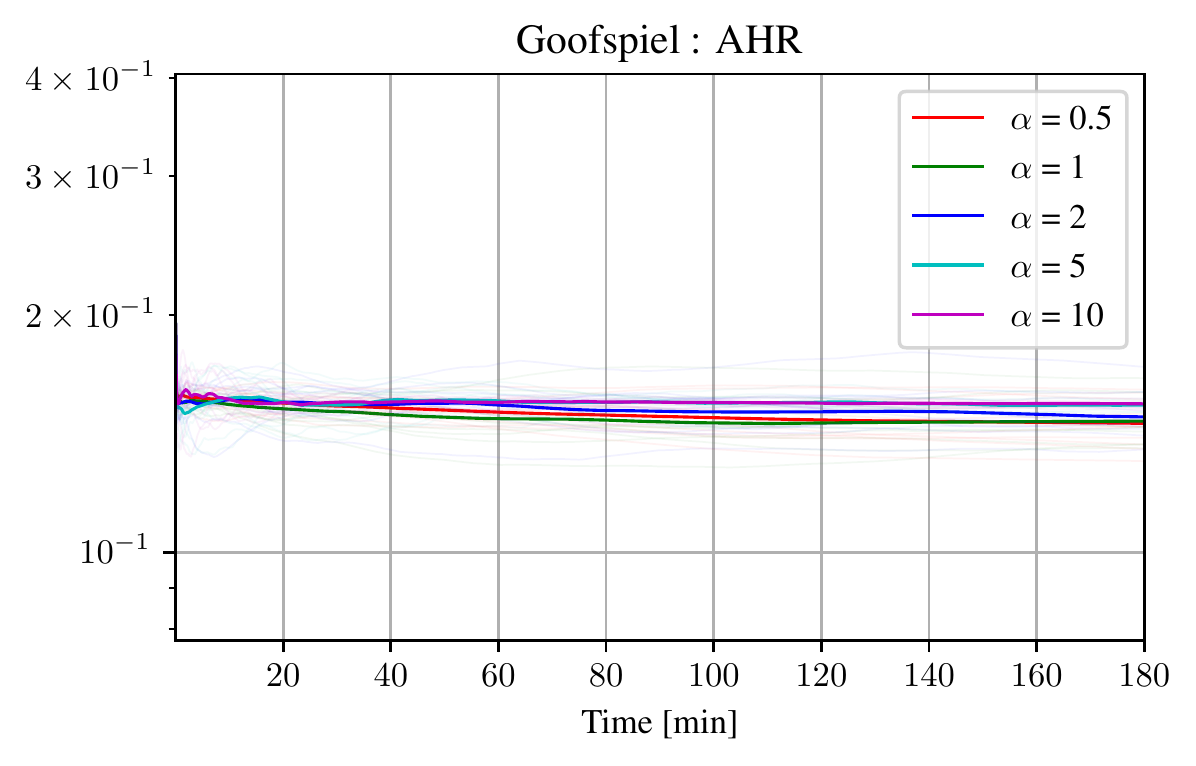}
        \includegraphics[width=\scaletwopics\textwidth]{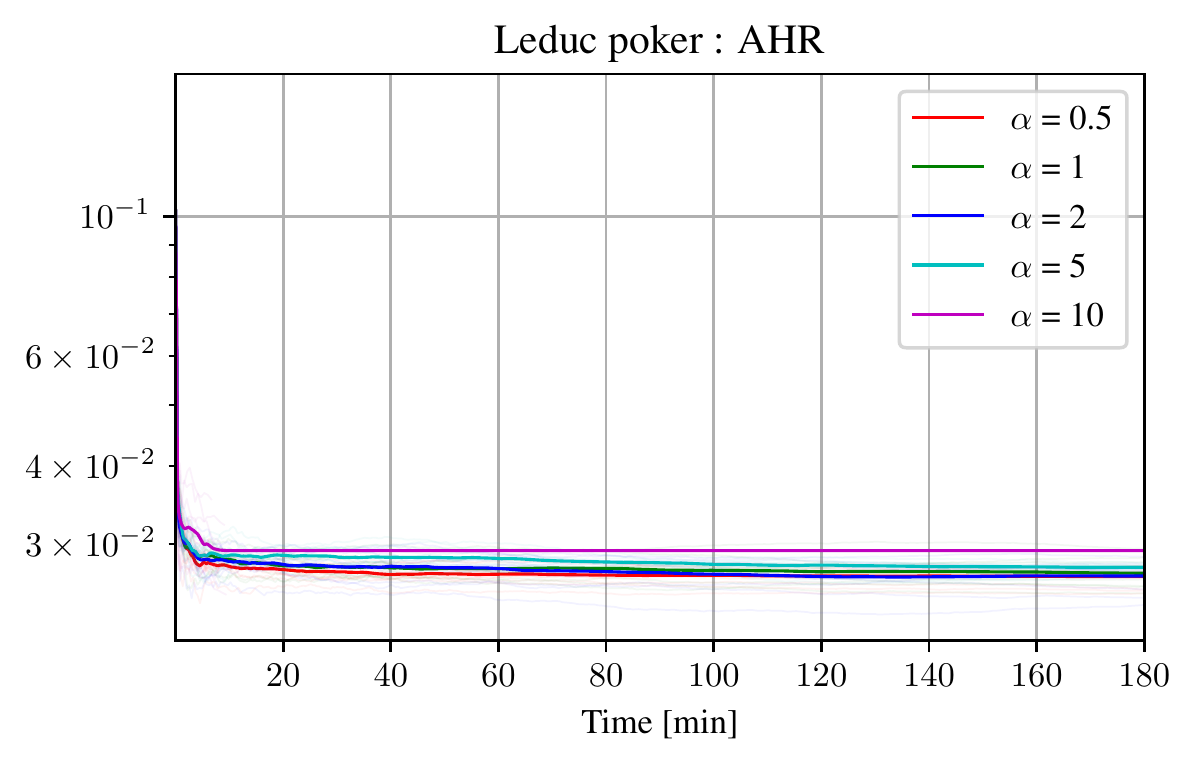}
    \caption{Average regret over 10 runs of the AHR algorithm using different step-size multipliers $\alpha$.}
    \label{fig:hyperparameters_ahr}
\end{figure}

\begin{figure}[H]
    \centering
        \includegraphics[width=\scaletwopics\textwidth]{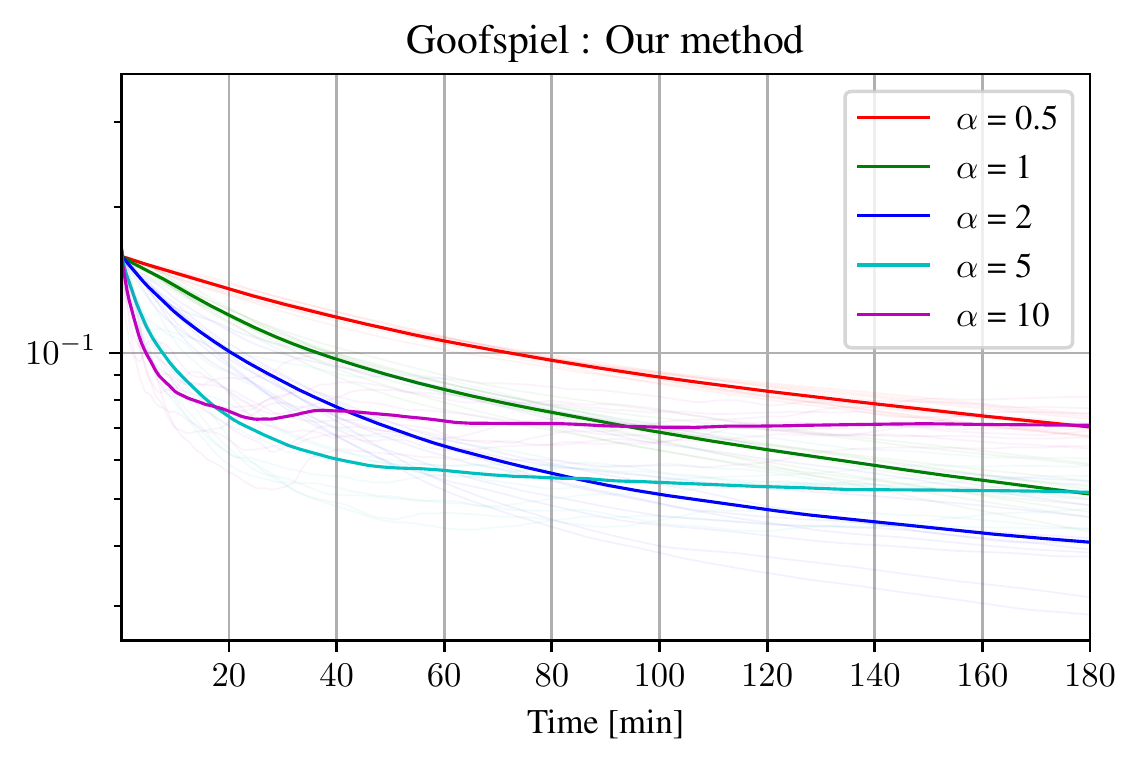}
        \includegraphics[width=\scaletwopics\textwidth]{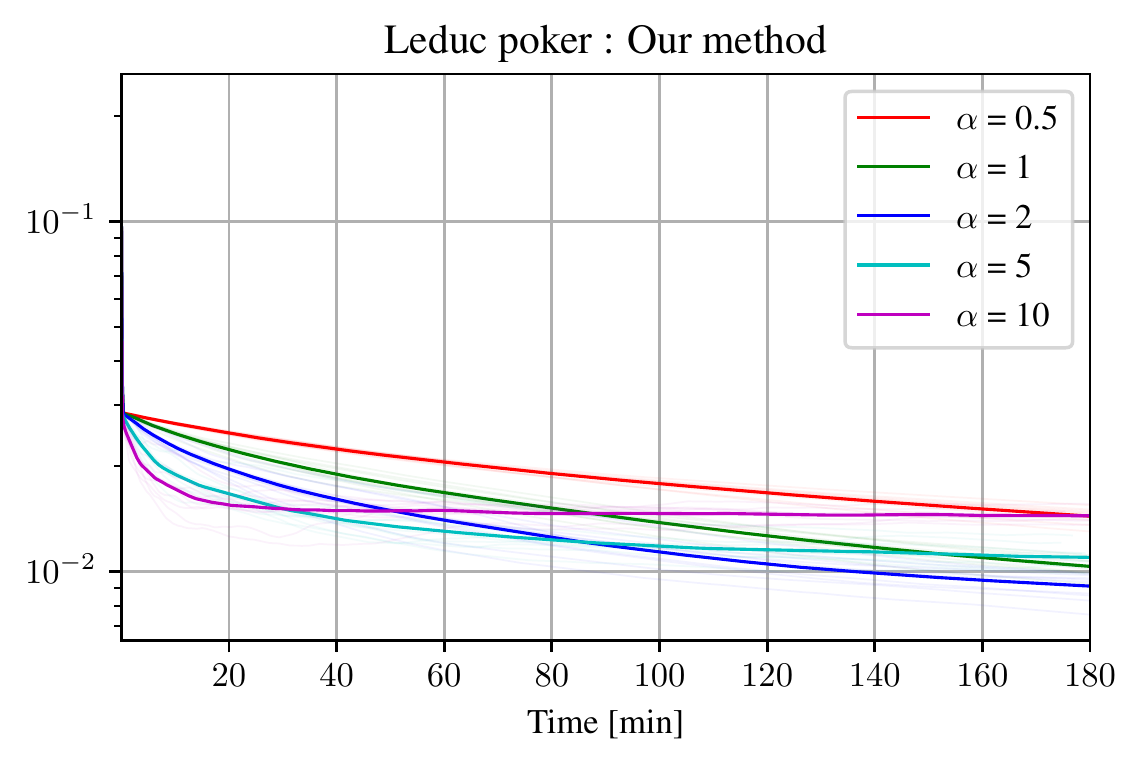}
        \vspace{0.5cm}
    \caption{Average regret over 10 runs of our algorithm using different step-size multipliers $\alpha$.}
    \label{fig:hyperparameters_bco}
\end{figure}

In AHR, for all choices of $\alpha$ we achieve very similar performances. For our method, the best choices of $\alpha$ seem to be $2$ and $5$.

\fi


\end{document}